\documentclass{article}
\usepackage[utf8]{inputenc}
\usepackage[dvipsnames]{xcolor}
\usepackage{amsfonts}
\usepackage{amsthm}
\usepackage{amsmath}
\usepackage{amssymb}
\usepackage{fullpage}
\usepackage{graphicx}
\usepackage{algorithm}
\usepackage{algpseudocode}
\usepackage{tikz-cd}
\usepackage{authblk}
\usepackage{enumitem}
\usepackage{soul}
\setlist{nosep}

\usepackage{caption}
\usepackage{subcaption}

\title{Brain Chains as Topological Signatures for Alzheimer's Disease}

\author[1,2]{Christian Goodbrake}
\author[1]{David Beers}
\author[1,4]{Travis B. Thompson}

\author[1,3]{Heather A. Harrington}
\author[1]{Alain Goriely}
\affil[1]{Mathematical Institute, University of Oxford}
\affil[2]{Willerson Center for Cardiovascular Modeling and Simulation, Oden Institute, The University of Texas at Austin}
\affil[3]{Wellcome Centre for Human Genetics,  University of Oxford}
\affil[4]{Department of Mathematics and Statistics, Texas Tech University}
\date{}                 

\theoremstyle{definition}
\newtheorem{exmp}{Example}[section]
\newtheorem{defin}{Definition}[section]
\newtheorem{remark}{Remark}[section]
\newtheorem{theorem}{Theorem}
\newtheorem{prop}{Proposition}
\newtheorem{cor}{Corollary}

\graphicspath{{./figures/}}

\newcommand{\pol}{\mathrm{Pol}}
\newcommand{\pa}{\mathrm{Par}}

\begin{document}

\maketitle

\begin{abstract}
    \begin{itemize}
    Topology is providing new insights for neuroscience.  For instance, graphs, simplicial complexes, directed graphs, flag complexes, persistent homology and convex covers have been used to study functional brain networks, synaptic connectivity, and hippocampal place cell codes. 
    We propose a topological framework to study the evolution of Alzheimer's disease, the most common neurodegenerative disease. The modeling of this disease starts with the representation of the brain connectivity as a graph and the seeding of a toxic protein in a specific region represented by a vertex. Over time, the accumulation of toxic proteins at vertices and their propagation along  edges are modeled by a dynamical system on this graph. These dynamics provide an order on the edges of the graph  according to the damage created by high concentrations of proteins. This sequence of edges defines a filtration of the graph. We consider different filtrations given by different disease seeding locations. To study this filtration we propose a new combinatorial and topological method. A filtration defines a maximal chain in the partially ordered set of spanning subgraphs ordered by inclusion. To identify similar graphs, and define a topological signature, we quotient this poset by graph homotopy equivalence, which gives maximal chains in a smaller poset. We provide an algorithm to compute this direct quotient without computing all subgraphs and then propose bounds on the total number of graphs  up to homotopy equivalence. To compare the maximal chains generated by this method, we extend Kendall's $d_K$ metric for permutations to more general graded posets and establish bounds for this metric. We then demonstrate the utility of this framework on actual brain graphs by studying the dynamics of tau proteins on the structural connectome. {We show that the proposed topological brain chain equivalence classes distinguish different simulated subtypes of Alzheimer's disease.}
\end{itemize}

\end{abstract}

\section{Introduction}
Brain networks are an important research topic for  understanding  both healthy brain function and pathology \cite{bullmore2009}. In particular, network theory has been used to model the evolution of neurodegenerative diseases, a family of conditions characterised by progressive dysfunction of neurons until they ultimately die. The central idea of the 'prion-like' hypothesis is that some neurodegenerative diseases, such as Alzhemier's or Parkinson's are governed by the accumulation and propagation of toxic proteins on the \textit{connectome}, the brain structural network \cite{jucker2013}.

Indeed, it is well appreciated that the most common neurodegenerative disease, Alzheimer's disease (AD),  is characterized by the propagation of misfolded tau proteins over the connectome. 
The disease presentation differs across patients subpopulations, which has led to so-called \textit{AD subtypes}.  For example, in the clinic, the \textit{ limbic subtype} is more strongly associated to memory decline whereas the \textit{middle temporal lobe sparing (MTL) subtype} is associated with changes in language proficiency \cite{vogel2021}. Vogel et al. recently postulated that disease subtypes are characterized by misfolded protein at different starting locations in the brain called \textit{seeding regions} or \textit{epicenters}  \cite{vogel2021}. Over many decades, toxic proteins will propagate from an epicenter to the rest of the network and different epicenter locations  can give rise to different toxic protein patterns.  An open problem is early detection of a patient's AD subtype, which has implications to clinical care and treatment.  The current gold standard for determining a patient's AD subtype is the postmorterm study of the distribution of tau neurofibrillary tangles within the brain.  An open question is how to go about developing effective, non-invasive techniques to discern the progression of AD and a patient-specific subtype from that progression.  This work proposes a topological approach for studying AD subtypes by combining a mechanistic model of AD progression, a recent hypothesis on AD subtypes, and the examination of filtration images under a topological equivalence relation.

Studying the brain with algebraic topology was first proposed in 1962 by Zeeman \cite{zeeman1962topology}. Since then, neurotopology has flourished into a well-established field. For example, grid cells or place cells can be studied with neural codes \cite{curto2017can} and single neuron shape can be analysed as trees with the topological morphology descriptor  \cite{kanari2018topological,beers2022barcodes}. Brain networks ranging from vascular, functional to structural have benefited from topological data analysis (TDA) \cite{reimann2017cliques,sizemore2019importance,expert2019topological,bendich2016persistent}. 
The most prominent algorithm in topological data analysis, persistent homology, takes in a filtration, a nested sequence of spaces built on data, and outputs a persistence module. Persistent homology has been applied to functional networks (from functional magnetic resonance imaging data) to study schizophrenia, epileptic seizures, and Alzheimer's disease \cite{caputi2021promises,xing2022spatiotemporal,stolz2021topological}. Previously, persistent homology distinguished different geometric patterns that arise as contagions spread along a network \cite{taylor2015topological}. In that study and here, we consider different ``seeding" sites for the start of a contagion or toxic protein spreading along a specified network; however, the topological approach we use here is finer than persistent homology.

Here, we start with a filtration, corresponding to toxic protein propagation or progression on a brain graph from a seeding site.
The fundamental principle guiding this approach is to study aspects of progressions defined on a (fixed) graph by considering the spanning subgraphs, subgraphs containing a vertex set $V$, arising from filtrations generated by that progression. 
Since we are interested in comparing topological characteristics of one or more network progressions, given by different seeding sites, we will consider the set of spanning subgraphs modulo a topologically significant equivalence relation. We would like to distinguish the number of loops in each connected component to discern disease subtypes; therefore, we will require a finer invariant than homology.

The topological analysis of graphs via their geometric representations is a well-established method \cite{archdeacon1996topological, hatcher2002algebraic}, with multiple  interesting problems \cite{lapaugh1980subgraph, lingas2009exact}.  
The most prominent example of translating a topological equivalence relationship to a graph theoretic equivalence relation is the graph homeomorphism, where two graphs are considered equivalent if their geometric realisations are homeomorphic. Here, we apply a similar approach to  study the filtrations associated  with different protein propagations; however, we choose a looser topological equivalence relation, namely that of \textit{homotopy equivalence} \cite{hatcher2002algebraic}.
We establish a unique representation for the induced equivalence classes, and examine the relationships between these classes induced by the spanning subgraphs partially ordered by inclusion.

In standard topology, homotopies are continuous curves whose points are continuous functions. Over the past couple decades, specialised combinatorial analogues of homotopy theories have been developed \cite{barcelo2001foundations, barcelo2005perspectives,grigor2014homotopy}. These combinatorial homotopies are path graphs whose vertices are graph homomorphisms. 
While the homotopy equivalence for graphs we utilise is induced by the standard topological definition applied to the graphs' geometric realisations, we introduce a notion of discrete homotopy for maximal chains in ordered posets-- the homotopy poset.
Further, the discrete nature of this homotopy allows us to define a discrete metric between these chains, which gives us a rigorous basis by which different neurodegenerative progressions can be compared.
This discrete homotopy for chains in posets can likely be formulated in terms of the aforementioned combinatorial homotopies applied to maximal chains in graded posets.

We define a metric between maximal chains, which is bounded from below by a metric defined on top dimensional simplices of the order complex. While we do not construct the order complex associated with the chains appearing in our work, earlier work has explored the topological structure inherent in the combinatorial objects appearing in order theory \cite{wachs2006poset, matouvsek2003using, kozlov2005topologicalcombinatorics}.
 The metric we define distinguishes quotiented filtrations based on the connectivity of the affected regions of the brain graph at each stage in disease progression. Thus, the proposed view yields a quantitative metric on the set of calculated brain chains, with some chains representing protein progressions that propagate the entire graph quickly whereas some brain chains have topological signatures that form local cycles close to the seeding site before spreading to the full brain graph.

\subsection{Organization}
This work sits at the intersection of algebraic topology, combinatorics, dynamical systems, graph theory, number theory, network science, neuroscience and order theory. 
Throughout, we provide necessary definitions and direct readers to fuller treatment elsewhere.
 In Section~\ref{sec:graph-filtrations}, we introduce the edge filtration of a graph, the connection between these filtrations, and the partially ordered set of spanning subgraphs $SS\left(G\right)$. We introduce the graph homotopy polynomial, and the homotopy poset, the quotient of $SS\left(G\right)$ by graph homotopy equivalence. 
This polynomial uniquely encodes the graph's homotopy equivalence class.
Therefore, this polynomial can be used to compute any topological quantity that is invariant under homotopy equivalence. In Section~\ref{sec:pe-quotient-construction}, we present an algorithm using these polynomials to directly compute the homotopy poset (without having to construct the entire set of spanning subgraphs). 
In Section~\ref{sec:complexity}, we establish complexity bounds on the number of elements that are in the poset as the number of vertices grows to provide a qualitative view of how the previous algorithm scales, and discuss connections to number theory.
In Section~\ref{sec:chain-compare} we generalise Kendall's $d_K$ metric on permutations to graded posets, and use this to define the discrete homotopy metric on graded poset chains.
 In Section~\ref{sec:meshes} we introduce left and right-covering conditions, sufficient structures that guarantee finite discrete homotopy distances between arbitrary maximal chains.
We establish upper and lower bounds on this metric, and in  Section~\ref{sec:application}, we apply these ideas to the problem of neurodegenerative diseases to distinguish and topologically describe subtypes of Alzheimer's disease.

\section{Graph Edge Filtrations}\label{sec:graph-filtrations}
We consider a simple graph $G$ with vertex set $V = \{v_i\}$ with $|V|=N$ vertices and edge set $E = \{e_{ij} = \{v_i,v_j\}\}$ with $|E|=M$ edges.
For the purposes of this section and the following section, we will consider $G$ to be the complete graph on $V$.
A similar analysis can be applied to non-complete graphs, but the direct construction presented in Section \ref{sec:pe-quotient-construction} only applies to complete graphs.
We examine the poset of spanning subgraphs of $G$, $SS\left(G\right)$, i.e. those subgraphs $S \subseteq G$ containing all of the vertices in $V$, partially ordered by subgraph inclusion.
This poset is isomorphic to the power set of $E$, $\mathcal{P}\left(E\right)$.
This spanning subgraph poset is graded by Euler characteristic $L = N-|E(S)|$, or equivalently, since $V$ is fixed, graded by the number of edges, $|E(S)|$, where we have denoted the edge set of a subgraph $S$ as $E(S)$.
We denote this strict partial ordering relation with $\prec$, with $\lessdot$ denoting the covering relation, i.e. $a$ ``is covered by" $b$, $a\lessdot b$ if and only if $a \prec b$, and there does not exist $c$ such that $a \prec c \prec b$, which in our case, $a \lessdot b$ is equivalent to $a \subset b$ and $|E(b)\setminus E(a)|=1$.

We seek to understand and quantitatively compare the possible routes of propagation through $G$, thus we define \emph{edge filtrations}.
\begin{defin}[Edge Filtration]\label{defin:edge-filtration}
    An \emph{edge filtration} $\mathcal{F}=S_0\subset S_1\subset ... \subset S_M$ of $G$ is a sequence of spanning subgraphs of $G$ such that $E(S_0) = \{\}$, $S_M=G$, and $S_{i-1}\lessdot S_i,\,\forall i\in\{1,...,M\}$.
\end{defin}
Equivalently, edge filtrations are maximal chains in $SS\left(G\right)$, i.e. they are totally ordered subsets of $SS\left(G\right)$ such that no element can be added without $\mathcal{F}$ ceasing to be totally ordered.
Any edge filtration can be represented by a permutation $\Xi$ of the edge set, since the edge sets of the subgraphs in this filtration are partial unions of this permutation.
Given $\mathcal{F}$, $\Xi\left(\mathcal{F}\right)$ can be recovered by taking successive set differences;
\begin{equation}
\label{eqn:edge_permutation}
    \Xi\left(\mathcal{F}\right) = \left(\begin{array}{cccc}1 & 2 & \dots & M \\ e_1 & e_2 & \dots & e_M \end{array}\right)
        \text{ where } e_i = S_i\setminus S_{i-1}.
\end{equation}
For convenience, we use the notation 
    $\Xi = \{e_1,e_2,\dots,e_M\}$,
in lieu of \eqref{eqn:edge_permutation}, to represent an edge permutation.
We seek to rigorously compare the similarity of two edge filtrations $\mathcal{F}_1$, and $\mathcal{F}_2$, which we can do indirectly through their associated edge set permutations.
These permutations can be compared through Kendall's $d_K$ metric \cite{Kendall_1983}:
\begin{defin}[Kendall's $d_K$]
Let $\Xi_1, \Xi_2$ be permutations of $E$, and let $\sigma$ be the element of the permutation group $S_{M}$ satisfying 
\begin{equation}
    \Xi_2 = \sigma \left(\Xi_1\right).
\end{equation}
Kendall's $d_K\left(\Xi_1,\Xi_2\right)$ is defined to be the length of the minimal representation of $\sigma$ in terms of adjacent transpositions.
\end{defin}
The distance $d_K\left(\Xi_1,\Xi_2\right)$ is equivalently given as the number of operations to transform $\Xi_1$ into $\Xi_2$ using bubble sort, and is equal to the number of discordant pairs in $\Xi_1$ and $\Xi_2$.
Kendall's $d_K$ satisfies all the axioms of a metric, hence we can use it as a means of comparing the similarity of two permutations of a set.
We therefore extend the definition of $d_K$ to compare edge filtrations:
\begin{equation}
d_K\left(\mathcal{F}_1,\mathcal{F}_2\right) = d_K\left(\Xi\left(\mathcal{F}_1\right),\Xi\left(\mathcal{F}_2\right)\right).
\end{equation}

\begin{exmp}[$K_3$]
Consider the complete graph on $3$ vertices, $K_3$. Denoting the three edges of this graph as $\{e_1,e_2,e_3\}$, the eight distinct spanning subgraphs of this graph partially ordered by inclusion can be depicted by the following face poset.
\begin{figure}[h]
    \centering
    \includegraphics[width=.5\textwidth]{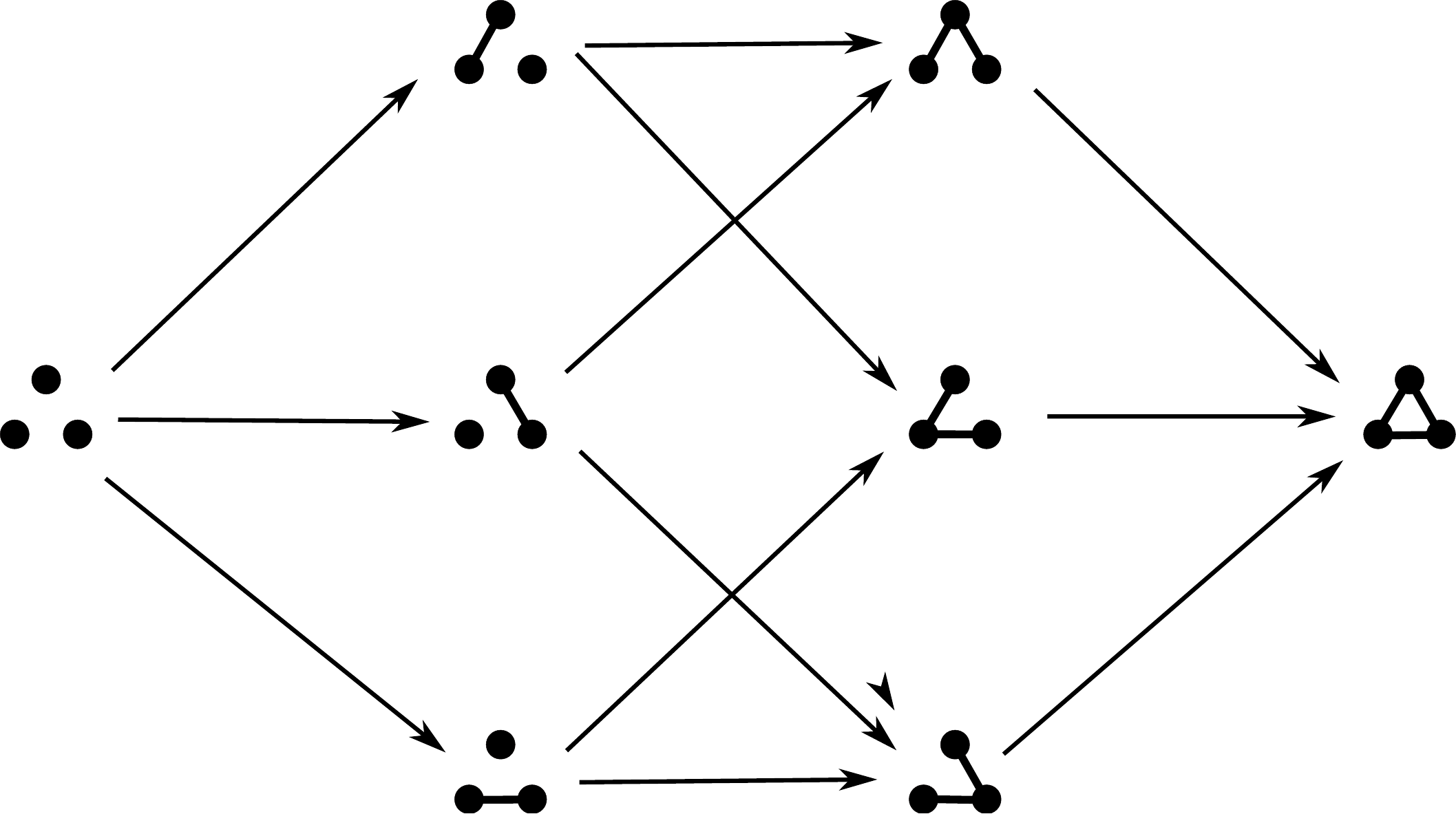}
    \caption{$SS\left(K_3\right)$, the spanning subgraphs of $K_3$ partially ordered by subgraph inclusion.}
    \label{fig:k3}
\end{figure}
\end{exmp}

The size of $SS\left(K_N\right)$ grows quite quickly as $N$, the number of vertices in $V$, increases ($SS\left(K_N\right)| = 2^{\frac{N\left(N-1\right)}{2}}$), which makes constructing the spanning subgraph poset computationally expensive even for modest $N$.
Additionally, we want to focus our attention on the topology of the underlying spanning subgraphs, rather than merely considering their edge sets as combinatorial subsets.
We therefore seek a quotienting of $SS\left(K_N\right)$ under some equivalence relation $\sim$ that preserves the graded structure, and identifies spanning subgraphs that are topologically equivalent in some sense.
We further want this quotiented poset $\mathcal{H}\left(K_N\right)=SS\left(K_N\right)/\sim$ to be directly constructed, rather than computing the entire $SS\left(K_N\right)$, and then further compute all of the equivalence relations. 

We will identify two spanning subgraphs if their geometric realisations are homotopy equivalent, and examine our edge filtrations under this equivalence relation.
\begin{defin}[homotopy equivalence]
Two topological spaces $X$, and $Y$ are \emph{homotopy equivalent} if there exists continuous maps $f:X\rightarrow Y$ and $g:Y\rightarrow X$ such that the compositions $g\circ f$ and $f \circ g$ are homotopic to the identity maps on $X$ and $Y$ respectively, i.e. there exists two continuous one-parameter family of maps $\eta_t :X\rightarrow X$ and $\mu_t:Y\rightarrow Y$ for $t\in\left[0,1\right]$, such that $\eta_0 = \operatorname{id}_X$, $\eta_1 = g \circ f$, $\mu_0 = \operatorname{id}_Y$, and $\mu_1 = f \circ g$.
\end{defin}
This quotienting is analogous to persistent homology, where a sequence of homology groups is obtained from a filtration. The underlying objects we are quotienting are subgraphs; therefore, reducing to homology groups destroys too much information since graphs only have two nontrivial homology groups, $H_0\left(G\right)\cong\mathbb{Z}^{Q}$ where $Q$ is the number of connected components, and $H_1\left(G\right)\cong\mathbb{Z}^{M-N+Q}$ counting the number of loops.
We use homotopy equivalence instead, as this preserves the number of loops in each connected component, rather than merely the number of loops in total.

In the general case determining whether two spaces are homotopy equivalent is  difficult.
Since the spaces we study are graphs, they have a simple enough structure for us to determine homotopy equivalence by explicitly constructing the homotopy equivalence class from a representative member.
Every connected simple graph is homotopy equivalent to a wedge sum of circles \cite{hatcher2002algebraic}, hence every general (possibly disconnected) simple graph is homotopy equivalent to the disjoint union of a wedge sum of circles.

We record this data through a \emph{graph homotopy polynomial} $h_G\left(x\right)$.

\begin{defin}[Graph Homotopy Polynomial]\label{defn:homotopy-poly}
Let $G$ be a graph. 
For each connected component of $G$, associate the term $x^k$ where $k$ is the number of loops in this component.
The graph homotopy polynomial $h_G\left(x\right)$ is the sum of these terms.
\end{defin}

Because the graph homotopy polynomial $h_G\left(x\right)$ is uniquely determined by the homotopy equivalence class of $G$, to determine if two graphs are homotopy equivalent, one must simply compute the graph homotopy polynomial for each and see if these polynomials match.
Further, because the homotopy polynomial encodes a graph's homotopy equivalence classes, any topological quantity of a graph $G$ that is preserved under homotopy equivalence can be computed from $h_G\left(x\right)$.

\begin{exmp}[Topological computations from $h_G$]\label{exmp:topo-comput}
\hfill
\begin{itemize}
\itemsep0em 
\item $h_G\left(1\right)$ is the number of connected components of $G$, the $0$-th Betti number. Therefore $H_0\left(G\right) \cong \mathbb{Z}^{h_G\left(1)\right)}$
\item The number of loops present in $G$ is $\frac{dh_G}{dx}\left(1\right)$, hence the Euler characteristic of $G$ is $L = N-|E|= h_G\left(1\right)-\frac{dh_G}{dx}\left(1\right)$.
\item $H_1\left(G\right) \cong \mathbb{Z}^{\frac{dh_G}{dx}\left(1\right)}$
\end{itemize}
\end{exmp}

Therefore, given an edge filtration $\mathcal{F}$ of $K_N$, we can compute a sequence of graph homotopy polynomials. 
Because the original filtration was a maximal chain in $SS\left(K_N\right)$ graded by the number of edges, or equivalently Euler characteristic, and homotopy equivalence preserves the Euler characteristic, the resulting sequence of graph homotopy polynomials remains a maximal chain in $\mathcal{H}\left(N\right)$. 
We only have to compute $\mathcal{H}\left(N\right)$.

\section{Direct Construction of the Homotopy Poset}\label{sec:pe-quotient-construction}
We seek to directly construct $\mathcal{H}\left(N\right)$, the spanning subgraph poset of the complete graph on $N$ vertices quotiented by graph homotopy equivalence. 
We do this utilising the graph homotopy polynomials defined in the previous section, exploiting the graded structure of the quotiented poset by Euler characteristic.
We note that the unique lowest spanning subgraph in this poset is the subgraph $S_0$ containing no edges, which of necessity is simply the collection of $N$ vertices, and its corresponding graph homotopy polynomial is simply $h_{S_0}\left(x\right) = N$.

Next, we consider the effect that adding an edge to a subgraph has on that subgraph's homotopy polynomial.
Since any subgraph is a disjoint union of connected components, at the graph level, the endpoints of a newly added edge either connect two previously disconnected components, which creates no new loops, and reduces the number of connected components by one, or the endpoints lie in the same previously connected component, increasing the number of loops in that component by one.
To see the effects of this at the level of the graph homotopy polynomial, let us denote the homotopy polynomial of the original subgraph as $h_{S_r}\left(x\right)$, and the homotopy polynomial of the augmented subgraph as $h_{S_{r+1}}\left(x\right)$.
The polynomial $h_{S_r}\left(x\right)$ takes the form
\begin{equation}
    h_{S_r}\left(x\right) = \bigoplus_{k=0}^{\infty}a_k x^k,
\end{equation}
where $a_k$ are non-negative integers, and we have employed a direct sum to indicate that these polynomials can in principle be of arbitrarily high degree, with all but finitely many terms being $0$.
In practice these homotopy polynomials can be implemented as an integer-keyed and integer valued dictionary, where keys are exponents, and values are coefficients.

In the case where the endpoints of the added edge lie in originally disjoint components of the subgraph, originally possessing $i$ and $j$ loops respectively, the number of components in the augmented subgraph possessing $i$ and $j$ loops each decrease by $1$, with the number of components possessing $i+j$ loops increased by $1$. 
The transformation on the homotopy polynomial is thus 
\begin{equation}
    x^i + x^j \rightarrow x^{i+j}
\end{equation}
hence the polynomial transforms as follows
\begin{equation}
\label{eqn:loop_merge}
    h_{S_{r+1}}\left(x\right) = h_{S_r}\left(x\right) - x^i -x^j + x^{i+j}.
\end{equation}
Note that this is true even in the event that $i=j$, or $i=0$, or $j=0$, provided that the incrementing/decrementing is done for each component, i.e. when $i=j$, $a_i\geq 2$, and when $i\neq j$, $a_i, a_j \geq 1$. 

The alternative is that the added edge's endpoints reside in the same originally disjoint component, increasing the number of loops appearing in that component by $1$, say from $i$ to $i+1$.
The transformation is then clearly 
\begin{equation}
    x^i\rightarrow x^{i+1}, 
\end{equation}
and the polynomial transformation is
\begin{equation}
\label{eqn:loop_promotion}
    h_{S_{r+1}}\left(x\right) = h_{S_{r}}\left(x\right) -x^i + x^{i+1}.
\end{equation}
These are the only potential transformations possible, hence any cover of $h_S\left(x\right)$ in $\mathcal{H}\left(N\right)$ must arise from $h_S\left(x\right)$ by one of these transformations.

Because there are only finitely many such transformations, (by virtue of the finite number of nonzero terms in the direct sum), all the possible covers of any element of the poset can be exhaustively considered. 
\begin{algorithm}[h]
\caption{Possible Successors of $h$}
\label{alg:succs}
\begin{algorithmic}
    \Procedure{Succs}{$h$}
    \State $\alpha \gets$ \textbf{dict}$\left(h\right)$ \Comment{Initialise a dictionary with key-value pairs representing exponent-coefficient pairs of $h$}
    \State $S\gets \{\}$ \Comment{Initialise an empty set for successors}
    \For{$y\in $\textbf{ keys}$\left(\alpha\right)$} \Comment{Loop through monomials}
    \State $g \gets \alpha$ \Comment{Initialise $g$ as the dictionary $\alpha$}
    \If{$g\left(y\right)>0$} \Comment{If this monomial's coefficient is nonzero} 
    \State $g\left(y\right) \mathrel{-}= 1$ \Comment{Decrease this coefficient by $1$}
    \For{$z\in$ \textbf{keys}$\left(g\right)$} \Comment{Loop through remaining monomials,
    obtain successors by merging}
        \State $f\gets g$ \Comment{Initialise a dictionary for the original decremented polynomial}
        \If{$f\left(z\right)>0$} \Comment{If this monomial's coefficient is nonzero} 
        \State $f\left(z\right) \mathrel{-}= 1$ \Comment{Decrease this coefficient by $1$}
        \State $f\left(y+z\right) \mathrel{+}= 1$ \Comment{Increment the coefficient for the merged monomial by $1$}
        \State $S \gets$ \textbf{push} $f$ \Comment{Add the successor generated by merging at $y$ and $z$}
        \EndIf
    \EndFor
    \State $g\left(y+1\right) \mathrel{+}= 1$ \Comment{Increment the coefficient corresponding promoting at $y$}
    \State $S \gets$ \textbf{push} $g$ \Comment{Add the successor generated by promoting at $y$}   
    \EndIf
    \EndFor
    \State \textbf{return} $S$ \Comment{Return the set of successors}
\EndProcedure
\end{algorithmic}
\label{alg:algorithm-1}
\end{algorithm}

Further, we note that the homotopy polynomial for any graph $G$ with no cycles is $h_G = k$ for some positive integer $k$. Since removing edges cannot create a cycle, any spanning subgraph lower than $G$ must have a homotopy polynomial $h = k + r$ for some positive integer $r$.
Therefore, any chain with lowest homotopy polynomial equal to some integer can be uniquely extended backwards by simply incrementing the integer by $1$ at each stage. 
Similarly, any connected subgraph has a homotopy polynomial taking the form $h = l^k$ for some non-negative integer $k$. 
Since adding an edge cannot disconnect the graph, any graph higher than this has a homotopy polynomial of the form $h = x^{k+r}$ for positive integer $r$. 
Therefore any chain with highest homotopy polynomial equal to $x^k$ can be uniquely extended forward by incrementing the exponent $k$.
We use this to extend the chains in Section~\ref{sec:application} to maximal chains for comparison in the posets directly constructed in this section.

Because $\mathcal{H}\left(N\right)$ possesses a unique minimal element, any element in $\mathcal{H}\left(N\right)$ can be obtained through the repeated application of these transformations starting with this minimal element.
However, not every element constructed in this way is the image of some subgraph in the original poset under the quotienting operation, so we next present a method for filtering out the constructed covers that cannot be obtained by quotienting a subgraph in the original poset.

We note that we can compute the Euler characteristic of a graph $G$ in two different ways: $L = \|V\| - \|E\|$ or  $L =h_G\left(1\right) - \frac{dh_G}{dx}\left(1\right)$.
Since all spanning subgraphs $S$ have the same number of vertices, $N$, the Euler characteristic of $h_S\left(x\right)$ can give the number of edges that must have been in \emph{any} spanning subgraph that quotients to $h_S\left(x\right)$.
Given $h_S\left(x\right)$, we can compute $C\left(h_S\left(x\right)\right)$, the minimum number of edges present in \emph{any} simple graph (not necessarily on $N$ vertices) that quotients to $h_S\left(x\right)$.
We denote this number $C\left(h_S\left(x\right)\right)$, the \emph{edge cost} of the homotopy polynomial.
A key property of this cost function, is linearity:
\begin{equation}
    C\left(h+g\right)=C\left(h\right)+C\left(g\right).
\end{equation} 
The addition of graph homotopy polynomials corresponds to the disjoint union of their preimages, and the disjoint union of two graphs has the same number of edges as the sum of the edges in the two separate graphs.
We therefore only have to know $C\left(x^k\right)$ for all $k$, with $C\left(h_S\left(x\right)\right)$ being determined by linearity.

Let us consider a few examples for small $k$, and the general pattern will be made clear.
Clearly, $C(x^0)=0$, since the graph consisting of a single vertex with no edges has a homotopy polynomial of $x^0$, and it possesses no edges.
Indeed, the cost of any graph with no edges must be zero, hence the cost of any integer is zero.
Next, considering $k=1$, the minimal simple graph with a single loop is the complete graph on $3$ vertices, which has $3$ edges, hence $C\left(x\right) = 3$.

As one final example, how many edges are required to compute a connected graph with two loops?
If we consider $K_3$, we have a minimal graph with one loop.
To add another loop, we at least have to add another edge; however, because $K_3$ is a complete graph, it doesn't have room for an additional edge. 
We therefore need to add an additional vertex first, and then add two additional edges to first connect the new vertex to $K_3$, and then to form the additional loop. 
Therefore, $C\left(x^2\right) = 5$, and we see a more general pattern emerging:
If a minimal representative of the graph polynomial $x^k$ is not a complete graph, then $C\left(x^{k+1}\right) = C\left(x^k\right)+1$. 
If the minimal representative of the graph polynomial $x^k$ is a complete graph, then $C\left(x^{k+1}\right) = C\left(x^k\right) +2$.
Denoting $C\left(x^k\right)$ as $C_k$, we obtain the sequence
\begin{equation}
    C_k = 0, 3, 5, 6, 8, 9, ...
\end{equation}
This sequence is the natural numbers with a set of numbers skipped, these skipped numbers $B_k = 1,2,4,7,...$ being one more than the number of edges in the complete graphs $K_k$.
To avoid having to compute these skipped numbers and manually extract them from the naturals, we provide the following direct expression for $C_k$:
\begin{equation}
    C_0 = 0,
\end{equation}
\begin{equation}
    C_{k>0} = \left \lfloor k+\frac{3+\sqrt{8k-7}}{2}\right \rfloor.
\end{equation}
The expression for $C_{k>0}$ is obtained by inverting the expression for the number of edges in a complete graph on $k$ vertices, with some additional modifications to ensure the numbers one above these numbers of complete graph edges are skipped.

The general algorithm for constructing $\mathcal{H}\left(N\right)$ is then as follows:

\begin{algorithm}[h]
\caption{Direct Construction of $\mathcal{H}\left(N\right)$}
\label{alg:polyposet}
\begin{algorithmic}
\Procedure{HomotopyPolynomialPoset}{$N$}
    \State $A \gets \{N\}$ \Comment{Initialise a set of polynomials to process with the polynomial $N$}
    \State $V \gets \{\}$ \Comment{Initialise vertex set as empty}
    \State $E \gets \{\}$ \Comment{Initialise edge set as empty}
    \While{$A \neq \{\}$} \Comment{While there are unprocessed polynomials}
        \State $h \gets$ \textbf{pop} $A$ \Comment{Select a polynomial to process}
        \If{$h\notin V$}\Comment{Check to see if we have already processed $h$}
        \State $V \gets$ \textbf{push} $h$ \Comment{Add $h$ to the vertex set $V$}
        \State $m \gets \frac{dh}{dl}(1)-h(1)+N$ \Comment{$m$ is the edge allowance of $h$.}
        \State $B \gets$ \textsc{Succs}$\left(h\right)$ \Comment{Generate the successors of $h$}
        \For{$g$ in $B$} \Comment{Loop through successors}
            \If{$C(g)\leq m+1$} \Comment{Check to ensure edge allowance isn't exceeded}
            \State $A \gets$ \textbf{push} $g$ \Comment{Add $g$ to be processed}
            \State $E \gets$ \textbf{push} $\{h, g\}$ \Comment{Add the edge $\{h,g\}$ into the edge set}
            \EndIf
        \EndFor
        \EndIf
    \EndWhile
    \State \textbf{return} $V$, $E$ \Comment{Return the vertices and edges in the face poset of $\mathcal{H}\left(N\right)$}
\EndProcedure
\end{algorithmic}
\label{alg:algorithm-2}
\end{algorithm}

Note that this algorithm is self terminating; once we reach layer $\frac{N\left(N-1\right)}{2}+1$, all candidate polynomials will violate the cost constraint and we will be left with no further terms to examine.

\begin{exmp}[$\mathcal{H}\left(3\right)$]
We compute the graded poset $\mathcal{H}\left(3\right)$ in two different ways. 
First, we explicitly compute the quotienting of the poset depicted in Figure \ref{fig:k3}.
Note that the graphs appearing at each layer in this Figure are all isomorphic as graphs, so by counting connected components and loops, we easily deduce that the resulting quotiented poset will represented by the following path graph:

\begin{figure}[h]
    \centering
    \includegraphics[width = .5\textwidth]{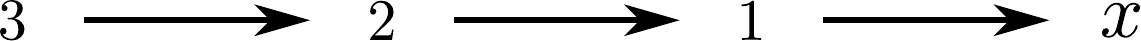}
    \caption{$\mathcal{H}\left(3\right)$}
    \label{fig:K3_homotopy}
\end{figure}
Secondly, we want to compute this poset directly.
\begin{enumerate}
    \item Begin with $A \leftarrow  \{3\}$.
    \item Pop $3$ from $A$; add $3$ to $V$.
    \item $m\leftarrow 0$; $B\leftarrow \textsc{Succs}\left(3\right) = \{2, x+2\}$
    \item $C\left(2\right)=0\leq m+1$ and $C\left(x+2\right)=3>m+1$, so we add $2$ to $A$, add the edge $\{3,2\}$ to $E$, and discard $x+2$.
    \item Pop $2$ from $A$; add $2$ to $V$.
    \item $m\leftarrow 1$; $B\leftarrow \textsc{Succs}\left(2\right) = \{1,x+1\}$
    \item $C\left(1\right) = 0\leq m+1$, and $C\left(x+1\right)=3>m+1$ which means that we add $1$ to $A$, $\{2,1\}$ to $E$ , and discard $x+1$.
     \item Pop $1$ from $A$; add $1$ to $V$.
    \item $m\leftarrow 2$; $B\leftarrow \textsc{Succs}\left(1\right) = \{x\}$
    \item $C\left(x\right) = 3\leq m+1$, so we add $x$ to $A$, $\{1,x\}$ to $E$.
    \item Pop $x$ from $A$; add $x$ it to $V$
    \item $m\leftarrow 3$; $B\leftarrow \textsc{Succs}\left(x\right) = \{x^2\}$
    \item $C\left(x^2\right)=5>m+1$, so we discard $x^2$, and the algorithm terminates.
\end{enumerate}
\end{exmp}

While in this trivial case, computing the quotienting directly is relatively straightforward due to the fact that the spanning subgraphs at each layer are obviously isomorphic to each other, and perhaps easier than directly computing $\mathcal{H}\left(3\right)$ using homotopy polynomials, the advantage of homotopy polynomials is quickly seen in the case of even slightly larger graphs.

\begin{exmp}[$\mathcal{H}\left(5\right)$]
We want to directly construct the quotiented poset for $K_5$ by the above algorithm.
\begin{enumerate}
    \item We begin with $A = \{5\}$.
    \item As in the previous example, we don't satisfy the edge cost criterion with promoted polynomials until layer $3$, so we pick up the algorithm at that state with $V=\{5,4,3\}$, $E=\{\{5,4\},\{4,3\},\{3,2\},\{3,x+2\}\}$, and $A=\{2,x+2\}$.
    \item $\textsc{Succs}\left(2\right)=\{1, x+1\}$, both of which are retained, so we have $V=\{5,4,3,2\}$, add edges $\{2,1\}$ and $\{2,x+1\}$, and $A=\{x+2,1,x+1\}$.
    \item $\textsc{Succs}\left(x+2\right)=\{x+1, 2x+1, x^2+2\}$, but only $x+1$ is retained, so we have $V=\{5,4,3,2,x+2\}$, add edge $\{x+2,x+1\}$, and $A=\{1,x+1\}$.
    \item $\textsc{Succs}\left(1\right) = \{x\}$ which is retained, so we have $V=\{5,4,3,2,x+2,1\}$, add edge $\{1,x\}$, and $A=\{x+1,x\}$.
    \item $\textsc{Succs}\left(x+1\right)=\{x,2x,x^2+1\}$, of which only $x$ and $x^2+1$ are retained, so we have $V=\{5,4,3,2,x+2,1,x+1\}$, add edges $\{x+1,x\}$ and $\{x+1,x^2+1\}$, and $A=\{x,x^2+1\}$.
    \item \textsc{Succs}$\left(x\right)=\{x^2\}$, which is retained, so we have $V=\{5,4,3,2,x+2,1,x+1,x\}$, add edge $\{x,x^2\}$, and $A=\{x^2+1,x^2\}$.
    \item $\textsc{Succs}\left(x^2+1\right)=\{x^2,x^3+1,x^2+x\}$, but only $x^2$ and $x^3+1$ are retained, so we have $V=\{5,4,3,2,x+2,1,x+1,x,x^2+1\}$, add edges $\{x^2+1,x^2\}$ and $\{x^2+1,x^3+1\}$, and $A=\{x^2,x^3+1\}$.
    \item $\textsc{Succs}\left(x^2\right)=\{x^3\}$ and $\textsc{Succs}\left(x^3+1\right) =\{x^3,x^3+x,x^4\}$, but only $x^3$ is retained from either of these, so after two passes through the while loop, we add edges $\{x^2,x^3\}$ and $\{x^3+1,x^3\}$, and $A=\{x^3\}$
    \item Taking successors causes the exponent in $x^3$ simply to increment one by one until we reach $x^7$, which violates the edge cost constraint, and our algorithm terminates.
\end{enumerate}
\begin{figure}[h]
    \centering
    \includegraphics[width = \textwidth]{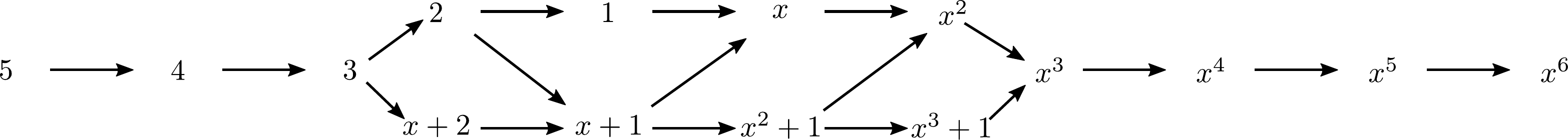}
    \caption{$\mathcal{H}\left(5\right)$ directly computed using graph homotopy polynomials.}
    \label{fig:K5Poset}
\end{figure}

There are $2^{10}$ spanning subgraphs in $K_5$; these ten steps are much simpler than computing and comparing $1024$ subgraphs explicitly, deriving their homotopy polynomials, and grouping the results according to all the induced covering relations derived from the subgraph inclusion relationship.
\end{exmp}

\section{Complexity Bounds and Connections to Number Theory}\label{sec:complexity}

A natural question with any new algorithm is time complexity.
We do not answer the time complexity of Algorithm \ref{alg:polyposet} directly, but instead examine the poset $\mathcal{H}(N)$ as $N$ grows. 
In this section, we establish rough bounds for the size of this poset by instead asking an easier question: how many elements are in $\mathcal{H}(N)$ at Euler characteristic $L = N - M$.
Equivalently, we could ask the following question:
    How many graphs are there up to homotopy equivalence with $M$ edges and Euler characteristic $L$?
We call the quantity answering this questions $H(M,L)$. 
Let $\pol(M,L)$ denote the set of polynomials with nonnegative integer coefficients with cost less than or equal to $M$ and Euler characteristic equal to $L$.
Asking the size of $|\pol(M,L)|$ is equivalent to asking $H(M,L)$, with one notable exception. 
The number of graphs with $M$ edges and Euler characteristic $L$ up to homotopy is exactly the number of elements of $\pol(M,L)$ \textit{that can be attained by the edge contraction construction} from a graph with $M$ edges and $L+M$ vertices. For most values of $M$ and $L$, every element of $\pol(M,L)$ can be attained by this construction. However $\pol(M,0)$ contains the zero polynomial, which if $M > 0$ cannot be obtained as the homotopy polynomial of a real graph, since $N=M-L$, and the zero polynomial is the homotopy polynomial of the empty graph. 
Hence in this case $H(M,L) = |\pol(M,L)| - 1$.

In summary,
\begin{equation}
    H(M,L) = \begin{cases}
            |\pol(M,L)| - 1 &  L = 0,\; M>0 \\
            |\pol(M,L)| & \mathrm{otherwise.} 
        \end{cases}
        \label{eqn:htpypolcount}
\end{equation}
Hence we may study $H(M,L)$ via the size of the set $\pol(M,L)$. The following notion will be extremely useful to us.

\begin{defin}
Let $S \subseteq \mathbb{Z}^d$ and $\mathbf{n} \in \mathbb{Z}^d$, and $\mathbb{N}:= \{0,1,2,3,\ldots\}$. The set of \emph{generalized partitions of} $\mathbf{n}$ \emph{into elements of $S$}, denoted $\pa_S(\mathbf{n})$, refers to the set of all functions $f:S \to \mathbb{N}$ that are zero with finite exceptions satisfying
\begin{equation}
    \sum_{s\in S} sf(s)=\mathbf{n}.
\end{equation}
For convenience, we also define $p_S(\mathbf{n}) : = |\pa_S(\mathbf{n})| \in \mathbb{N} \cup \{\infty\}$.
When $d=1$ we allow ourselves to omit the word generalized, calling $\pa_S(\mathbf{n})$ the partitions of $\mathbf{n}$ into elements of $S$.
\end{defin}

The most famous example of this notion comes not from topology, but number theory. Let $\mathbb{N}_{\ge 1}: =\{1,2,3,\ldots\}$. The number $p_{\mathbb{N}_{\ge 1}}(n)$, often denoted simply by $p(n)$, is called the number of partitions of $n$, the number of ways $n$ can be written as a sum of positive natural numbers, up to commutativity.

For our purposes we will need to study more complicated kinds of generalized partitions. Let
\begin{equation}
    Q := \{(C_k,1-k)\}_{k\in\mathbb{N}}.
\end{equation}
Hence $Q$ is the set of paired costs and Euler characteristics of $l^k$, for $k$ varying across nonnegative numbers. There is a bijection between polynomials with Euler characteristic $L$ and cost \emph{exactly} $M$ and $\pa_Q(M,L)$, namely
\begin{equation}
    \bigoplus_{i = 0}^\infty a_k l^k\; \longmapsto\; \big(f:(C_i,1-i) \mapsto a_i\big).
\end{equation}

Of course, what we are actually interested in is the set $\pol(M,L)$, the polynomials with Euler characteristic $L$ and cost less than or equal to $M$. With a slight modification to $Q$ we can relate $\pol(M,L)$ to a set of generalized partitions. Let
\begin{equation}
    R := Q \cup \{(1,0)\}.
\end{equation}
There is a bijection between $\pol(M,L)$ and $\pa_R(M,L)$ given by
\begin{equation}
    \bigoplus_{i = 0}^\infty a_k l^k\; \longmapsto\; \bigg(f:(C_i,1-i) \mapsto a_i, \; (1,0) \mapsto M - \bigoplus_{i=0}^\infty C_ia_i\bigg) .
\end{equation}.

Hence, $|\pol(M,L)| = p_R(M,L)$. Using this in Equation \ref{eqn:htpypolcount}, we have
\begin{equation}
    H(M,L) = \begin{cases}
            p_R(M,L) - 1 &  L = 0,\; M>0 \\
            p_R(M,L) & \mathrm{otherwise.} 
        \end{cases}
        \label{eqn:htpyparcount}
\end{equation}

We now deduce properties of $p_R(M,L)$ in order to do the same for $H(M,L)$. We have an injection from $\pa_R(M,L)$ to $\pa_R(M+1,L)$, given by $f\mapsto \widetilde{f}$, where $\widetilde{f} = f$ except on $(1,0)$, where $\widetilde{f}(1,0) = f(1,0) + 1$, i.e.
\begin{equation}
    \widetilde{f}(M,L):=\begin{cases}
        f(M,L) + 1 & (M,L) = (1,0)\\
        f(M,L) & \mathrm{otherwise}.
    \end{cases}
\end{equation}

Similarly, we have an injection from $\pa_R(M,L)$ to $\pa_R(M,L+1)$ given by $f\mapsto \overline{f}$, where $\overline{f} = f$ except on $(0,1)$, where $\overline{f}(0,1) = f(0,1) + 1$, i.e.
\begin{equation}
    \overline{f}(M,L):=\begin{cases}
        f(M,L) + 1 & (M,L) = (0,1)\\
        f(M,L) & \mathrm{otherwise}.
    \end{cases}
\end{equation}

This shows:

\begin{prop}
\label{prop:princ}
$p_R(M,L)$ is weakly increasing in both coordinates.
\end{prop}

From this we can deduce the following about $H(M,L)$.
\begin{cor}
$H(M,L)$ is weakly increasing in $L$. $H(M,L)$ is also weakly increasing in $M$, with the exception that $H(0,0) = 1$, while $H(1,0) = H(2,0) = 0$.
\end{cor}

\begin{proof}
For the first statement, by Equation \ref{eqn:htpyparcount} and the previous proposition, we only need to check that $H(M,-1) \leq H(M,0)$ when $M>0$, or equivalently, that $p_R(M,-1) < p_R(M,0)$ when $M>0$. This is true since the map $f\mapsto \overline{f}$ from $\pa_R(M,-1)$ to $\pa_R(M,0)$ is not a surjection. Indeed, the function $g \in \pa_R(M,0)$ which sends $(1,0)$ to $M$ and everything else to zero cannot be in the image of this map.

For the second statement, we always have $H(M,L) = H(M+1,L)$ by Equation \ref{eqn:htpyparcount} and the previous proposition unless $M = L = 0$. Here, direct computation shows $p_R(0,0) = p_R(1,0) = p_R(2,0) = 1$, but $p_R(3,0) = 2$. Applying Equation \ref{eqn:htpyparcount} gives the exception.
\end{proof}

Next, we want to show that $H(M,L)$ stabilises as a function of $L$ when $L$ is large enough.
Intuitively, once we have more vertices than edges, adding more vertices cannot increase the number of permissible loops in any connected component. Instead, each homotopy equivalence class will merely have one more connected component for each vertex we added. 
As usual we first prove an analogous result for $p_R$.

\begin{prop}
$p_R(M,L_1) = p_R(M,L_2)$ whenever $L_1,L_2 \geq 0$.
\end{prop}

\begin{proof}
It suffices to show that the map $f \mapsto \overline{f}$ from $\pa_R(M,L)$ to $\pa_R(M,L+1)$ is a surjection and hence a bijection when $L\geq0$. 
If the map $f\mapsto \overline{f}$ is not surjective, that means there is a map $g\in p_R(M,L+1)$ such that $g(0,1) = 0$. Letting $W = R-\{(0,1)\}$, we see
\begin{equation}
    \sum_{s\in W} sg(s) = \sum_{s\in R} sg(s) = (M,L+1).
\end{equation}
On the left we have a sum of elements of $\mathbb{Z}^2$ with nonpositive second coordinate. On the right we have an element of $\mathbb{Z}^2$ with positive second coordinate. This is a contradiction, and we are done.
\end{proof}

Immediate is the following result about $H(M,L)$.
\begin{cor}
$H(M,L_1) = H(M,L_2)$ whenever $L_1,L_2 > 0$.
\end{cor}

In light of this result, it makes sense to define $H(M)$, the limit of $H(M,L)$ as $L$ approaches infinity. The fact that $H(M)$ is dependent on only one variable suggests that we might be able write $H(M)$ as the number of partitions of $M$ into elements of some subset of $\mathbb{N}_{\ge 1}$. To this end, we define the set
\begin{equation}
    T := \{1\} \cup \{C_k\}_{k\in \mathbb{N}_{\ge 1}} = \{1,3,5,6,8,9,\ldots\}.
\end{equation}
It turns out $H(M)$ is the number of partitions of $M$ into elements of $T$:
\begin{prop}
$H(M) = p_T(M)$.
\end{prop}
\begin{proof}
By the previous corollary it suffices to show that $H(M,1) = p_T(M)$. Then by Equation \ref{eqn:htpyparcount} and the last proposition it suffices to show that $p_R(M,0) = p_T(M)$. We have a map from $\pa_R(M,0)$ to $\pa_T(M)$:

\begin{equation}
    \bigg(f:(C_k,1-k) \mapsto a_k, \; (1,0) \mapsto b\bigg)\; \longmapsto \bigg(g:1 \mapsto b,\; C_k \mapsto a_k \textrm{ for }k\geq 1\bigg),
\end{equation}
and this map has inverse
\begin{equation}
    \bigg(g:1\mapsto b,\; C_k\mapsto a_k \textrm{ for }k\geq 1\bigg)\; \longmapsto \;\bigg(f:(1,0) \mapsto b,\; (C_k,1-k) \mapsto a_k \textrm{ for } k\geq 1,\; (0,1) \mapsto -\bigoplus_{k=1}^\infty a_k(1-k)\bigg),
\end{equation}
completing the proof.
\end{proof}

We now state the main result of the section

\begin{theorem}
The functions $H(M)$ and $|\mathcal{H}(N)|$ satisfy the asymptotic formulae
\begin{align*}
    \log H(M) & \sim \pi \bigg[\frac{2}{3}\bigg]^{1/2} M^{1/2}\\
    \log |\mathcal{H}(N)| &\sim 3\Big[\frac{1}{4}\zeta(3)\Big]^{1/3}N^{2/3},
\end{align*}
where $\zeta(3)$ is the constant
\begin{equation*}
    \zeta(3) := \sum_{n = 1}^{\infty} \frac{1}{n^3}.
\end{equation*}
\end{theorem}

\begin{proof}
As a formal power series we have the equality
\begin{equation}
\label{eqn:ptgen}
    \sum_{n = 0}^\infty p_T(n)z^n = \prod_{k=1}^\infty (1-z^{T_k})^{-1},
\end{equation}
where $T_k$ is the $k^{\mathrm{th}}$ smallest entry of $T$. This is obtained by factoring the right side of Equation \ref{eqn:ptgen} using
\begin{equation*}
    (1-z^{T_k})^{-1} = (1 + z^{T_k} + z^{2T_k} + \ldots).
\end{equation*}

Similarly, one obtains
\begin{equation*}
    \sum_{(n,m)\in \mathbb{Z}^2}^\infty p_R(n,m)z_1^nz_2^m = (1-z_1)^{-1}\prod_{k=1}^\infty (1-z_1^{C_k}z_2^{1-k})^{-1}.
\end{equation*}
Setting $z = z_1 = z_2$ we have
\begin{equation*}
    \sum_{n = 0}^{\infty}\sum_{a + b = n} p_R(a,b) z^n = (1-z)^{-1}\prod_{k=1}^\infty (1-z^{C_k + 1-k})^{-1} = (1-z)^{-2}\prod_{k=3}^{\infty}(1-z^k)^{-(k-2)}.
\end{equation*}
The last equality above follows from the fact that the sequence $\{C_k\}_{k\geq 1}$ can be constructed by starting with 3, then skipping the number 4, continuing with the next two numbers (5 and 6), skipping the number 7, continuing with the next three numbers, and so on.

By replacing $z$ with $e^{-s}$ we have
\begin{align*}
        \sum_{n = 0}^\infty p_T(n)e^{-sn} &= \prod_{k=1}^\infty (1-e^{-sT_k})^{-1}\\
        \sum_{n = 0}^{\infty}\sum_{a + b = n} p_R(a,b) e^{-sn} &= (1-e^{-s})^{-2}\prod_{k=3}^{\infty}(1-e^{-sk})^{-(k-2)}.
\end{align*}
The numbers $p_T(n)$ are increasing in $n$ since $1\in T$. Similarly, the numbers $\sum_{a+b = n}p_R(a,b)$ are increasing in $n$ by Proposition \ref{prop:princ}.
Notice also that letting $Y = \{\frac{n(n+1)}{2} + 1 : n\in\mathbb{N}\}$ and $Y_{\leq n} = \{k\in Y: k\leq n\}$, we have
\begin{align*}
    \big|T \cap \{1,\ldots, n\}\big| =  n - |Y_{\leq n}| + 1 = n + O(\sqrt{n}).
\end{align*}
Moreover, for $w\geq 3$
\begin{equation*}
    2 + \sum_{k = 3}^w (k-2) = 2 + \sum_{k=1}^{w-2} k = 2 + \frac{(w-2)(w-1)}{2} = \frac{1}{2}w^2 + O(w).
\end{equation*}
Applying to the above conclusions to the main result of \cite{brigham1950general}, we obtain the asymptotic formulas
\begin{align*}
    \log p_T(n) & \sim \pi \bigg[\frac{2}{3}\bigg]^{1/2} n^{1/2}\\
    \log \sum_{a + b = n} p_R(a,b) &\sim 3\Big[\frac{1}{4}\zeta(3)\Big]^{1/3}n^{2/3}.
\end{align*}
The theorem then follows from the fact that $H(M) = p_T(M)$ and
\begin{equation*}
    |\mathcal{H}(N)| = \sum_{M + L = N} H(M,L) = -1 + \sum_{M+L=N}p_R(M,L),
\end{equation*}
whenever $N\geq 1$.
\end{proof}

\section{Chain Comparisons}\label{sec:chain-compare}
We have established the space of possible spanning subgraphs, partially ordered by inclusion, graded by Euler characteristic, and quotiented by homotopy equivalence.
We now must construct a metric to compare chains in these posets analogous to our comparison in Section \ref{sec:graph-filtrations} using Kendall's $d_K$. 
To do this, we take inspiration from the unquotiented poset, and generalise to the quotiented poset.
After quotienting by homotopy equivalence, the underlying symmetric group structure is destroyed, which is apparent by observing that while the unquotiented posets are symmetric with respect to reversing the order relation (with this transformation being identified with replacing each element with its complement), but the quotiented poset is not symmetric with respect to reversing the order relation.
Therefore, we seek to generalise Kendall's $d_K$ to chains in more general graded posets; we want to examine the impact of an adjacent transposition in $\Xi$ on its associated edge filtration, which we denote $\mathcal{F}\left(\Xi\right)$.

Suppose $\Xi = \{e_1,e_2...,e_k,e_{k+1},...,e_{M-1}, e_M\}$, and let $\sigma$ be the transposition between elements $k$ and $k+1$.
Examining $\mathcal{F}\left(\Xi\right)$ and $\mathcal{F}\left(\sigma\left(\Xi\right)\right)$, we see that for $i<k$, $S_i\left(\Xi\right)=S_i\left(\sigma\left(\Xi\right)\right)$, since $\Xi$ and $\sigma\left(\Xi\right)$ only differ at positions $k$ and $k+1$.
At position $k$, $S_k\left(\Xi\right)$ contains $e_k$, while $S_k\left(\sigma\left(\Xi\right)\right)$ contains $e_{k+1}$, hence these two chains differ at position $k$.
However, at position $k+1$, both of these chains contain elements $e_k$ and $e_{k+1}$, and hence are equal at position $k+1$, and all following positions.
Hence an adjacent transposition in $\Xi$ results in a change in $\mathcal{F}$ at exactly one position. 
This motivates the definition of a \emph{discrete homotopy}, through \emph{adjacent chains}.

\begin{defin}[Adjacent Chains]
    Two maximal chains $\mathcal{F}_1$ and $\mathcal{F}_2$ are \emph{adjacent} if they differ in exactly one position.
\end{defin}
\begin{defin}[Discrete Homotopy]
Let $\{\mathcal{F}_i\}|_{i=0}^d$ be a sequence of maximal chains in a graded poset with the property that $\mathcal{F}_k$ and $\mathcal{F}_{k+1}$ are adjacent for all $k\in\{0,...,d-1\}$.
The sequence $\{\mathcal{F}_k\}$ is a discrete homotopy of length $d$.
\end{defin}
Note that because the posets we are working in possess unique lowest and highest elements, all the discrete homotopies we will encounter can be considered relative discrete homotopies, i.e. those with fixed endpoints.
With this, we can generalise Kendall's $d_K$ metric: 
\begin{defin}[Discrete Homotopy Metric]\label{defn:sec:chain-compare:homotopy-metric}
Let $\mathcal{F}_I$ and $\mathcal{F}_F$ be maximal chains in a graded poset. The \emph{discrete homotopy distance} $d_H\left(\mathcal{F}_I,\mathcal{F}_F\right)$ between $\mathcal{F}_I$ and $\mathcal{F}_F$ is the length of the shortest discrete homotopy $\{\mathcal{F}_i\}|_{i=0}^d$ with $\mathcal{F}_0 = \mathcal{F}_I$ and $\mathcal{F}_d = \mathcal{F}_F$.
\end{defin}

This clearly reduces to Kendall's $d_K$ when the graded poset is the power set of some set partially graded by inclusion, however this definition can be applied to more general graded posets.
This metric can be formulated as the standard metric on a graph, where we consider a graph whose vertices are maximal chains, and whose edges connect maximal chains differing at exactly one position.
This metric is clearly bounded from below by $0$, but more sharply, for given $\mathcal{F}_I$ and $\mathcal{F}_F$, this metric is bounded from below by the number of locations where $\mathcal{F}_I$ differs from $\mathcal{F}_F$.

\begin{exmp}[Complete Construction for $K_4$]
\begin{figure}[h!]
    \centering
    \includegraphics[width = .25 \textwidth]{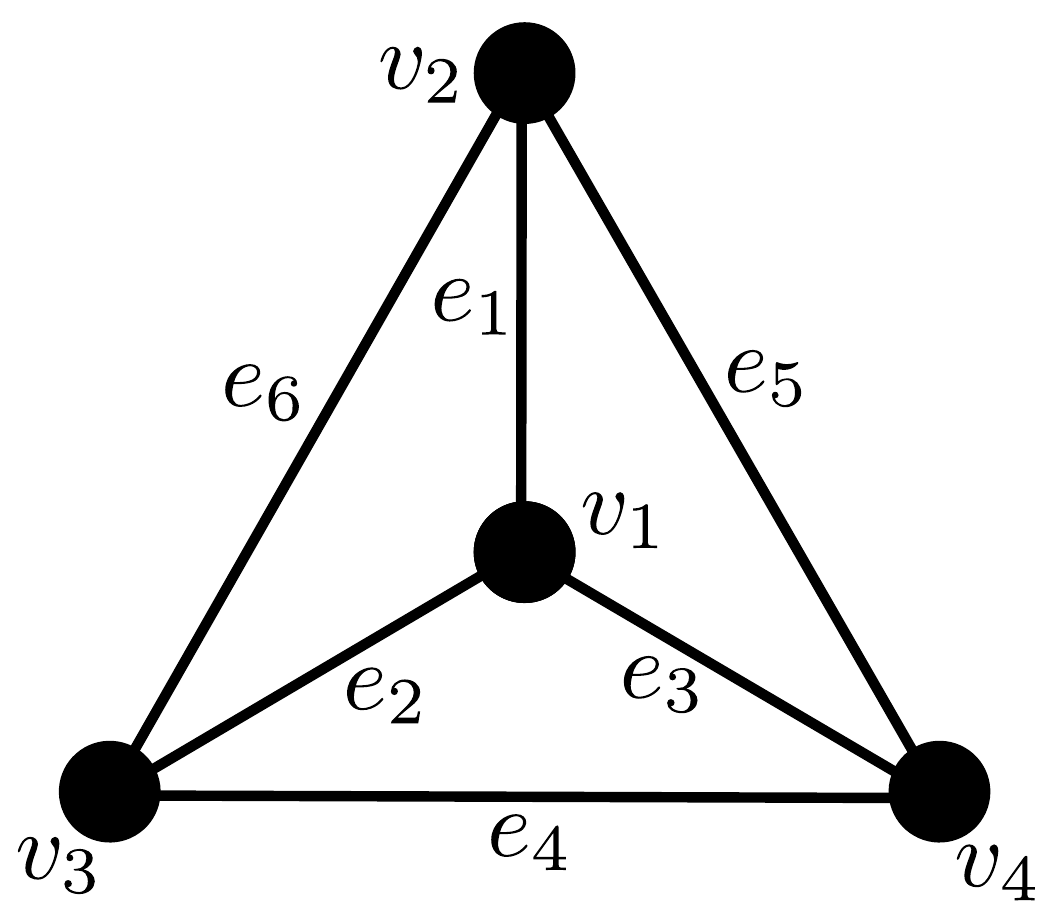}
    \caption{The labelled complete graph $K_4$.}
    \label{fig:K4}
\end{figure}
In order to fully convey the entire analysis pipeline, we examine the smallest complete graph that generates a nontrivial filtration space after quotienting, i.e. the smallest complete graph generating a poset with more than one maximal chain after quotienting by homotopy equivalence: $K_4$.

We begin with the complete graph on $4$ vertices, with vertices and edges labelled as in Figure \ref{fig:K4}.
Suppose we have two edge filtrations $\mathcal{F}_1$ and $\mathcal{F}_2$, given by the respective sequences $\Xi_1 = \{e_1,e_2,e_3,e_4,e_5,e_6\}$ and $\Xi_2 = \{e_6,e_5,e_4,e_3,e_2,e_1\}$.
We obtain the spanning subgraph poset depicted in Figure \ref{fig:PK4}, where the filtration $\mathcal{F}_1$ is coloured red, and $\mathcal{F}_2$ is coloured blue.

\begin{figure}[h!]
    \centering
    \includegraphics[width = .6 \textwidth]{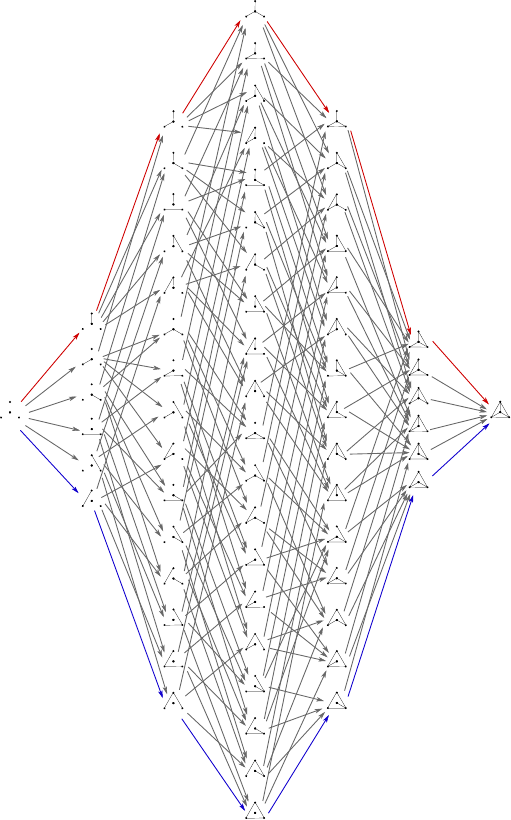}
    \caption{A depiction of $SS\left(K_4\right)$, partially ordered by subgraph inclusion, with filtration $\mathcal{F}_1$ in red and $\mathcal{F}_2$ in blue.}
    \label{fig:PK4}
\end{figure}

If we wanted to compare these two filtrations, we can use Kendall's $d_K$, since $d_K\left(\Xi_1,\Xi_2\right)=15$, and as shown by, these paths in Figure \ref{fig:PK4}, these two chains are as far apart as possible when measured according to Kendall's $d_K$.
This diagram, however is quite large, even for the small graph $K_4$, so we seek to quotient it by homotopy equivalence.
Applying Algorithm \ref{alg:polyposet}, we obtain the poset depicted in Figure \ref{fig:PK4quotient}.

\begin{figure}[h!]
    \centering
    \includegraphics[width = .6\textwidth]{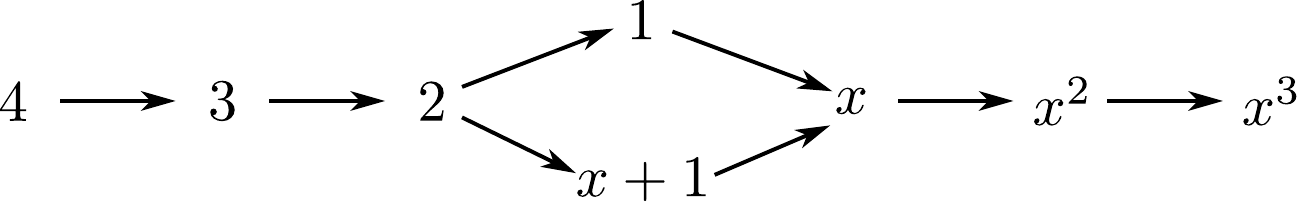}
    \caption{$\mathcal{H}\left(4\right)$ where each equivalence class is represented by its homotopy polynomial.}
    \label{fig:PK4quotient}
\end{figure}
We see clearly that the quotiented poset is significantly smaller than the unquotiented poset. 
Examining the two filtrations, we obtain the two sequences of homotopy polynomials $\mathcal{F}_1/\sim \,=\, \{4,3,2,1,x,x^2,x^3\}$ and $\mathcal{F}_2/\sim \,=\, \{4,3,2,x+1,x,x^2,x^3\}$.
These two chains are adjacent (differing only at position 3), hence they have discrete homotopy distance $1$, and are precisely the $2$ distinct maximal chains in Figure \ref{fig:PK4quotient}.
\end{exmp}

\section{Sufficient Conditions for Bounded $d_H$}\label{sec:meshes}
For the unquotiented poset, we clearly have a finite discrete homotopy distance between arbitrary maximal chains, since the symmetric group on the edge set is a finite group generated by adjacent transpositions.
This means discrete homotopy distances correspond to $d_K$, taking its maximum value of $\frac{M\left(M-1\right)}{2}$ for a total order inversion.
Therefore, the images of any two actual chains under quotienting will also have a finite discrete homotopy distance between them.
We can take the discrete homotopy connecting them in the unquotiented poset, and map that sequence under the quotient map to obtain a sequence of chains in the quotiented poset.
This sequence is not necessarily a discrete homotopy in the quotiented poset, since chains that initially differ may get mapped to the same quotiented chain, but the number of differences in successive chains cannot increase.
Therefore, successive elements in the quotiented sequence of chains will differ in at most $1$ position.
Removing duplicated chains will yield a shorter discrete homotopy in the quotiented space.
Therefore, any two actual chains in $\SS\left(K_N\right)$ will have a finite discrete homotopy connecting their images in $\mathcal{H}\left(N\right)$, bounded from above by $d_K$ applied to the unquotiented chains.
However, not all chains in a general quotiented poset are the images of chains in the corresponding unquotiented poset.
We seek to define a more general structure that is sufficient for guaranteeing the existence of finite discrete homotopies between arbitrary maximal chains, since we may be interested in comparing general chains in quotiented graded posets.

To prove that this metric is finite for a more general graded poset, we have to compute an upper bound for arbitrary maximal chains.
We introduce a condition which provides a sufficiency condition for the existence of arbitrary discrete homotopies.
\begin{defin}[left-covering condition]
Let $P$ be a finite graded poset such that for every $z_1$, $z_2\in P$ satisfying $z_0\lessdot z_1$ and $z_0 \lessdot z_2$ for some $z_0\in P$, there exists $z_3$ satisfying $z_1\lessdot z_3$ and $z_2 \lessdot z_3$. Such a $P$ satisfies the left-covering condition.
This condition is expressed in the following diagram, where the existence of $z_0$, $z_1$, $z_2$, and the solid arrows implies the existence of $z_3$ and the dashed arrows:
\[
\begin{tikzcd}[row sep=small]
&z_1\arrow[dr, dashed]&\\
z_0\arrow[ur] \arrow[dr] && z_3\\
&z_2\arrow[ur, dashed] &
\end{tikzcd}
\]
\end{defin}

The right-covering condition is obtained by reversing the order relationship in the above definition, and is similarly expressed by the following diagram:
\[
\begin{tikzcd}[row sep=small]
&z_1\arrow[dr]&\\
z_3\arrow[ur, dashed] \arrow[dr, dashed] && z_0 \\
&z_2\arrow[ur] &
\end{tikzcd}
\]

We have the following useful result:
\begin{theorem}
If a poset satisfies the left-covering condition, and has a unique lowest element, then arbitrary maximal chains can be connected by a finite discrete homotopy.
\end{theorem}

\begin{proof}
Let $\mathcal{F}_I$ and $\mathcal{F}_F$ be two arbitrary maximal chains in such a poset that we want to connect through a discrete homotopy. 
We constructively prove a discrete homotopy between $\mathcal{F}_I$ and $\mathcal{F}_F$ exists by induction by demonstrating:
\begin{itemize}
    \item (base case): Chains in this poset differing in only the last position can be connected by a discrete homotopy.
    \item (inductive step): If arbitrary chains in this poset differing in the last $k$ positions can be connected by a discrete homotopy, then arbitrary chains in this poset only differing in the last $k+1$ positions can be connected by a discrete homotopy.\\
\end{itemize}

First proving the base case.
By definition, if two chains, $\mathcal{F}_1$ and $\mathcal{F}_2$, only differ in the last position, they are adjacent chains.
They are connected by the discrete homotopy $\{\mathcal{F}_1,\mathcal{F}_2\}$.\\

Next, proving the inductive step.
Let $\mathcal{F}_1$ and $\mathcal{F}_2$ be maximal chains differing only in the last $k+1$ positions.
Let $p$ be the fist position where $\mathcal{F}_1$ differs with $\mathcal{F}_2$, $z_0$ be the common element in both chains at position $p-1$, and $z_1$ and $z_2$ be the differing elements at position $p$ in $\mathcal{F}_1$ and $\mathcal{F}_2$ respectively. 
By hypothesis, this poset possesses a unique minimal element, and $\mathcal{F}_1$ and $\mathcal{F}_2$ are maximal.  
Therefore $p > 0$, because at $p=0$, both chains must take the unique minimal element, i.e. $z_0$ exists.
Therefore, we have $z_0 \lessdot z_1$ and $z_0 \lessdot z_2$.
By the left covering condition, there therefore exists a common element $z_3$ at position $p+1$ such that $z_1 \lessdot z_3$ and $z_2 \lessdot z_3$.

We seek to construct two adjacent maximal chains \[\mathcal{F}_{M_1} = \{...,z_0,z_1,z_3,...\}\] and \[\mathcal{F}_{M_2} = \{...,z_0,z_2,z_3,...\}\] only differing at position $p$, with $\mathcal{F}_{M_1}$ identical to $\mathcal{F}_1$ and $\mathcal{F}_{M_2}$ identical to $\mathcal{F}_2$ up to and including position $p$.
Such chains, if they exist, only differ with $\mathcal{F}_1$ and $\mathcal{F}_2$ in the last $k$ positions, and therefore by hypothesis there exists a discrete homotopies  $\{\mathcal{F}_1,...,\mathcal{F}_{M_1}\}$ and $\{\mathcal{F}_{M_2},...,\mathcal{F}_2\}$.
However, because $\mathcal{F}_{M_1}$ and $\mathcal{F}_{M_2}$ are adjacent to each other, these discrete homotopies can be concatenated to form a discrete homotopy between $\mathcal{F}_1$ and $\mathcal{F}_2$: $\{\mathcal{F}_1,...,\mathcal{F}_{M_1},\mathcal{F}_{M_2},...,\mathcal{F}_2\}$.
All that remains is to show that $\mathcal{F}_{M_1}$ and $\mathcal{F}_{M_2}$ exist.\\

Because $\mathcal{F}_{M_1}$ and $\mathcal{F}_{M_2}$ agree with both $\mathcal{F}_{1}$ and $\mathcal{F}_{2}$ before position $p$, take either the value $z_1$ or $z_2$ at position $p$, take the value $z_3$ at position $p+1$, the values of these chains are already defined up to position $p+1$.
Further, these chains must agree with each other at all positions after $p+1$, so we must find a sequence of elements, each covering the previous, beginning at $z_3$ extending $\mathcal{F}_{M_1}$ and $\mathcal{F}_{M_2}$ into maximal chains.

If $z_3\in \mathcal{F}_1$, then we can take this continuation to simply be the remainder of $\mathcal{F}_1$.
If not, then both $z_3$ and the element in $\mathcal{F}_1$ at position $p+1$ cover $z_1$, which by the left covering condition implies the existence of another element, $z_4$ covering the element in $\mathcal{F}_1$ at position $p+1$, such that $z_3\lessdot z_4$.
We take this $z_4$ to be the next element of $\mathcal{F}_{M_1}$ and $\mathcal{F}_{M_2}$.
If these chains are not yet maximal, we repeat this argument to obtain the next element.

Namely, if $z_4\in \mathcal{F}_1$, we take the remainder of these chains to be identical with $\mathcal{F}_1$, and if not, we know there exists an element $z_5$ that covers both $z_4$ and the element in $\mathcal{F}_1$ at position $p+2$, because both of these elements in turn cover the element in $\mathcal{F}_1$ at position $p+1$.

Continuing by induction, suppose we have extended $\mathcal{F}_{M_1}$ and $\mathcal{F}_{M_2}$ up to the element $z_{s+2}$ that covers the element in $\mathcal{F}_1$ at position $p+s-1$.
If $z_{s+2}\in \mathcal{F}_1$, we take the remainder of these chains to be identical with $\mathcal{F}_1$, and if not, we know there exists an element $z_{s+3}$ that covers both $z_{s+2}$ and the element in $\mathcal{F}_1$ at position $p+s$, because both of these elements in turn cover the element in $\mathcal{F}_1$ at position $p+s-1$.
Because $\mathcal{F}_1$ is maximal, the left covering condition implies that this can be repeated until $\mathcal{F}_{M_1}$ and $\mathcal{F}_{M_2}$ are maximal, finishing the construction.

\end{proof}
Dually, arbitrary maximal chains in posets satisfying the right-covering condition with unique highest elements can be connected with finite discrete homotopies, and the proof is obtained by dualising the above proof.
The quotiented posets we are working with satisfy both the left and right-covering conditions with unique lowest and highest elements, hence both of these constructions can be used to establish upper bounds on the discrete homotopy metric between two arbitrary maximal chains we encounter.

This proof does not necessarily construct the optimal discrete homotopy, though it does bound its length from above.
The exact computation of $d_H\left(\mathcal{F}_1,\mathcal{F}_2\right)$ could be accomplished through a bidirectional search beginning from $\mathcal{F}_1$ and $\mathcal{F}_2$, and moving alternatively through adjacent chains until a discrete homotopy connecting $\mathcal{F}_1$ and $\mathcal{F}_2$ is found. 
By construction this discrete homotopy will have minimal length, however this bidirectional search rapidly becomes computationally intractable.
Therefore, we seek to establish an easily computed lower bound; in many cases the lower bound we will construct in the next section coincides with the length of the discrete homotopy described above, proving optimality for these examples while sidestepping the cost of directly computing $d_H$.

In the many query setting, i.e. when we want to compute many distances between maximal chains in the same $\mathcal{H}\left(N\right)$, it may be more efficient to construct all maximal chains, taking this set as the vertex set of a graph.
The edges of this graph can then be taken to be the set of pairs of adjacent chains.
This graph must only be constructed once; the computation of the discrete homotopy metric then becomes the problem of finding a shortest path through this unweighted graph.

\section{Comparison with Homology}\label{sec:homotopy_vs_homology}
We have considered quotienting the poset of spanning subgraphs by homotopy equivalence, and defining and computing a discrete metric on the space of chains in this quotiented poset.
This construction can similarly be repeated, but using homology as the equivalence relation for quotienting rather than homotopy equivalence.

Because homology groups can be determined from graph homotopy polynomials, this homology poset can be obtained by quotienting the directly constructed homotopy poset, without having to compute the spanning subgraph poset.
This means that chains and discrete homotopies in the poset quotiented by homotopy equivalence can be directly converted to chains and discrete homotopies in the poset quotiented by homology by computing the homology groups of each element in the chains, and deleting duplicated chains in the discrete homotopy.
Further, because graphs have only two nontrivial homology groups, the resulting homology poset can be represented by a subgraph of the square lattice graph, where the vertex coordinates $(b_1,b_0)$ are the first and zeroth Betti numbers.

Because the edges in this lattice graph are directed and either increment $b_1$ or decrement $b_0$ (i.e. pointing from $(b_1,b_0)$ to either $(b_1+1,b_0)$ or $(b_1,b_0-1)$), if a chain in this poset possesses an adjacent chain at a specified position, this adjacent chain is unique, and the four edges not held in common form a square. This can be seen by considering all four possible sequences of three homology groups in this poset, depicted in figure \ref{fig:homology_chains}.

\begin{figure}[h]
     \centering
     \begin{subfigure}[b]{0.22\textwidth}
         \centering
         \includegraphics[width=\textwidth]{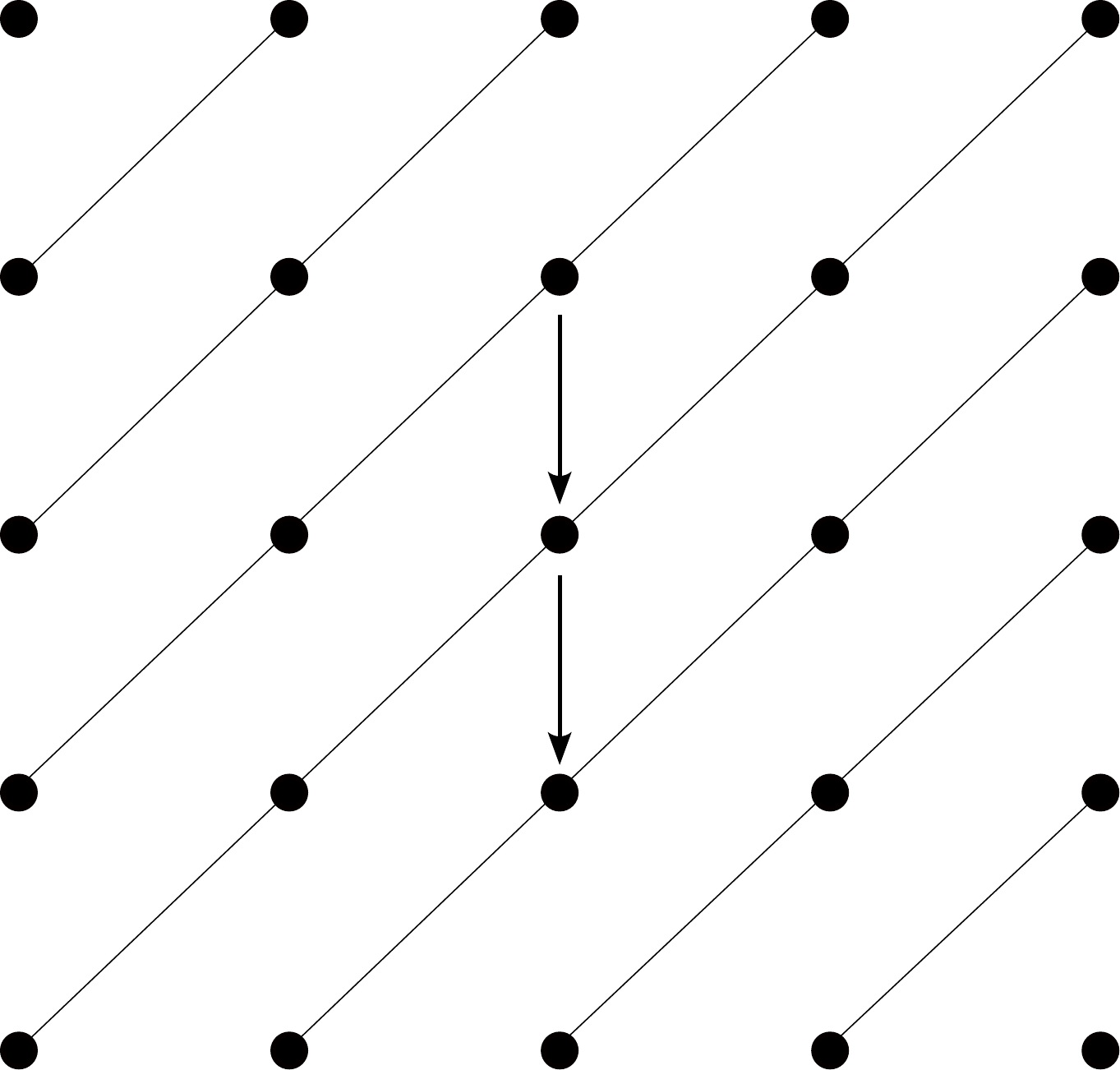}
         \caption{Twice decrementing $b_0$}
         \label{fig:2dec_b0}
     \end{subfigure}
     \hfill
     \begin{subfigure}[b]{0.22\textwidth}
         \centering
         \includegraphics[width=\textwidth]{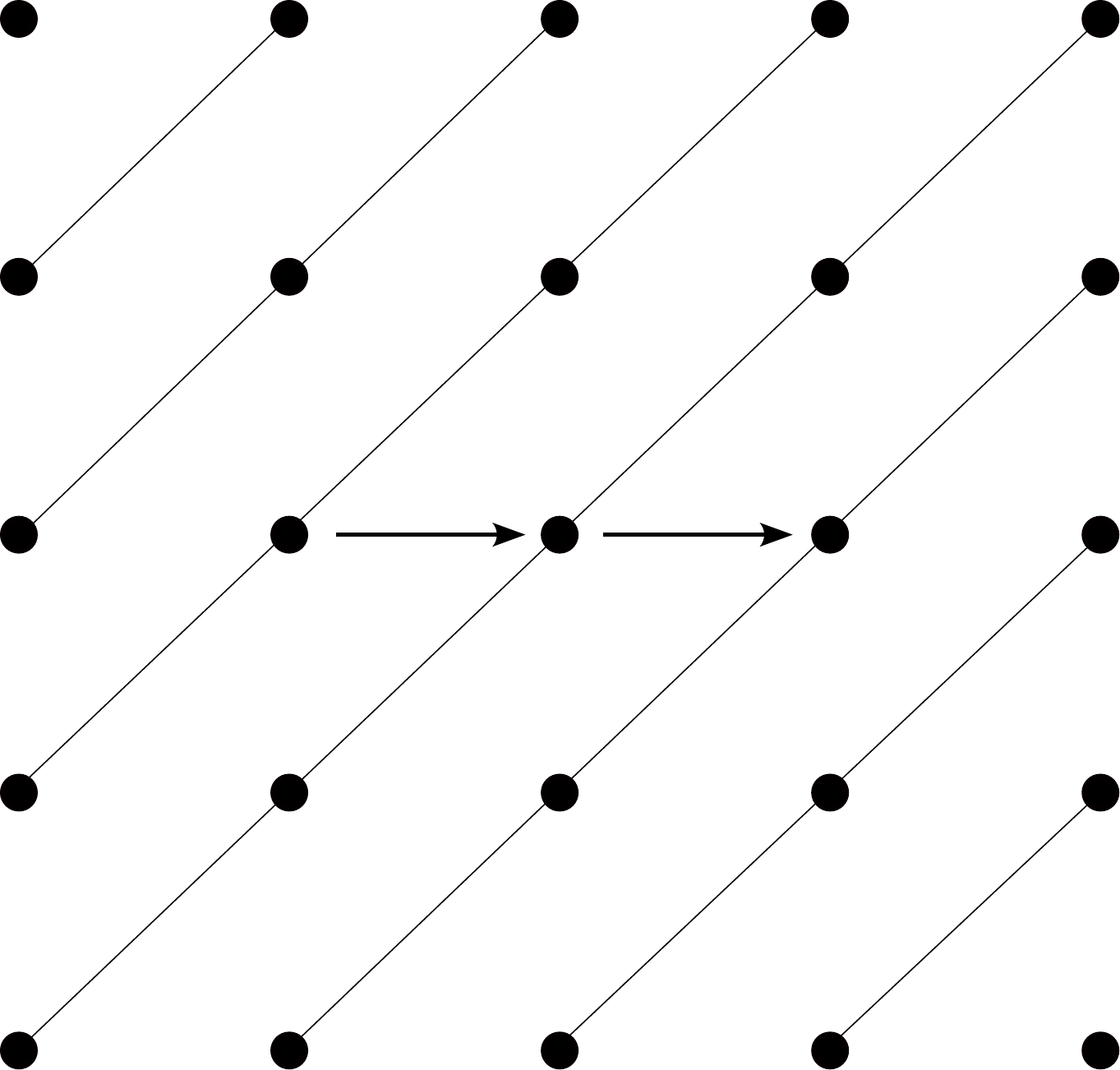}
         \caption{Twice incrementing $b_1$}
         \label{fig:2inc_b1}
     \end{subfigure}
     \hfill
     \begin{subfigure}[b]{0.22\textwidth}
         \centering
         \includegraphics[width=\textwidth]{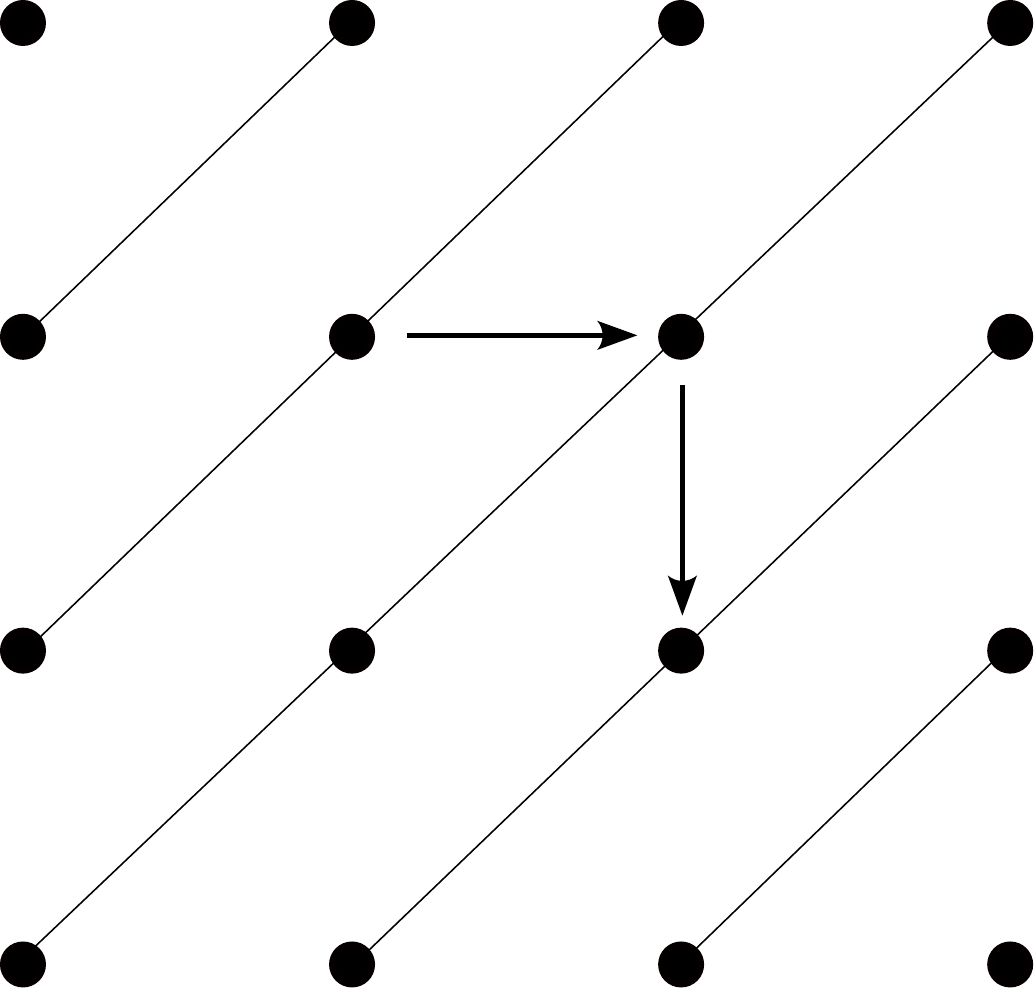}
         \caption{incrementing $b_1$ then decrementing $b_0$}
         \label{fig:inc_b1_dec_b0}
     \end{subfigure}
     \hfill
      \begin{subfigure}[b]{0.22\textwidth}
         \centering
         \includegraphics[width=\textwidth]{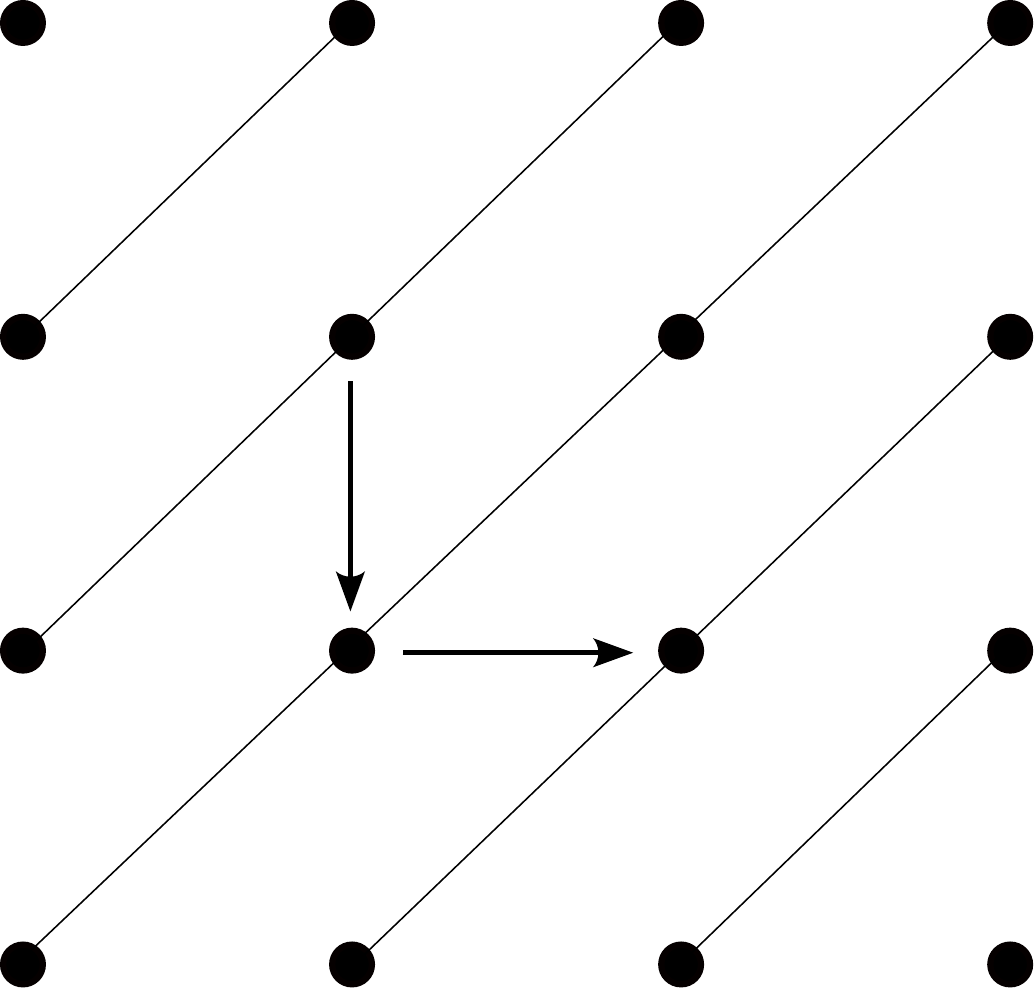}
         \caption{decrementing $b_0$ then incrementing $b_1$}
         \label{fig:dec_b0_inc_b1}
     \end{subfigure}
        \caption{All four possible three graph chains up to homology. The first two do not possess any adjacent chains at the central position, while the latter two are the unique chains adjacent to each other at the central position. Diagonal lines indicate the poset's grading.}
        \label{fig:homology_chains}
\end{figure}

A direct corollary of this is that the quotiented poset can be represented by a planar graph, and the discrete homotopy metric between chains in this poset corresponds to computing the area between the two boundary chains in the plane containing this lattice graph, since the area enclosed between two adjacent chains is exactly $1$.
This provides a qualitative interpretation of the discrete homotopy metric as a surface area bounded between curves, though to apply this interpretation to more general posets the embedding surface may require nontrivial topology.

Further, because discrete homotopies in the homotopy poset can be mapped to discrete homotopies in the homology poset, the discrete homotopy metric in the homology poset bounds the discrete homotopy metric in the homotopy poset from below, giving us an easily computed sharper lower bound for $d_H$.
\section{Application to Alzheimer's disease}\label{sec:application}
Alzheimer's disease (AD) is the most common form of dementia and is the fifth leading cause of death, for adults 65 and over, worldwide.  Alzheimer's disease is thought to be due to specific proteins in the brain, namely the proteins  \textit{Amyloid-$\beta$ (A$\beta$)} and \textit{Tau protein ($\tau$P)}, switching from a healthy form to a toxic one.  Toxic versions of these proteins have two main effects: toxic A$\beta$ and $\tau$P can cause their healthy counterparts to become toxic; and toxic A$\beta$ and $\tau$P can `stick together' to form large protein structures that deposit within in the brains of deceased AD patients.  Toxic proteins, from single proteins to smaller aggregates, are mobile within the brain \cite{braak1991neuropathological,jucker2013,cho2016vivo}.  Toxic proteins use the connections between different brain areas, called axon fibers, to move from one brain region to another, spread their toxic state and ultimately form the aggregated protein deposits that are the hallmarks of AD.  Moreover, toxic $\tau$P and its aggregates are thought to play a role in a number of brain malfunctions in AD, such as reducing the brain's energy creation capacity and increasing brain inflammation, and serve as a key factor causing primary brain cells, called neurons, to die as part of the disease.

Our natural abilities, such language and memory, are associated with specific brain regions and understanding the differences in how toxic $\tau P$ moves around in the brain can help us understand the specific challenges, and needs for care, faced by our elderly.  Typical AD, experienced by the largest number of AD patients, is characterised by aggregates of toxic $\tau P$ depositing in the brain in a relatively orderly fashion through six specific brain areas, called the \textit{Braak regions} \cite{braak1991neuropathological,kuhl2021pet}.  However, not all patients exhibit this type of toxic $\tau P$ progression.  Other forms of AD, called \textit{AD subtypes} \cite{murray2011, vogel2021}, are characterised by a different temporal, regional patterns of toxic $\tau P$ aggregates depositing within the brain.  Due to the differences in toxic $\tau P$ progression, each subtype differs in the way in which AD manifests as cognitive decline.  

A simple, but fundamental, question we can ask is: might AD subtypes display \textit{topological} differences?  More specifically, does the pattern of toxic $\tau P$ progression, how toxic $\tau P$ moves from one brain region to another, using the connective brain axonal fibers, exhibit different topological features from one AD subtype to the next?  A recent, large-scale patient imaging data study \cite{vogel2021} suggests that the AD subtype differences in the patterns of toxic $\tau P$ progression can be explained by differences in the origin of the toxic $\tau P$ infection.  The anatomical of the brain where toxic $\tau P$ originates in AD is typically referred to as the $\tau P$ \textit{seeding location}.  In the sections that follow, we apply a contemporary model of toxic $\tau P$ spreading through the brain.  As the model progresses, we track a marker of damage done to the brain by the toxic $\tau P$ which, in turn, we use to build a filtration from which we produce a maximal chain in $\mathcal{H}(N)$.  We will do this for each AD subtype and compare the differences between the corresponding chains in $\mathcal{H}(N)$.

 \subsection{Generating filtrations from graph neurodegeneration
 modelling}\label{subsec:application:filtration-gen}
In AD, toxic $\tau P$ moves from one anatomical brain region to the next and it does this using the brain's internal highway of \textit{axon bundles} connecting anatomically distinguished brain areas.  Toxic $\tau P$ traverses the brain's axon bundles from region to region and recruits the brain's existing population of healthy $\tau P$ to a toxic form.  This process has been mathematically modelled, by several authors, as a dynamical system posed on a special graph that represents the brain.  In this section, we describe the brain graph, the mathematical model of $\tau P$ spreading and the subsequent generation of a filtration, from the $\tau P$ spreading pattern, for each of four AD subtypes.

\subsubsection{The structural connectome}

Neurons, the primary brain cell that we use for processing electrical impulses, project a long, thin nerve fibre, called an axon, that they use to exchange electrical signals with other neurons.  In the brain, axon fibres from many neurons can join together to form a bundle, similar to how a rope is a bundle made up of numerous individual fibres.  The axon bundles stretch across the brain and provide a highway for the exchange of electrical impulses from one anatomical region of the brain to another.  This brain highway network of axonal bundles heavily biases the flow of watery fluid within the brain.  Two sets of medical images can be gathered from a patient: the first set of images (T1/T2) can be used to scan the brain's anatomical structure; the second set of images (DTI) gives information about the way water flows in the brain.  

Powerful, open-source scientific software can be used to transform patient images into a digital representation of a network graph that can be used to mathematically model the movement of toxic $\tau P$ in AD.  Software, called FreeSurfer \cite{dale1999,desikan2006}, can use the T1/T2 images to group together, or \textit{parcellate}, a patient's brain scan into its respective anatomical regions.  The set of DTI images can be processed by additional software packages, such as FSL \cite{jenkinson2012} or MRtrix \cite{tournier2019},  to construct approximations to the connectivity of the brain's axon bundle network.  Once a patient's brain regions are labelled and the connectivity between regions is known, this information defines a graph, $G=(V,E)$,  whose set of vertices, $V$, correspond to the anatomically labeled brain regions and whose set of edges, $E$, correspond to the axonal bundles connecting these regions.  Such a graph is referred to as a \textit{structural brain connectome}.

In this manuscript, we will use freely available structural brain connectomes based on patient data from the Human Connectome Project \cite{daducci2012,kerepesi2017}.  These graphs come with additional information about each edge (axon bundle) connecting two brain regions.  Using this information, we can form a canonical \textit{weighted adjacency matrix} for the graph which has entries defined by
\begin{equation}\label{eqn:weighted-adjacency-matrix}
\omega_{ij}=\frac{n_{ij}}{\ell_{ij}^2},\qquad i,j=1,\ldots,N,
\end{equation}
where $n_{ij}$ is the number of axon fibres in the (edge) axon bundle and $\ell_{ij}$ is the total length of the bundle connecting (vertex) brain region $i$ to (vertex) brain region $j$.  

The basic structural connectome has $|V|=N=83$ and is shown in Fig.~\ref{fig:connectomes-33}.  However, due to its size, working with $\mathcal{H}(N)$ is computationally intractable without further research into more optimal computer algorithms.  To work around this issue, we will use a selection of nested subgraphs of Figure~\ref{fig:connectomes-33} that range from $N=4$ vertices to $N=18$ vertices and all reside within the left hemisphere.  As we will discuss in Sec.~\ref{subsec:math-model-and-filtration}, there are four brain regions (vertices) that will need to be included in each graph, these are: the left entorhinal cortex (vertex $\tilde{v}_1$), the middle temporal gyrus (vertex $\tilde{v}_2$), the fusiform gyrus (vertex $\tilde{v}_3$) and the inferior temporal gyrus (vertex $\tilde{v}_4$).  To proceed, we constructed five nested, left hemisphere subgraphs of the structural connectome graph $G$, shown in Fig.~\ref{fig:connectomes-33}, denoted by $G_4 \subset G_6 \subset G_8 \subset G_{12} \subset G_{15}\subset G_{18}$, containing $4$, $6$, $8$, $12$, $15$ and $18$ vertices, respectively.  The nested subgraphs were constructed by: starting with the vertex set $\left\{\tilde{v}_1,\tilde{v}_2,\tilde{v}_3,\tilde{v}_4\right\}$ corresponding to the anatomical regions discussed above; considering the neighbours of the initial vertices and selecting, for each vertex, those neighbours whose edge weights $w_{ij}$ lie in the top 5\% of all left hemisphere vertex neighbours and adding these vertices to the set; this process was repeated until a vertex threshold for the subgraph was reached, resulting in a set of vertices $V_k$; the subgraph $G_k$ 
is then $G_k = (V_k, E_k)$ where $V_k$ is the constructed vertex set and $E_k$ consists of all edges in the original connectome graph (Fig.~\ref{fig:connectomes-33}) between any two vertices in $V_k$.  Since we begin with the initial set of vertices $\left\{\tilde{v}_1,\tilde{v}_2,\tilde{v}_3,\tilde{v}_4\right\}$ each time, this method of subgraph construction assures that $G_i \subset G_j$ whenever $i \leq j$.  The first five of these nested subgraphs are shown in Fig.~\ref{fig:connectome-33-subgraphs}.

\begin{remark}
{{
Structural connectomes, as in Figure~\eqref{fig:connectomes-33}, arise from applying algorithmic techniques to patient diffusion magnetic resonance imaging data.  Several factors can influence the final result, including the type of scanner used to collect the data, the choice of parcellation for the connectome graph, the algorithm used to carry out the structural tractography on the resulting patient image and the choice of thresholding method.  As far as the authors are aware, there is currently no `gold standard' approach for constructing a human brain connectome.}}
\end{remark}

\begin{figure}[]
  \centering
  \includegraphics[width=.7\textwidth]{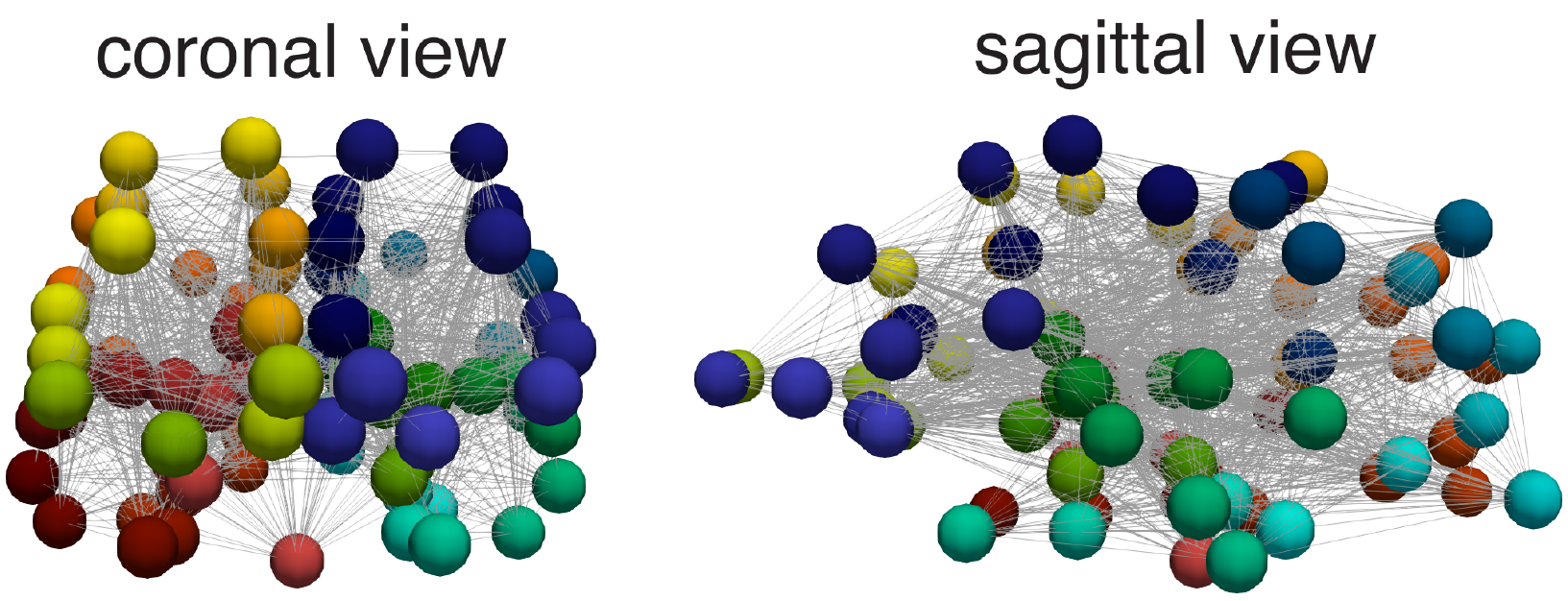}
    \centering
\caption{A structural connectome, with $N=83$, constructed from human magnetic resonance images.  Vertices represent cortical areas and edges represent bundles of axons connecting these cortical areas; individual vertices are colored by their classification into 83 distinct anatomical regions.}
\label{fig:connectomes-33}
\end{figure}

\begin{figure}
     \centering
     \begin{subfigure}[b]{0.2\textwidth}
         \centering
         \includegraphics[width=\textwidth]{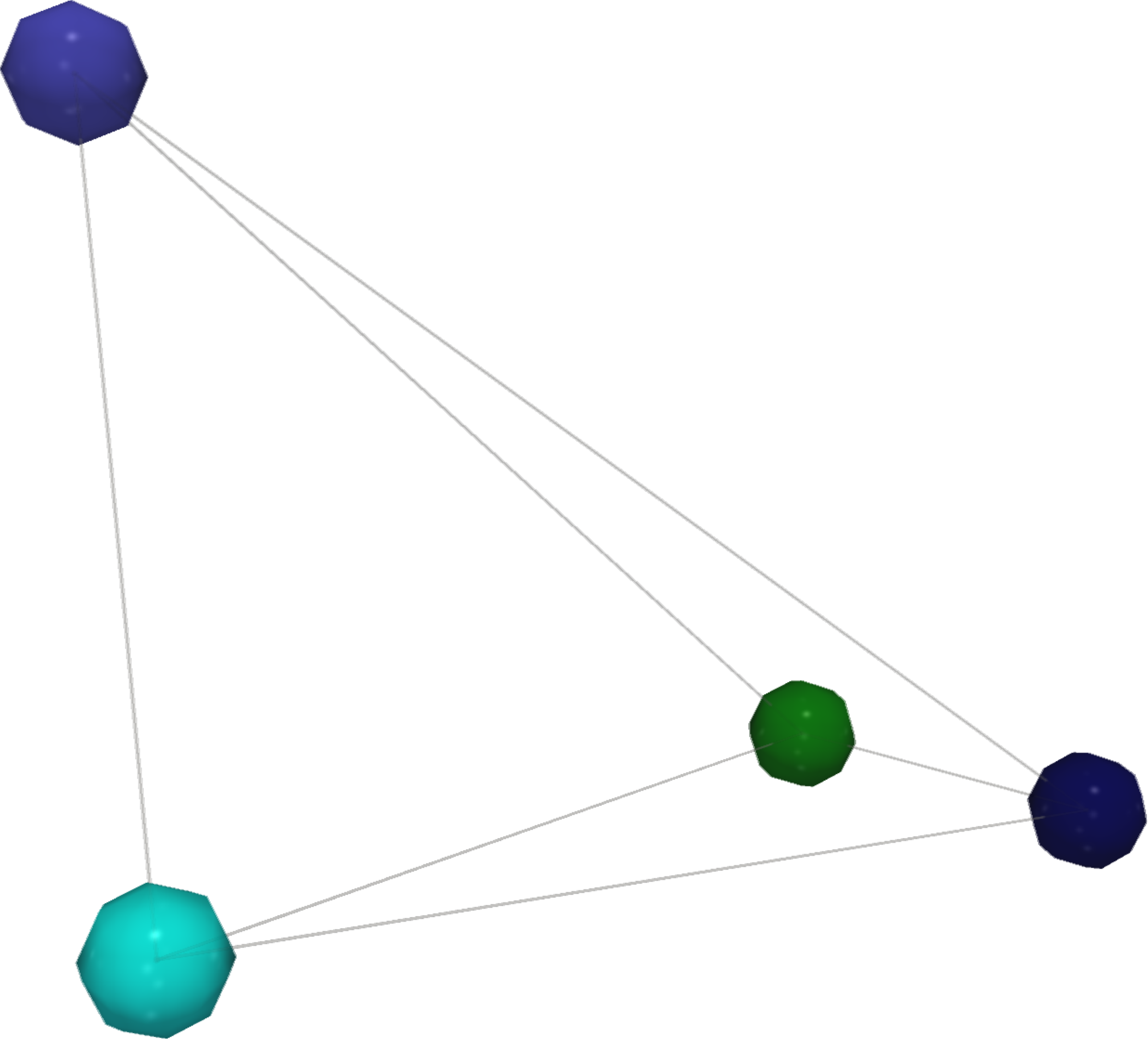}
         \caption{$N=4$}
         \label{fig:connectome-33-subgraphs:4}
     \end{subfigure}
     \hfill
     \begin{subfigure}[b]{0.2\textwidth}
         \centering
         \includegraphics[width=\textwidth]{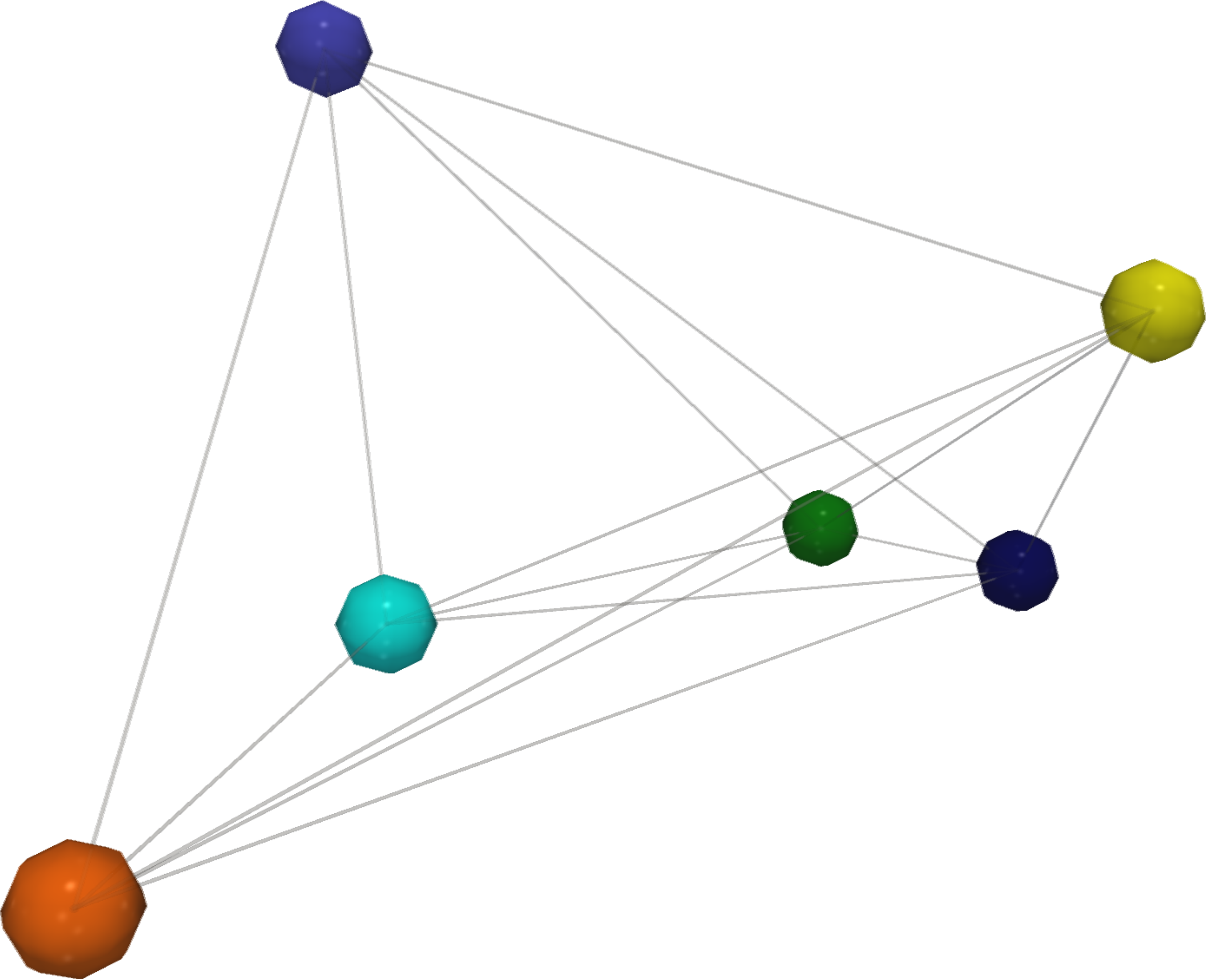}
         \caption{$N=6$}
         \label{fig:connectome-33-subgraphs:6}
     \end{subfigure}
     \centering
     \begin{subfigure}[b]{0.15\textwidth}
         \centering
         \includegraphics[width=\textwidth]{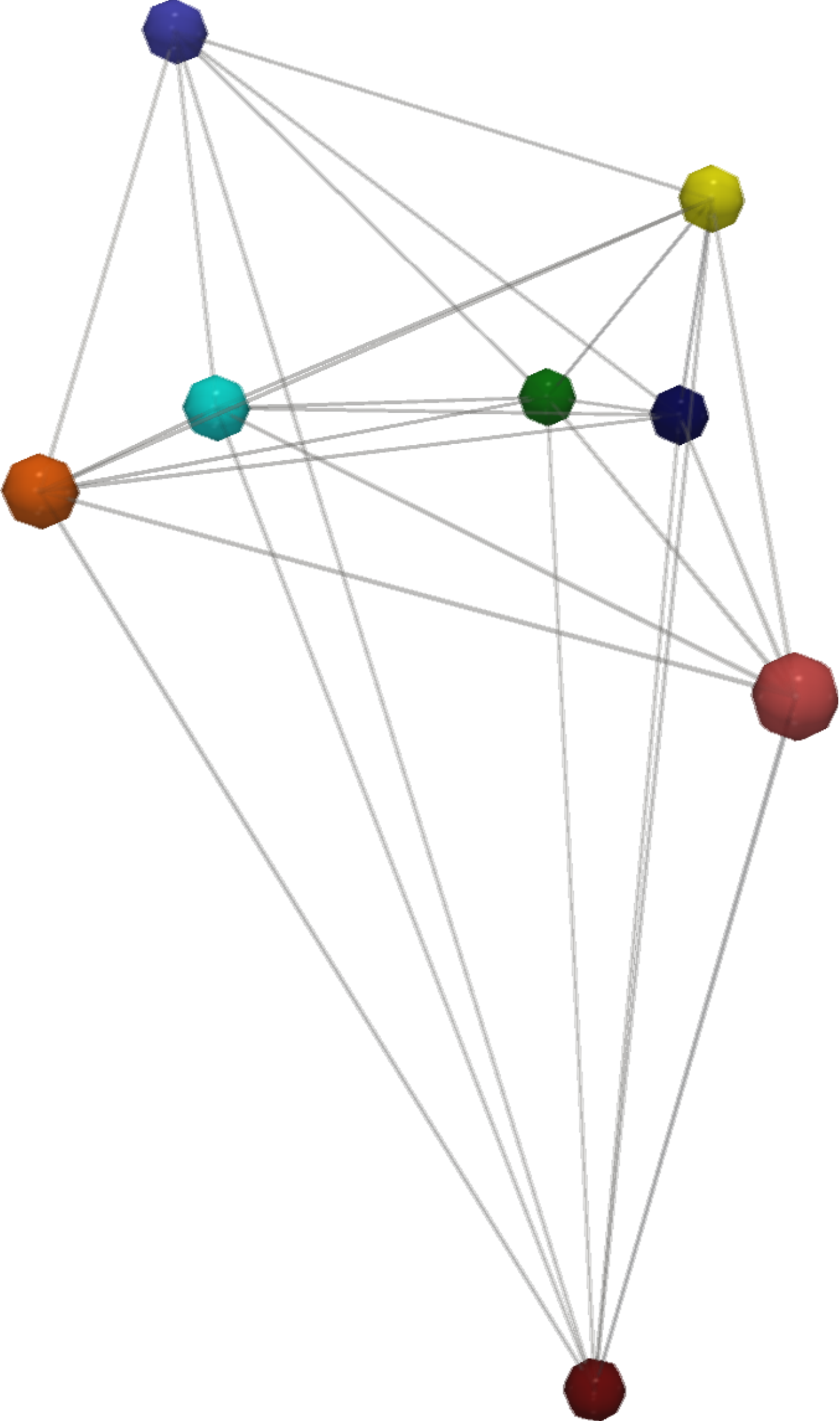}
         \caption{$N=8$}
         \label{fig:connectome-33-subgraphs:8}
     \end{subfigure}
     \hfill
     \begin{subfigure}[b]{0.15\textwidth}
         \centering
         \includegraphics[width=\textwidth]{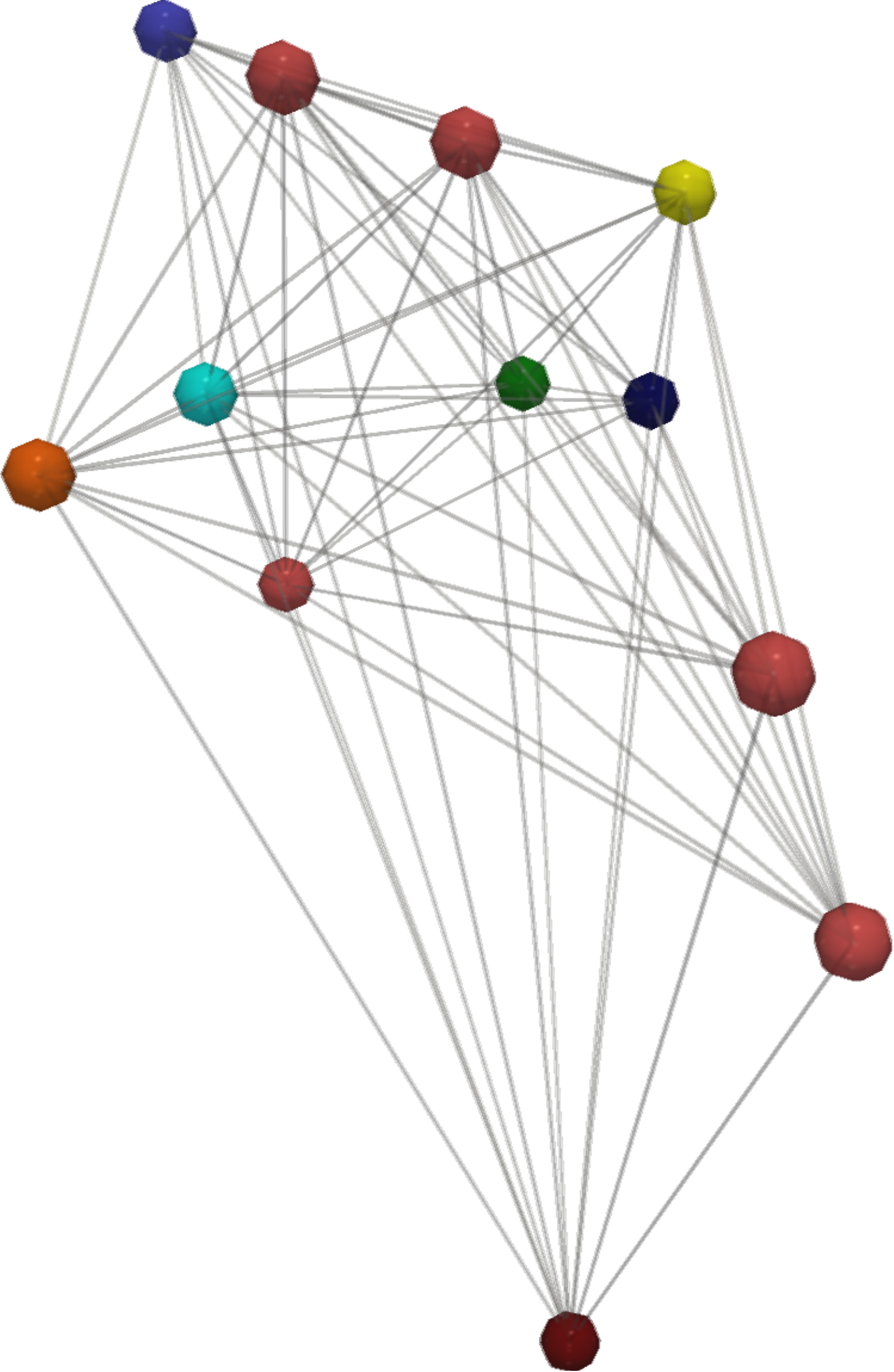}
         \caption{$N=12$}
         \label{fig:connectome-33-subgraphs:12}
     \end{subfigure}
    \hfill
     \begin{subfigure}[b]{0.15\textwidth}
         \centering
         \includegraphics[width=\textwidth]{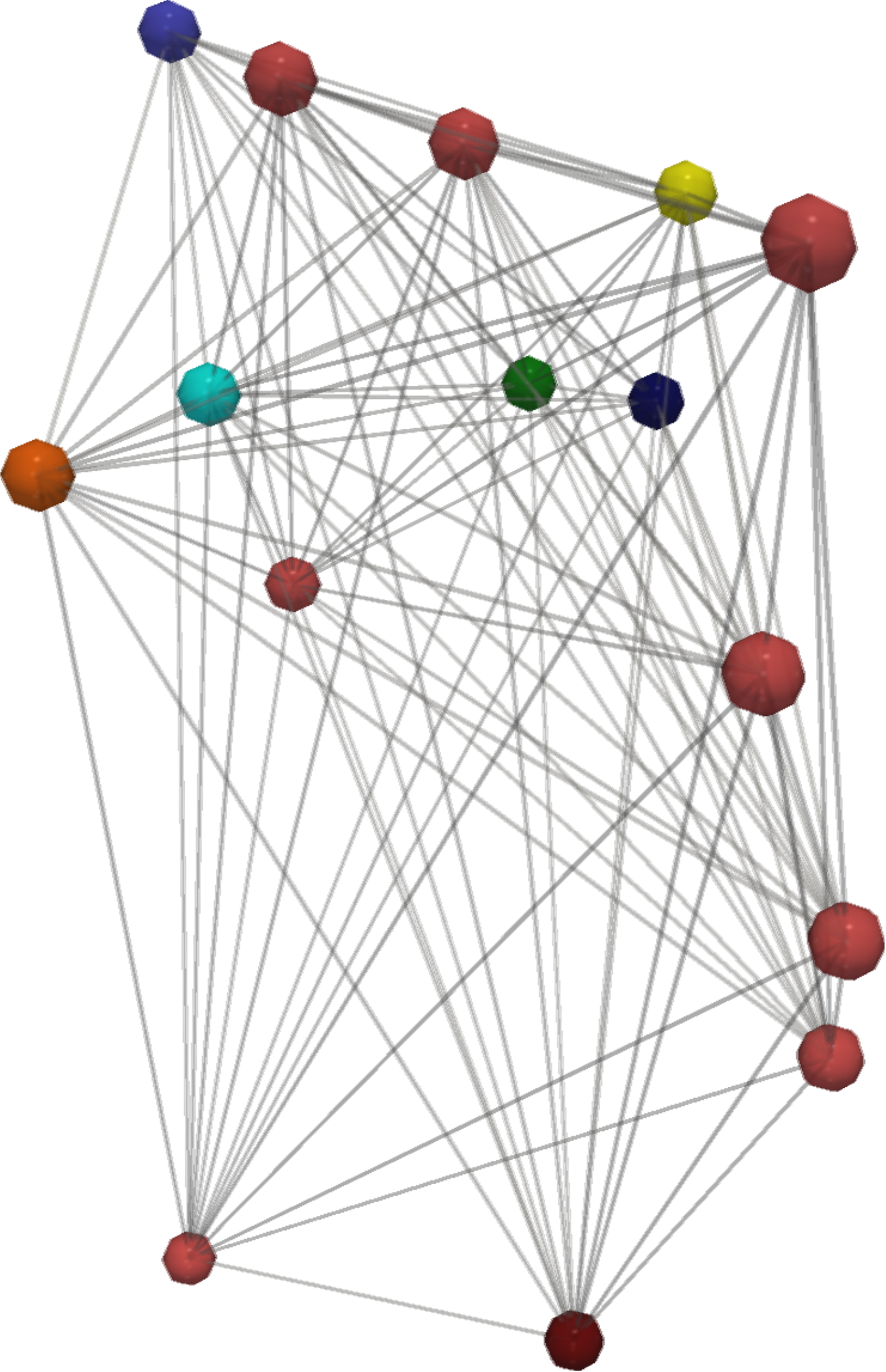}
         \caption{$N=15$}
         \label{fig:connectome-33-subgraphs:15}
     \end{subfigure}
        \caption{Nested brain region subgraphs extracted from the left hemisphere of the full connectome shown in  Figure~\ref{fig:connectomes-33}. The left entorhinal cortex and left middle temporal gyrus, respectively, correspond to the top left (royal blue) and bottom right (dark blue) vertices, respectively, shown in Figure~\ref{fig:connectome-33-subgraphs:4}.}
        \label{fig:connectome-33-subgraphs}
\end{figure}

 \subsubsection{A model of toxic $\tau P$ progression in AD and a filtration reflecting neurodegeneration}\label{subsec:math-model-and-filtration}

In this section, we introduce a network mathematical model of $\tau P$ progression in AD and discuss the construction of a filtration based on the neurodegenerative effects of toxic $\tau P$ on the brain.  As previously discussed, the general dynamics of toxic $\tau P$ in AD are: spreading from region to region using the brain's network of axon fibre bundles; and the inducing of the toxic state, in the presence of toxic $\tau P$, in otherwise healthy $\tau P$ \cite{jucker2013}.  The model introduced in this section also includes the neurodegenerative effects of the toxic $\tau P$ on the brain's axonal fibres.\\

\textbf{A parsimonious mathematical model of toxic $\tau P$ neurodegeneration in AD} The study of mathematical models on networks describing the propagation of toxic proteins within the brain in the context of dementia, like AD, is a mature field with over a decade of history \cite{sveva2,fornari2019,pandya2017brain,goriely2022,raj2015network,raj2012network,wang2017brain,weickenmeier2018multiphysics}.  High-dimensional models have been used to study the specifics toxic proteins may aggregation within the brain \cite{sveva2,fornari2019,gorielyclearance2021} while simpler Fisher-type network models,  approximating the dynamics undergirding the Nobel-award winning discovery of prion disease mechanisms \cite{prusiner1998}, have been used to study patient imaging data \cite{jack2013,kuhl2021pet,kuhl2020pet,kuhl2021bpet} and are now well established in the literature.  Our aim is to investigate the nascent potential for considering topological differences between AD subtypes and we select a parsimonious Fisher-type model, published previously \cite{goriely2020neuronal}, to do so. The model of \cite{goriely2020neuronal} reads:

\begin{subequations}\label{eqn:fkpp-edge-degrad}
\begin{align}
&\dot c_i = \rho \sum_{j=1}^N \left(w_{ij}-\delta_{ij} 
\sum_{k=1}^{N} w_{ik}
\right)c_j+\alpha c_i (1-c_i),&  i&=1,\dotsc,N,\label{eqn:fkpp-edge-degrad:c}\\
&\dot q_i =\beta c_i (1-q_i),\quad & i&=1,\dotsc, N, \label{eqn:fkpp-edge-degrad:q}\\
&\dot w_{ij} = -\gamma w_{ij} (q_i+q_j), & i,j&=1,\dotsc, N,\label{eqn:fkpp-edge-degrad:w}
\end{align}
\end{subequations}
where $\alpha,\beta,\gamma,\rho$ are parameters. For the initial conditions we have 
\begin{subequations}\label{eqn:fkpp-edge-degrad-ic}
\begin{align}
& c_i(0)=\epsilon \delta_{is} &i=1,\dotsc, N,\label{eqn:fkpp-edge-degrad-ic:c}\\
&q_i (0)=0, &i=1,\dotsc, N,\label{eqn:fkpp-edge-degrad-ic:q}\\
&w_{ij}(0) =\omega_{ij}, & i,j=1,\dotsc, N\label{eqn:fkpp-edge-degrad-ic:w},
\end{align}
\end{subequations}

The primary model, described mathematically by \eqref{eqn:fkpp-edge-degrad}, expresses three main attributes: the movement of toxic $\tau P$ from one brain region to the next; the conversion of healthy $\tau P$ to toxic $\tau P$ in the presence of toxic $\tau P$; and the damaging, neurodegenerative effects of toxic $\tau P$ on the brain.  
At each vertex, the concentration of toxic $\tau P$ is denoted by $c_i$.  The transport of toxic $\tau P$ from one brain region (vertex) to a connected neighbour is  modelled using the usual weighted \textit{graph Laplacian} tensor whose components are 
\begin{equation}
 L_{ij} = -w_{ij}+ \mbox{$\delta_{ij}\sum_{j=1}^{N}$} w_{ij},\qquad i,j\in\{1,\dots,N\},
 \end{equation}
where $w_{ij}$ is the corresponding entry in the weighted adjacency matrix for the graph being considered and $\delta_{ij}$ is the Kronecker delta ($\delta_{ij}=1$ for $i=j$ and zero otherwise).  The conversion of healthy $\tau P$ to toxic $\tau P$ in the presence of toxic $\tau P$ is modelled by the local logistic growth term $\alpha c_i(1-c_i)$, where $\alpha$ is a growth rate, in \eqref{eqn:fkpp-edge-degrad:c}.  This equation does not explicitly model the healthy $\tau P$ population but is derived from a model that does so \cite{fornari2019} and the saturation of toxic protein populations in AD has been observed in medical data \cite{jack2013}.  The trafficking of toxic $\tau P$ over the brain's axonal fibre bundle network  and the local reproduction of toxic $\tau P$ are expressed together in \eqref{eqn:fkpp-edge-degrad:c}, $\rho$ denotes a characteristic transport rate.

The neurodegenerative effect of toxic $\tau P$, on the various regions of the brain, is the final attribute reflected in \eqref{eqn:fkpp-edge-degrad}.  There are many harmful effects of toxic $\tau P$; the severity of harmful, degenerative effects correlates with the appearance of large aggregates of toxic $\tau P$ in AD.  The term $q_i$, in \eqref{eqn:fkpp-edge-degrad:q}, tracks the neurodegenerative damage done by the toxic $\tau P$ to the local brain region (graph vertex).  In the presence of toxic concentration $c_i$, $q_i$ increases from zero (healthy) to one (fully damaged).  Local damage to neurons within a brain region (graph vertex) causes those neurons and their axonal fibres to die.  The progression of local damage in the brain results in a degradation of the integrity of the edges (axonal fibre bundles) of the graph and is reflected in \eqref{eqn:fkpp-edge-degrad:w}. 

\begin{table}[t]
    \centering
    \caption{Alzheimer's disease subtypes, of \cite{vogel2021}, and their associated cortical $\tau$P seeding locations}
    \begin{tabular}{|c|c|c|}
    \hline
    Subtype name    &   Cortical seeding location  & Global vertex index\\
    \hline
    Limbic      &   Entorhinal cortex & $v_1$\\
    MTL sparing &   Middle temporal gyrus & $v_2$\\
    Posterior   &   Fusiform gyrus & $v_3$\\
    Temporal    &   Inferior temporal gyrus & $v_4$\\
    \hline
    \end{tabular}
    \label{tab:vogel-subtypes}
\end{table}

Various subtypes of AD have been identified in post-mortem \cite{murray2011} and imaging \cite{vogel2021} studies. It has been suggested \cite{vogel2021} that different subtypes of AD may arise from different locations where the initial toxic $\tau P$ population takes hold, or is \textit{seeded}.  The various AD subtypes, and their initial seeding locations, are summarised in Table~\ref{tab:vogel-subtypes}.  In the model, seeding corresponds to the initial conditions \eqref{eqn:fkpp-edge-degrad-ic}.  A specific brain region, corresponding to vertex $v_s$, is chosen as the seed region and \eqref{eqn:fkpp-edge-degrad-ic:c} indicates that only this region begins with a nonzero initial concentration, $0 < \epsilon \ll 1$, of toxic $\tau P$.  Varying the brain region corresponding to $v_s$, according to Table~\ref{tab:vogel-subtypes}, models the individual AD subtypes. The condition \eqref{eqn:fkpp-edge-degrad-ic:q} indicates that each brain region begins in an undamaged state while \eqref{eqn:fkpp-edge-degrad-ic:w} and \eqref{eqn:weighted-adjacency-matrix} fix the initial set of brain graph edges. 

\begin{remark}
{{
Propagation models, such as \eqref{eqn:fkpp-edge-degrad}, depend on the underlying structural connectome network, the choice of graph Laplacian weights and the parameters of the model.  The authors have previously studied propagation \cite{goriely2022} on a fixed connectome.  The results of that study indicate that the arrival of misfolded proteins at a target node are stable under small perturbations of the diffusion ($\rho$) and growth ($\alpha$) parameters of \eqref{eqn:fkpp-edge-degrad:c} in addition to the  initial seed value \eqref{eqn:fkpp-edge-degrad-ic}.  The question of differences in the evolution of trajectories under alterations of the underlying connectome graph, such as varying the tractography type, parcellation strategy, and edge weights, has also been previously investigated by the authors \cite{goriely2021}. Our findings in this study suggest that further efforts, from the imaging neuroscience community, are needed to standardise the construction of human structural brain connectomes and could greatly improve the interpretability, stability, and reproducibility of propagation models.}}
\end{remark}

 \begin{figure}[t]
     \centering
     \begin{subfigure}[b]{0.24\textwidth}
         \centering
         \includegraphics[width=\textwidth,trim=5 5 5 5,clip]{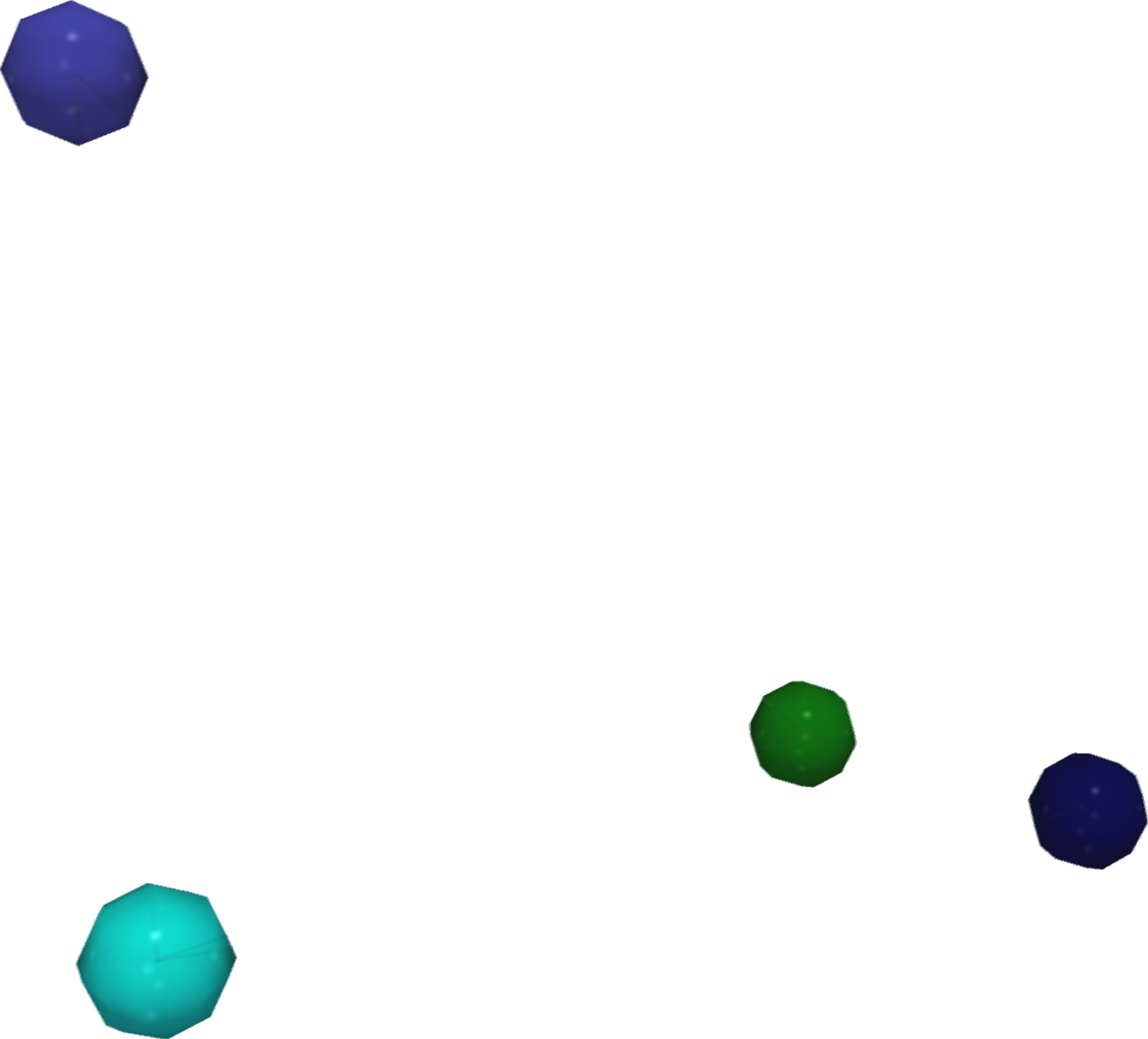}
         \caption{$S_0$}
     \end{subfigure}
     \hfill
     \begin{subfigure}[b]{0.24\textwidth}
         \centering
         \includegraphics[width=\textwidth,trim=5 5 5 5,clip]{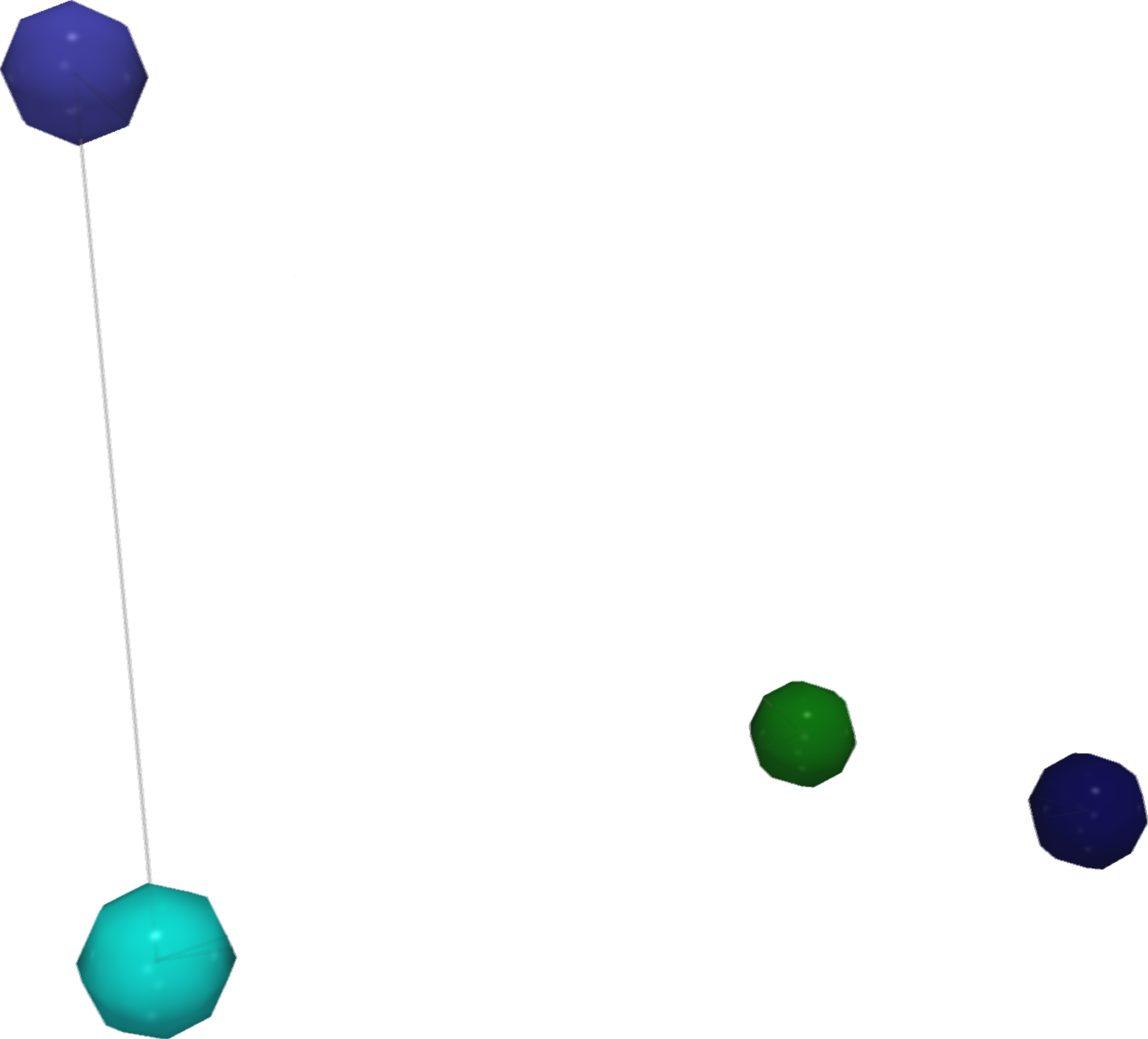}
         \caption{$S_1$}
     \end{subfigure}
     \hfill
     \centering
     \begin{subfigure}[b]{0.24\textwidth}
         \centering
         \includegraphics[width=\textwidth,trim=5 5 5 5,clip]{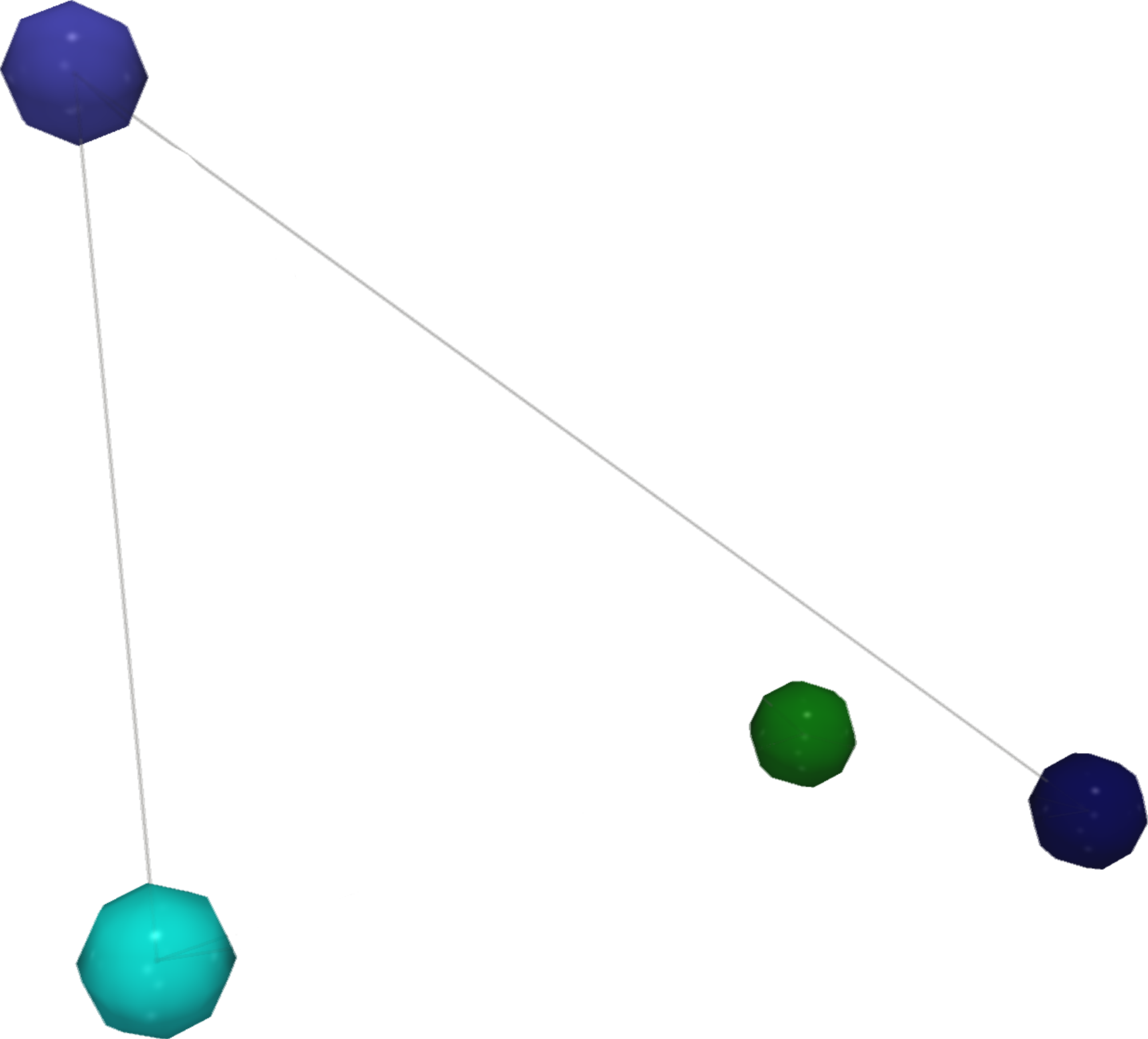}
         \caption{$S_2$}
     \end{subfigure}
     \hfill
     \begin{subfigure}[b]{0.24\textwidth}
         \centering
         \includegraphics[width=\textwidth,trim=5 5 5 5,clip]{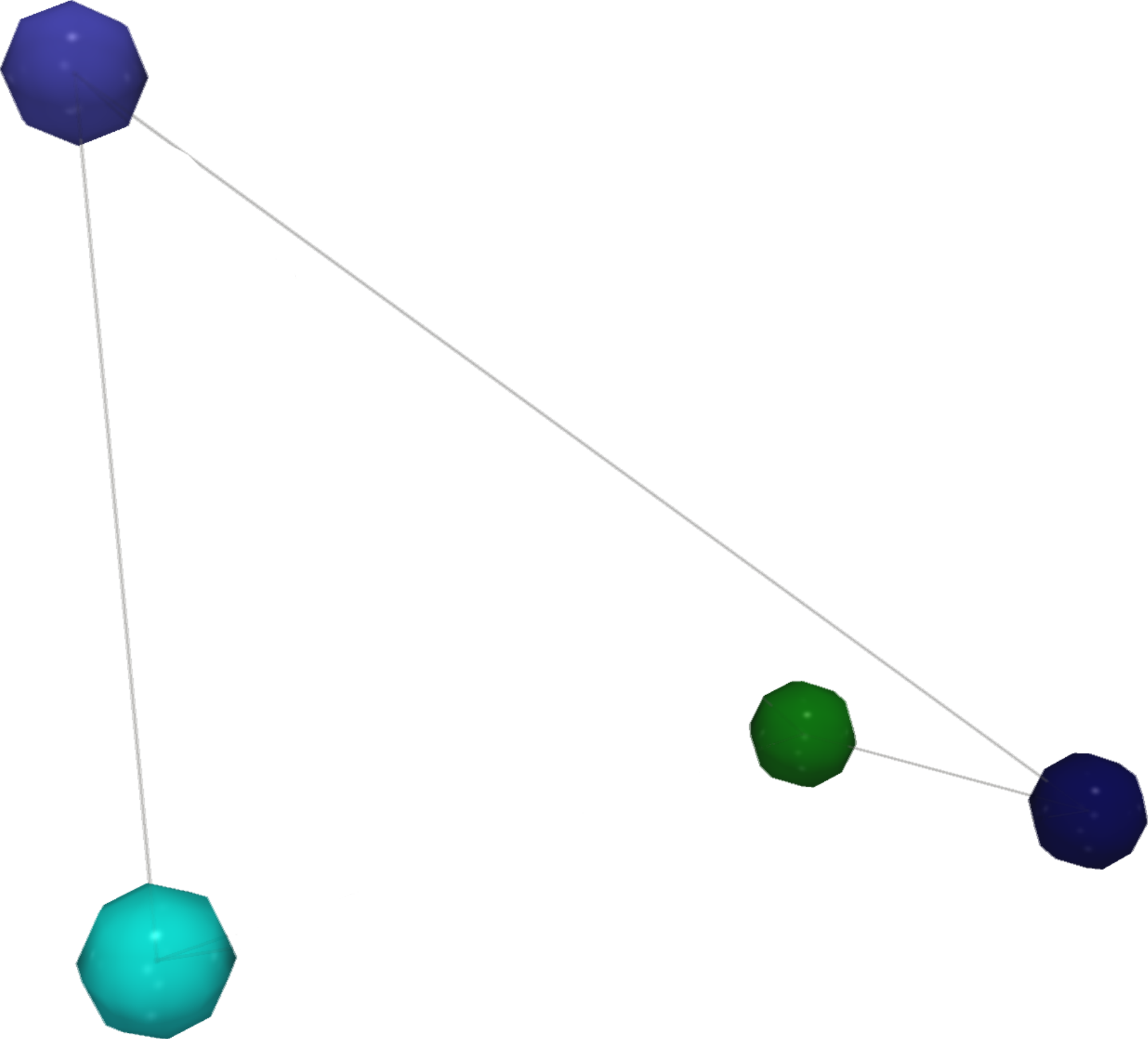}
         \caption{$S_3$}
     \end{subfigure}\hfill\\
          \centering
     \begin{subfigure}[b]{0.24\textwidth}
         \centering
         \includegraphics[width=\textwidth,trim=5 5 5 5,clip]{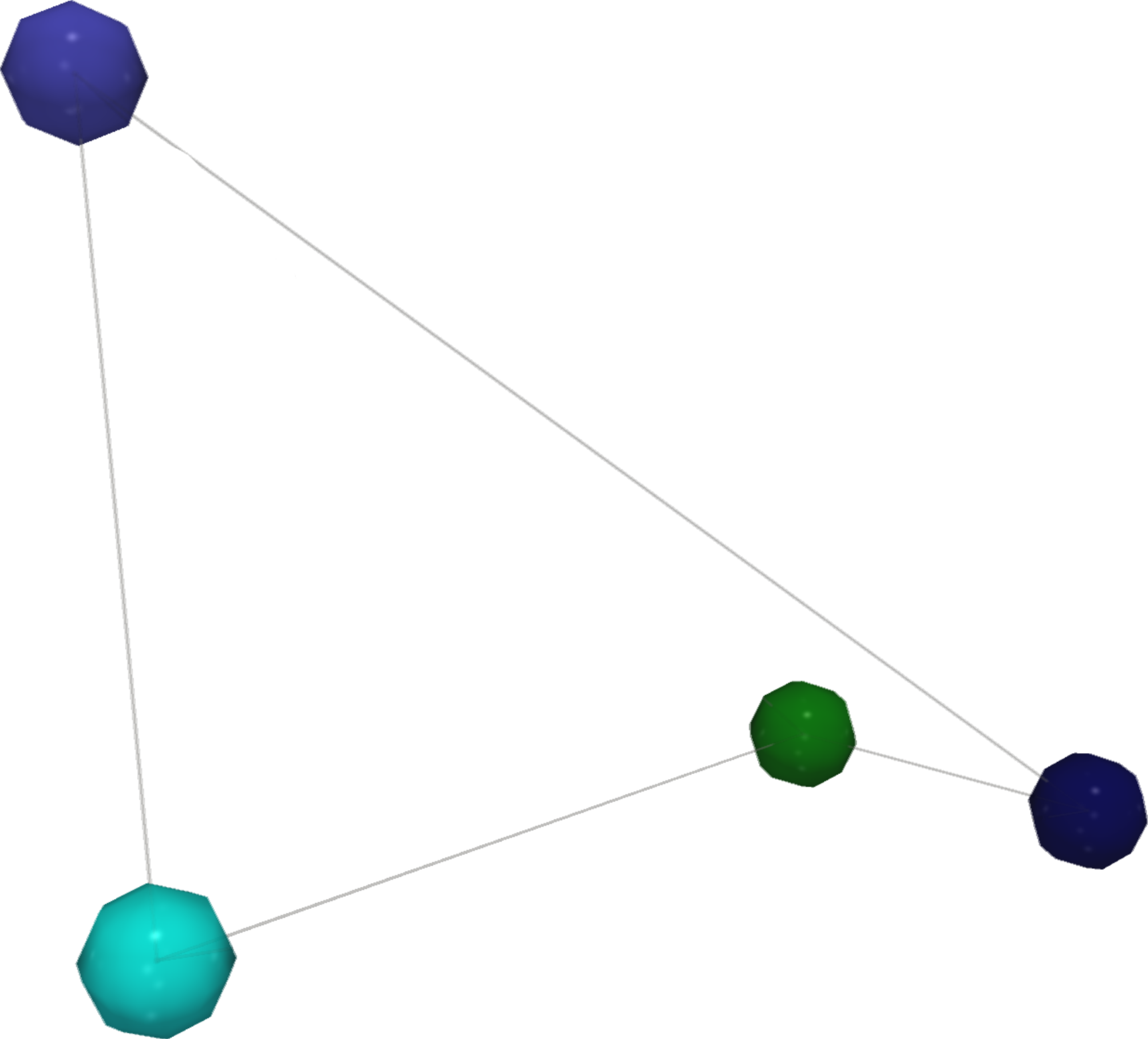}
         \caption{$S_4$}
     \end{subfigure}
    \hfill
     \begin{subfigure}[b]{0.24\textwidth}
         \centering
         \includegraphics[width=\textwidth,trim=5 5 5 5,clip]{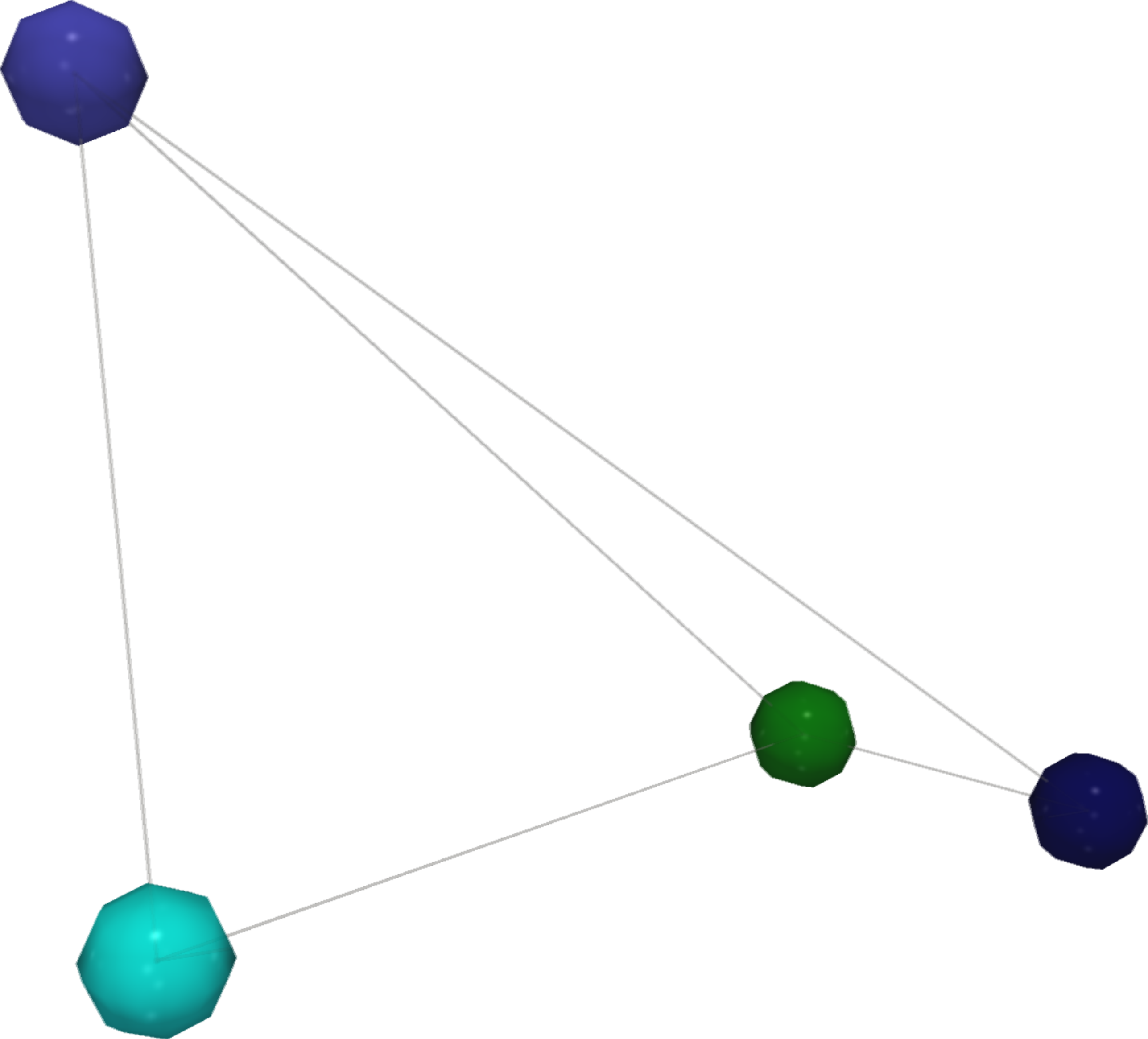}
         \caption{$S_5$}
     \end{subfigure}
     \hfill
     \begin{subfigure}[b]{0.24\textwidth}
         \centering
         \includegraphics[width=\textwidth,trim=5 5 5 5,clip]{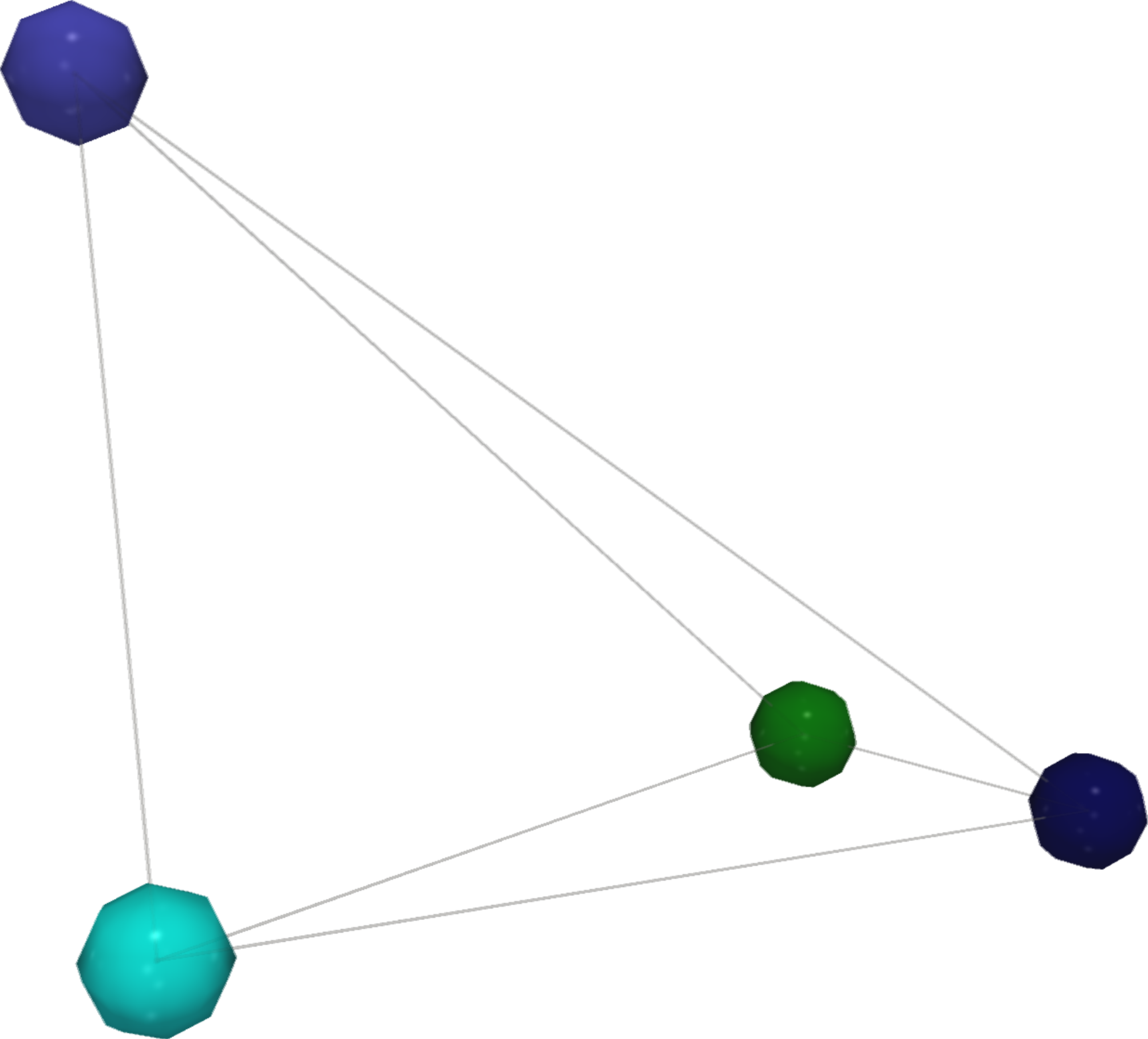}
         \caption{$S_6$}
     \end{subfigure}
    \hfill

        \caption{An illustrative edge filtration, $\mathcal{F}_s(\Xi_s)$, shown for the $N=4$ graph in Fig.~\ref{fig:connectome-33-subgraphs:4} with seeding in the entorhinal cortex (top left vertex).} 
        \label{fig:filtration-ex}
\end{figure}

\textbf{Filtrations from neurodegeneration}  In the model \eqref{eqn:fkpp-edge-degrad}-\eqref{eqn:fkpp-edge-degrad-ic}, the reduction in an edge weight, of the brain graph, is a metric for the extent of the neurodegeneration in the two brain regions connected by that edge.  As a result, the progression of toxic $\tau P$ proliferation and degeneration is commensurate with the order in which the brain graph edge weights degrade.  
From \eqref{eqn:fkpp-edge-degrad}, it can be shown \cite{goriely2020neuronal} that all of the edge weights, $w_{ij}$, converge asymptotically to zero and that late in the disease, around time 20-25 years, the relative values of edge weights are ordered, so that if $w_{ij}(t)>w_{i'j'}(t)$, then $w_{ij}(t')\geq w_{i'j'}(t')$ for all $t'\geq t>20$ and that this relationship holds for all edges of the graph. 

Let $G=(V,E)$ be a fixed choice of graph.  In practice $G$ will be one of the nested subgraphs shown in Fig.~\ref{fig:connectome-33-subgraphs}.  From the discussion above, selecting a seeding location $v_s$, with an initial seed $0 < \epsilon \ll 1$, and a time $t_0 \approx 20$ determines an ordering of the edges, via an ordering of the associated weights $w_{ij}(t)$, and this ordering is preserved for all times $t \geq t_0$.  An edge permutation corresponding to $G$ and the initial seed vertex $v_s$ is defined by $\Xi_s = \{ e_1, e_2, \dots e_M\}$, where $|E|=M$, with the property that if $i < j$ then the edge weight $w$ associated to edge $e_i = \Xi_s(i)$ and the edge weight $\tilde{w}$ associated to edge $e_j = \Xi_s(j)$ satisfy $w(t) \leq \tilde{w}(t)$ for all $t \geq t_0$.  Given $\Xi_s$ we define an edge filtration (Definition~\ref{defin:edge-filtration}), denoted by $\mathcal{F}_s(\Xi_s)$, of $G$ by
\begin{equation}
\mathcal{F}_s(\Xi_s) = S_0 \lessdot S_1 \lessdot \dots \lessdot S_M, 
\end{equation}
where $S_0 = (V,\{\})$ and defining $S_r = (V,\{\Xi_s(1),\Xi_s(2),\dots,\Xi_s(r)\})$ for each $1 \leq r \leq M$; an example is shown in Fig.~\ref{fig:filtration-ex}.  The edge filtration $\mathcal{F}_s(\Xi_s)$ defines a maximal chain in $\mathcal{H}(N)$ (c.f.~Sec.~\ref{sec:graph-filtrations}) which we will also denote by $\mathcal{F}_s$.

\subsection{$\mathcal{H}\left(N\right)$ distinguishes AD subtypes with increasing $N$}\label{subsec:HN-differentiated-subtypes}
In this section we examine the topological implications of the hypothesis \cite{vogel2021} that AD subtypes may arise from different toxic $\tau P$ seeding locations within the brain (c.f.~Table~\ref{tab:vogel-subtypes}).  We will show that as the number of vertices in the brain subgraphs (c.f.~Fig.~\ref{fig:connectome-33-subgraphs}) increase, from $N=4$ to $N=18$, we see differences emerge between the simulated AD subtypes.  Our conclusion is cautious and two-fold.  First, we assert that this work offers the first evidence that edge filtrations generated by network neurodegeneration dynamical systems on (structural connectome) brain graphs may be applicable in AD, or other neurological, research.  Second, we claim that considering edge filtrations as chains in $\mathcal{H}(N)$ offers at least one novel means of study. \\

\textbf{Simulation and edge filtration generation for each AD subtype}  We assume, as hypothesised by the authors of one of the largest recent brain imaging studies of AD subtypes \cite{vogel2021},  that each AD subtype of Table~\ref{tab:vogel-subtypes} corresponds to a choice of seeding region (i.e.~an initial vertex $v_s$ where \eqref{eqn:fkpp-edge-degrad-ic:c} is nonzero). For each AD subtype, the model \eqref{eqn:fkpp-edge-degrad}-\eqref{eqn:fkpp-edge-degrad-ic} was used to determine an edge filtration (Sec.~\ref{subsec:math-model-and-filtration}) for each of the nested subgraphs $G_N$ for $N\in\{4,6,8,12,15,18\}$ (c.f.~Fig.~\ref{fig:connectome-33-subgraphs}).  The model parameters used in each run were identical.  The parameters used for each run of \eqref{eqn:fkpp-edge-degrad}-\eqref{eqn:fkpp-edge-degrad-ic} were: $\alpha= 3/4$, $\beta = 1/4$, $\gamma = 1/8$, $\rho = 1/100$, and $\epsilon=1/20$.  These parameters come from a previous statistical inference \cite{kuhl2021pet} using patient data and give a typical time scale of 30 years for the disease and produce the Braak staging \cite{braak1991neuropathological} of the disease, a well-accepted regional pattern of deposition of toxic $\tau P$ aggregates from post-mortem studies,  as shown in \cite{goriely2022}, with entorhinal cortex initial seeding.\\

\textbf{$\mathcal{H}(N)$ suggests three AD subtype categories for sufficient N} For each value of $N$ we constructed $\mathcal{H}(N)$ using a Mathematica implementation of Algorithm~\ref{alg:algorithm-1} and Algorithm \ref{alg:algorithm-2} (c.f.~Sec.~\ref{sec:pe-quotient-construction}).  A maximal chain was then produced by determining the homotopy polynomial, in $\mathcal{H}(N)$, for each spanning subgraph of each simulated filtration (c.f.~Sec.~\ref{sec:graph-filtrations}).  
Results are shown in Fig.~\ref{fig:chain-results-limbic-mtl} and Fig.~\ref{fig:chain-results-post-temp}.  Graphs of $\mathcal{H}(N)$ are shown in blue while the specific chains corresponding to the edge filtration, resulting from the simulations, are highlighted in red.  The first result is that $\mathcal{H}(4)$ and $\mathcal{H}(6)$, corresponding to the subgraphs $G_4$ and $G_6$ (c.f.~Fig.~\ref{fig:connectome-33-subgraphs:4}-\ref{fig:connectome-33-subgraphs:6}), are insufficient to exhibit differences in the edge filtrations produced by the simulations.  At $N=8$, the maximal chain for the limbic subtype is visually distinct from that of the MTL, posterior and temporal subtypes (Fig.~\ref{fig:chain-results-limbic-mtl:8:limbic} vs.~both Fig.~\ref{fig:chain-results-limbic-mtl:8:mtl} and Figs.~\ref{fig:chain-results-post-temp:8:post}-\ref{fig:chain-results-post-temp:8:temp}).   
This same observation holds for $N=12$.  On the subgraph $G_{15}$ (c.f.~Fig.~\ref{fig:connectome-33-subgraphs:15}) and $G_{18}$ the simulated edge filtrations in $\mathcal{H}(15)$ appear to further differentiate the limbic and MTL subtypes from the posterior and temporal subtypes, the latter two of which show the same result; the graphs for $\mathcal{H}(18)$ are not shown in Figs.~\ref{fig:chain-results-limbic-mtl}-\ref{fig:chain-results-post-temp} due to their prohibitive size.\\

\begin{figure}
     \centering
     \begin{subfigure}[b]{0.49\textwidth}
         \centering
         \includegraphics[width=\textwidth]{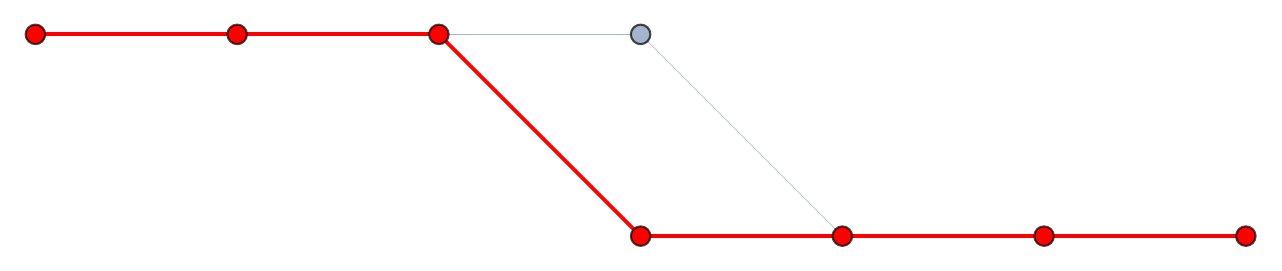}
         \caption{$N=4$, Limbic subtype, $\mathcal{F}_1$}
         \label{fig:chain-results-limbic-mtl:4:limbic}
     \end{subfigure}
     \hfill
     \begin{subfigure}[b]{0.49\textwidth}
         \centering
         \includegraphics[width=\textwidth]{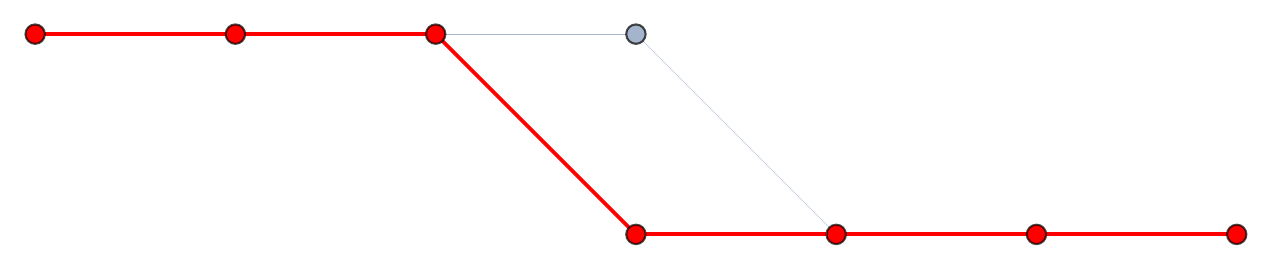}
         \caption{$N=4$, MTL subtype, $\mathcal{F}_2$}
         \label{fig:chain-results-limbic-mtl:4:mtl}
     \end{subfigure}\\
    \begin{subfigure}[b]{0.49\textwidth}
         \centering
         \includegraphics[width=\textwidth]{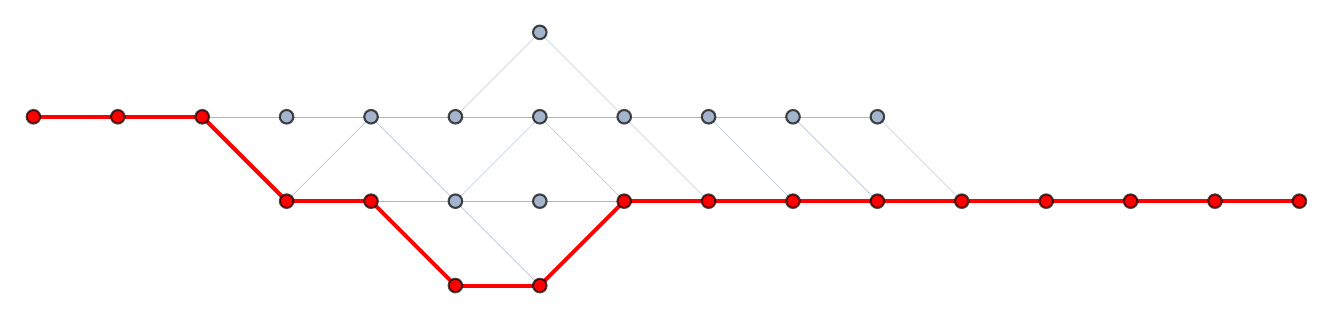}
         \caption{$N=6$, Limbic subtype, $\mathcal{F}_1$}
         \label{fig:chain-results-limbic-mtl:6:limbic}
     \end{subfigure}
     \hfill
     \begin{subfigure}[b]{0.49\textwidth}
         \centering
         \includegraphics[width=\textwidth]{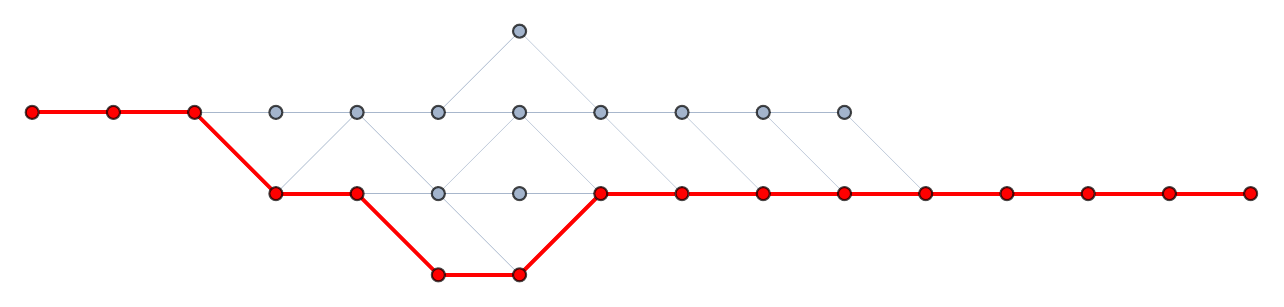}
         \caption{$N=6$, MTL subtype, $\mathcal{F}_2$}
         \label{fig:chain-results-limbic-mtl:6:mtl}
     \end{subfigure}\\
    \begin{subfigure}[b]{0.49\textwidth}
         \centering
         \includegraphics[width=\textwidth]{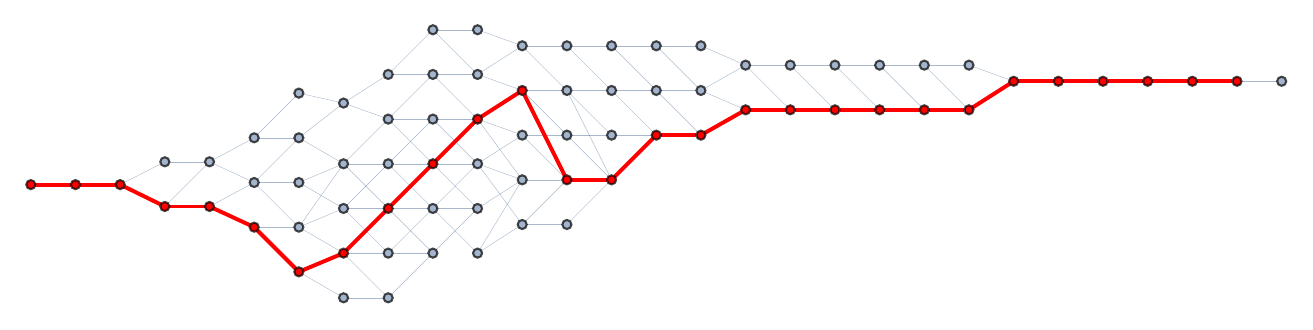}
         \caption{$N=8$, Limbic subtype, $\mathcal{F}_1$}
         \label{fig:chain-results-limbic-mtl:8:limbic}
     \end{subfigure}
     \hfill
     \begin{subfigure}[b]{0.49\textwidth}
         \centering
         \includegraphics[width=\textwidth]{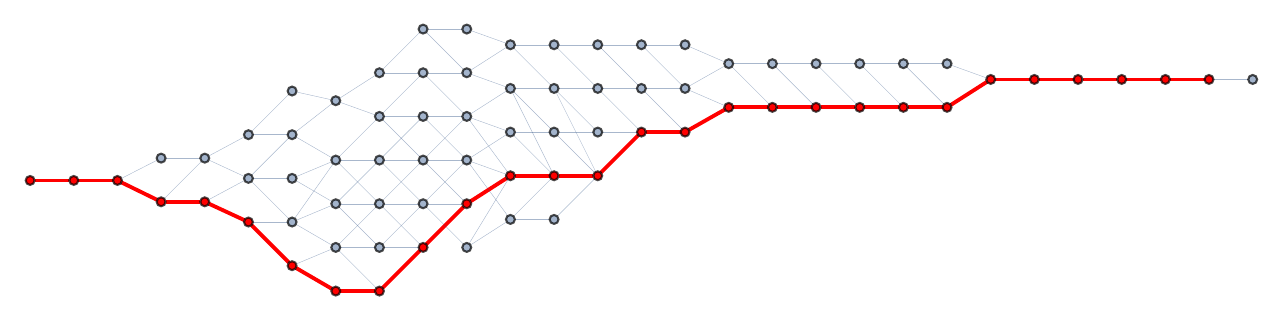}
         \caption{$N=8$, MTL subtype, $\mathcal{F}_2$}
         \label{fig:chain-results-limbic-mtl:8:mtl}
     \end{subfigure}\\
     \begin{subfigure}[b]{0.49\textwidth}
         \centering
         \includegraphics[width=\textwidth]{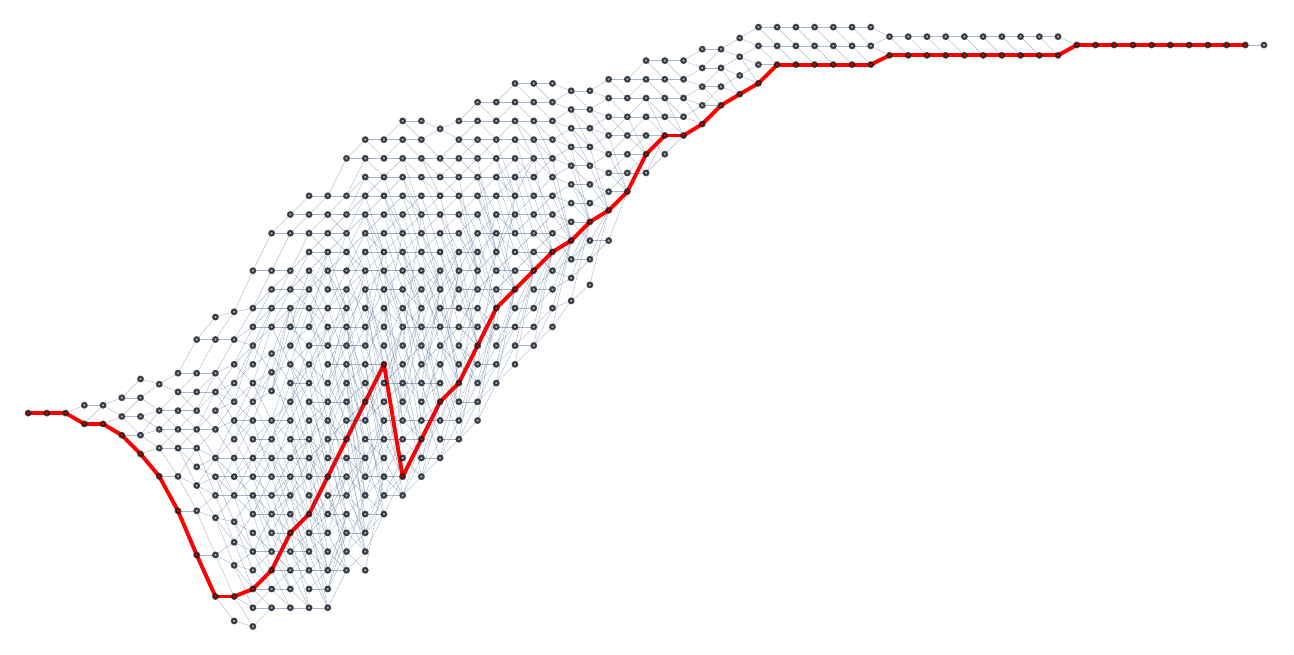}
         \caption{$N=12$, Limbic subtype, $\mathcal{F}_1$}
         \label{fig:chain-results-limbic-mtl:12:limbic}
     \end{subfigure}
     \hfill
     \begin{subfigure}[b]{0.49\textwidth}
         \centering
         \includegraphics[width=\textwidth]{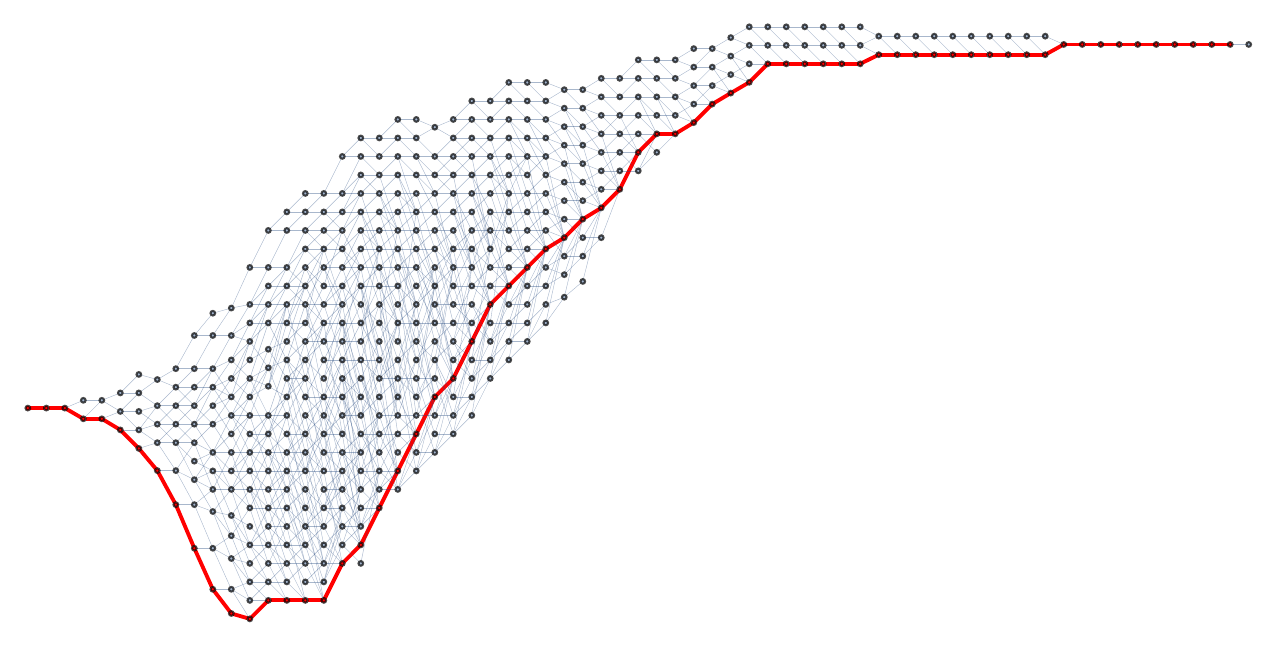}
         \caption{$N=12$, MTL subtype, $\mathcal{F}_2$}
         \label{fig:chain-results-limbic-mtl:12:mtl}
     \end{subfigure}\\
    \begin{subfigure}[b]{0.49\textwidth}
         \centering
         \includegraphics[width=\textwidth]{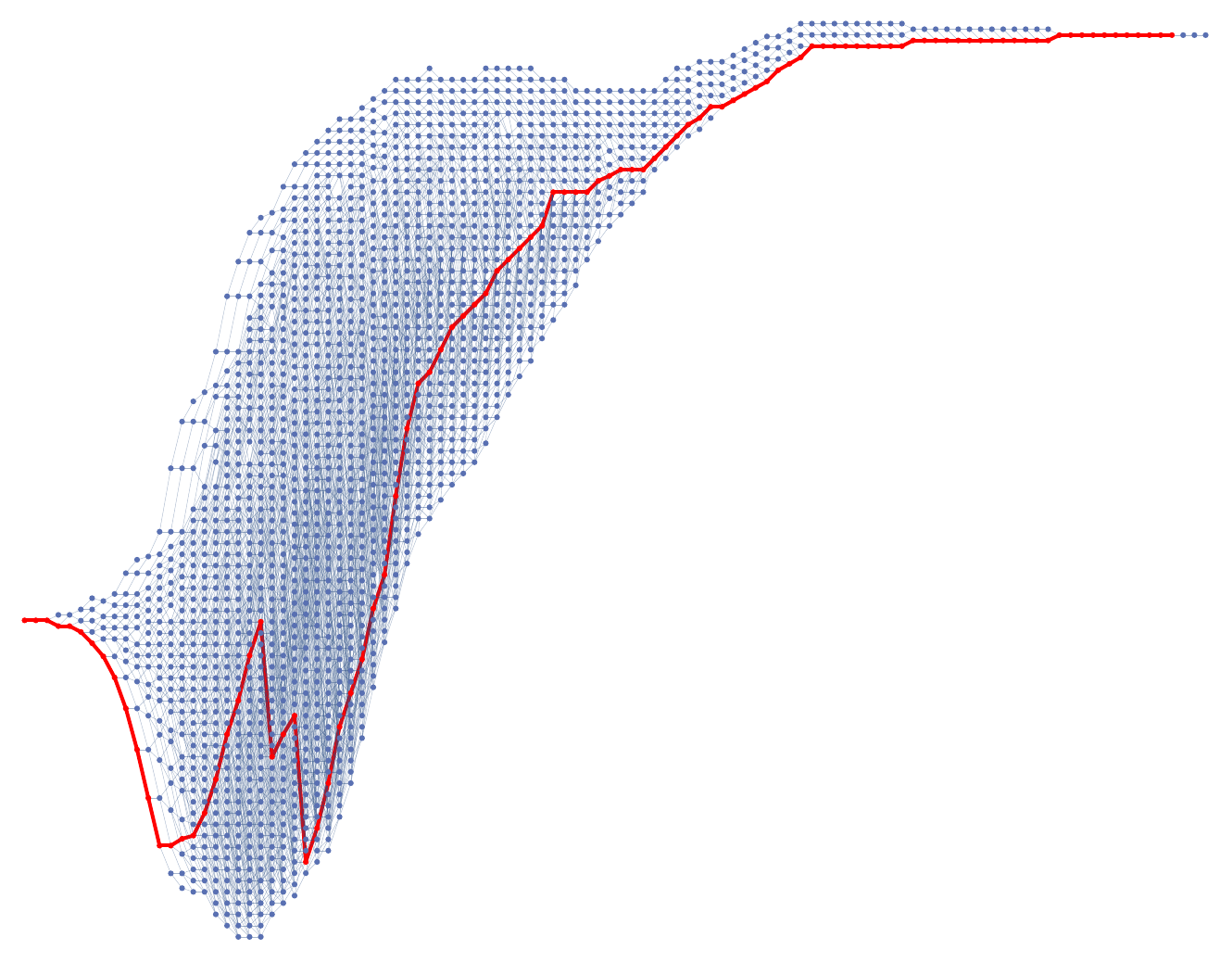}
         \caption{$N=15$, Limbic subtype, $\mathcal{F}_1$}
         \label{fig:chain-results-limbic-mtl:15:limbic}
     \end{subfigure}
     \hfill
     \begin{subfigure}[b]{0.49\textwidth}
         \centering
         \includegraphics[width=\textwidth]{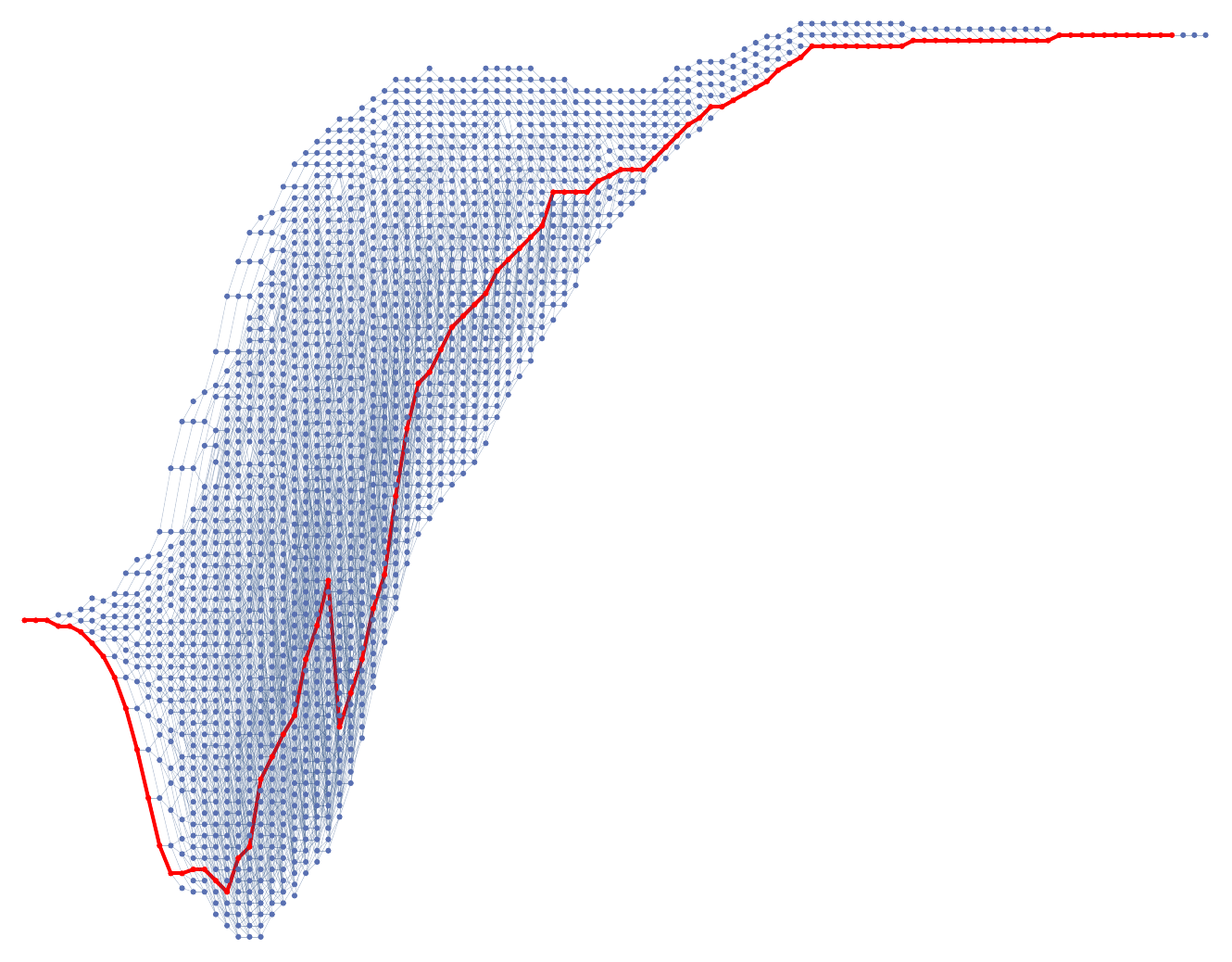}
         \caption{$N=15$, MTL subtype, $\mathcal{F}_2$}
         \label{fig:chain-results-limbic-mtl:15:mtl}
     \end{subfigure}\\
        \caption{Vertices are homotopy equivalence classes of spanning subgraphs of complete graphs of $N$ vertices. Chains (shown in red) corresponding to the limbic (left column) and MTL (right column) simulated subtypes of AD through the space, $\mathcal{H}\left(N\right)$, of homotopy polynomials for the nested left hemisphere graphs of Figure~\ref{fig:connectome-33-subgraphs}}
        \label{fig:chain-results-limbic-mtl}
\end{figure}

\begin{figure}
    \centering
    \begin{subfigure}[b]{0.49\textwidth}
         \centering
         \includegraphics[width=\textwidth]{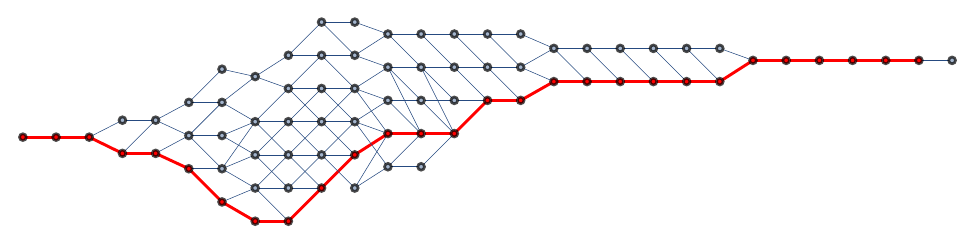}
         \caption{$N=8$, Posterior subtype, $\mathcal{F}_3$}
         \label{fig:chain-results-post-temp:8:post}
     \end{subfigure}
     \hfill
     \begin{subfigure}[b]{0.49\textwidth}
         \centering
         \includegraphics[width=\textwidth]{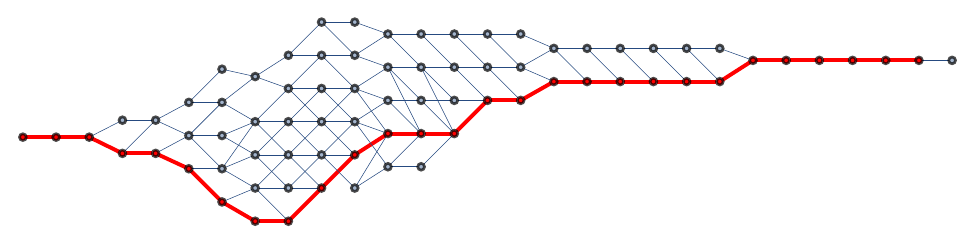}
         \caption{$N=8$, Temporal subtype, $\mathcal{F}_4$}
         \label{fig:chain-results-post-temp:8:temp}
     \end{subfigure}\\
    \begin{subfigure}[b]{0.49\textwidth}
        \centering
         \includegraphics[width=\textwidth]{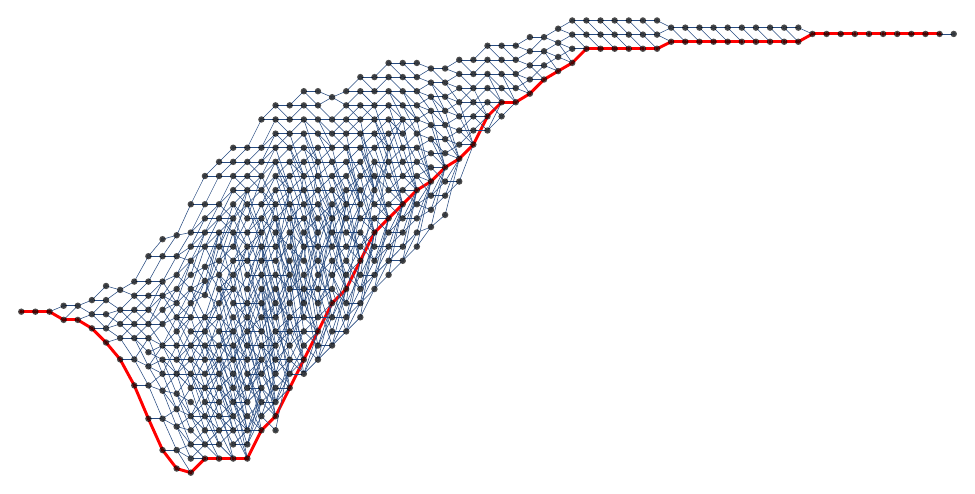}
         \caption{$N=12$, Posterior subtype, $\mathcal{F}_3$}
         \label{fig:chain-results-post-temp:12:post}
     \end{subfigure}
     \hfill
     \begin{subfigure}[b]{0.49\textwidth}
         \centering
         \includegraphics[width=\textwidth]{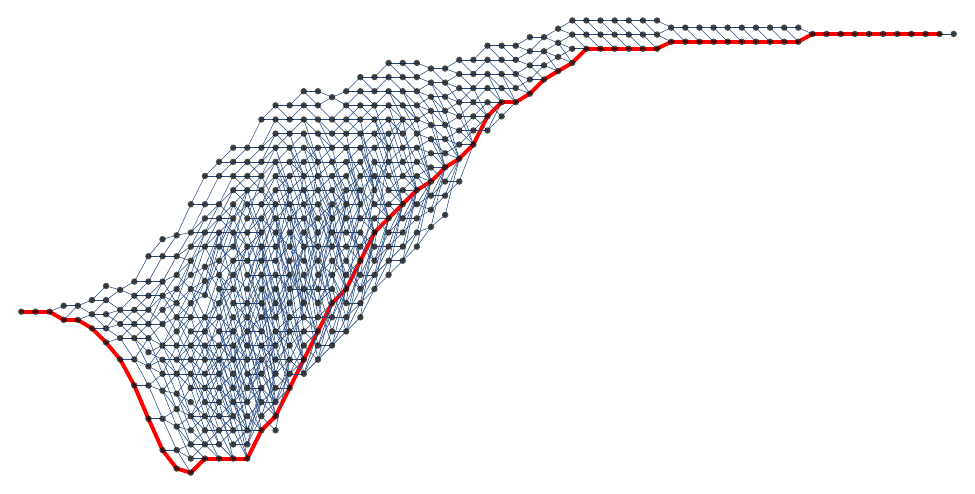}
         \caption{$N=12$, Temporal subtype, $\mathcal{F}_4$}
         \label{fig:chain-results-post-temp:12:temp}
     \end{subfigure}\\
    \begin{subfigure}[b]{0.49\textwidth}
         \centering
         \includegraphics[width=\textwidth]{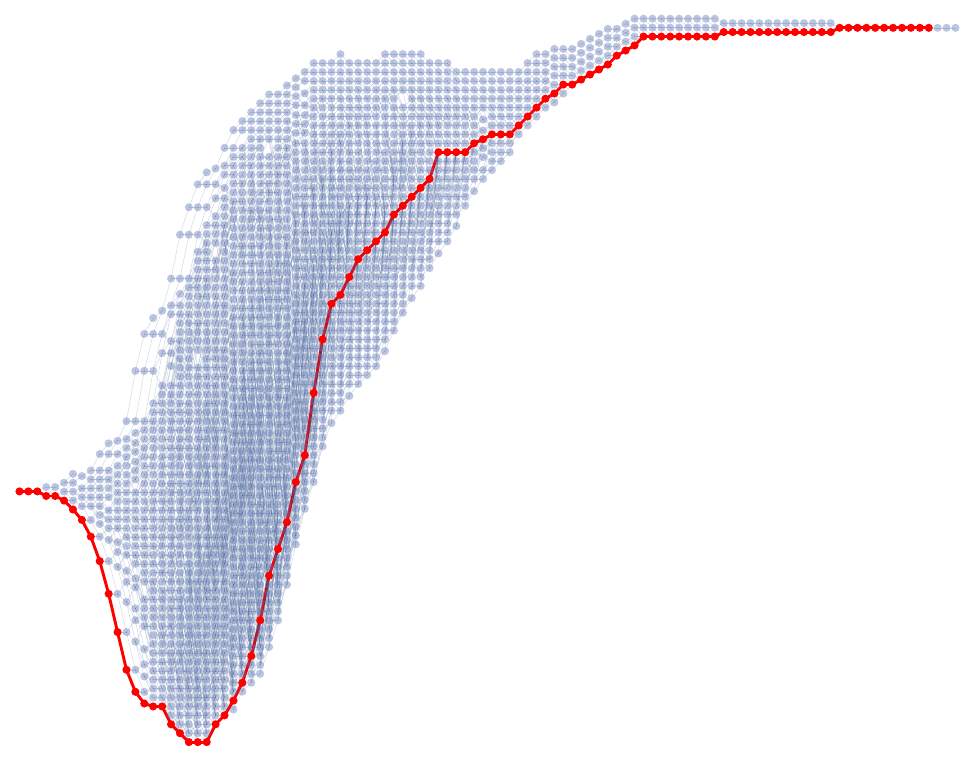}
         \caption{$N=15$, Posterior subtype, $\mathcal{F}_3$}
         \label{fig:chain-results-post-temp:15:post}
     \end{subfigure}
     \hfill
     \begin{subfigure}[b]{0.49\textwidth}
         \centering
         \includegraphics[width=\textwidth]{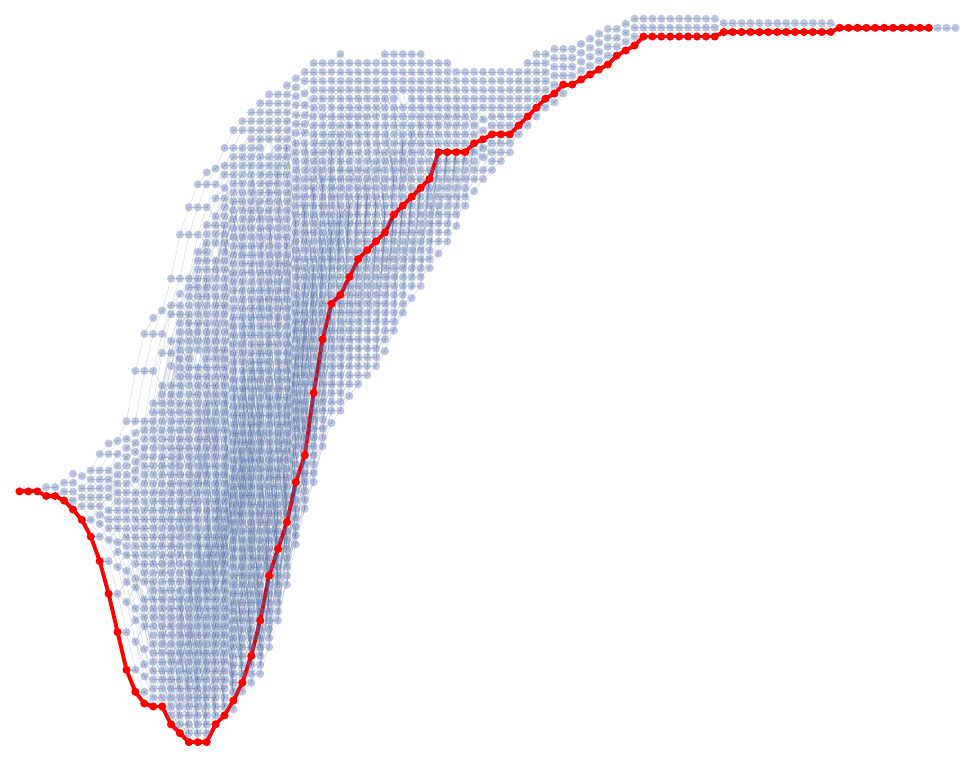}
         \caption{$N=15$, Temporal subtype, $\mathcal{F}_4$}
         \label{fig:chain-results-post-temp:15:temp}
     \end{subfigure}\\
        \caption{Chains corresponding to the posterior (left column) and temporal (right column) AD subtypes, as in Fig.~\ref{fig:chain-results-limbic-mtl}. Results corresponding to the brain subgraphs with $N=4,6$ (Fig~\ref{fig:connectome-33-subgraphs:4}-Fig~\ref{fig:connectome-33-subgraphs:6}) are identical to those of Fig~\ref{fig:chain-results-limbic-mtl:4:limbic}-Fig~\ref{fig:chain-results-limbic-mtl:6:mtl}.}
        \label{fig:chain-results-post-temp}
\end{figure}

\textbf{Computed discrete homotopy distances between chains corresponding to AD subtypes} For each maximal AD subtype chain in $\mathcal{H}(N)$, we computed the pairwise distance (c.f.~Defn. \ref{defn:sec:chain-compare:homotopy-metric}, Sec.~\ref{sec:meshes} and Sec.~\ref{sec:homotopy_vs_homology}) between every other maximal AD subtype chain for each fixed value of $N=8,12,15$ and $18$ (c.f.~Tables~\ref{tab:distances_n8}-\ref{tab:distances_n18}).  Computing the distance between two maximal chains is, in general, computationally prohibitive for large $N$.  However, we were able to find discrete homotopies whose length corresponded with the lower bounds computed through the homology poset (c.f.~Sec.~\ref{sec:meshes} and Sec.~\ref{sec:homotopy_vs_homology}); thus, the homology poset lower bounds provided exact distances, in these cases, and are tractable to compute even for large $N$.

\begin{table}[ht]
    \caption{Discrete homotopy distances between maximal chains $\mathcal{F}_j$ for $N=8$}
    \begin{center}
    \begin{tabular}{c|c|c|c|c|}
    \multicolumn{1}{c}{}& \multicolumn{1}{c}{Limbic ($\mathcal{F}_1$)} & \multicolumn{1}{c}{MTL ($\mathcal{F}_2$)} & \multicolumn{1}{c}{Posterior ($\mathcal{F}_3$)} & \multicolumn{1}{c}{Temporal ($\mathcal{F}_4$)}\\
    \cline{2-5}
    Limbic & 0 & 5 & 5 & 5\\
    \cline{2-5}
    MTL & 5 & 0 & 0 & 0 \\
    \cline{2-5}     
    Posterior & 5 & 0 & 0 & 0\\   
    \cline{2-5}     
    Temporal & 5 & 0 & 0 & 0\\
    \cline{2-5}
    \end{tabular}
    \end{center}
    \label{tab:distances_n8}
\end{table}

\begin{table}[ht]
    \caption{Discrete homotopy distances between maximal chains $\mathcal{F}_j$ for $N=12$}
    \begin{center}
    \begin{tabular}{c|c|c|c|c|}
    \multicolumn{1}{c}{}& \multicolumn{1}{c}{Limbic ($\mathcal{F}_1$)} & \multicolumn{1}{c}{MTL ($\mathcal{F}_2$)} & \multicolumn{1}{c}{Posterior ($\mathcal{F}_3$)} & \multicolumn{1}{c}{Temporal ($\mathcal{F}_4$)}\\
    \cline{2-5}
    Limbic & 0 & 9 & 9 & 9\\
    \cline{2-5}
    MTL & 9 & 0 & 0 & 0 \\
    \cline{2-5}     
    Posterior & 9 & 0 & 0 & 0\\   
    \cline{2-5}     
    Temporal & 9 & 0 & 0 & 0\\
    \cline{2-5}
    \end{tabular}
    \end{center}
    \label{tab:distances_n12}
\end{table}

\begin{table}[ht]
    \caption{Discrete homotopy distances between maximal chains $\mathcal{F}_j$ for $N=15$}
    \begin{center}
    \begin{tabular}{c|c|c|c|c|}
    \multicolumn{1}{c}{}& \multicolumn{1}{c}{Limbic ($\mathcal{F}_1$)} & \multicolumn{1}{c}{MTL ($\mathcal{F}_2$)} & \multicolumn{1}{c}{Posterior ($\mathcal{F}_3$)} & \multicolumn{1}{c}{Temporal ($\mathcal{F}_4$)}\\
    \cline{2-5}
    Limbic & 0 & 12 & 20 & 20\\
    \cline{2-5}
    MTL & 12 & 0 & 14 & 14 \\
    \cline{2-5}     
    Posterior & 20 & 14 & 0 & 0\\   
    \cline{2-5}     
    Temporal & 20 & 14 & 0 & 0\\
    \cline{2-5}
    \end{tabular}
    \end{center}
    \label{tab:distances_n15}
\end{table}

\begin{table}[ht]
    \caption{Discrete homotopy distances between maximal chains $\mathcal{F}_j$ for $N=18$}
    \begin{center}
    \begin{tabular}{c|c|c|c|c|}
    \multicolumn{1}{c}{}& \multicolumn{1}{c}{Limbic ($\mathcal{F}_1$)} & \multicolumn{1}{c}{MTL ($\mathcal{F}_2$)} & \multicolumn{1}{c}{Posterior ($\mathcal{F}_3$)} & \multicolumn{1}{c}{Temporal ($\mathcal{F}_4$)}\\
    \cline{2-5}
    Limbic & 0 & 22 & 36 & 36\\
    \cline{2-5}
    MTL & 22 & 0 & 24 & 24 \\
    \cline{2-5}     
    Posterior & 36 & 24 & 0 & 0\\   
    \cline{2-5}     
    Temporal & 36 & 24 & 0 & 0\\
    \cline{2-5}
    \end{tabular}
    \end{center}
    \label{tab:distances_n18}
\end{table}

\begin{remark}
Moving downward in $\mathcal{H}\left(N\right)$ is reflective of a edge filtration that favours branching out towards unique neighbours, at the expense of loop formation, whereas an upward trend reflects the promotion of loops.  Accordingly, the top homotopy polynomial path through $\mathcal{H}(N)$ is produced by an edge filtration which exclusively avoids loops until no longer possible while the bottom-most homotopy path is reflective of maximal loop promotion in the edge filtration until no other option remains.
\end{remark}

\subsection{Simulated AD subtype differences are not a random network effect}

 \begin{figure}[t]
     \centering
     \begin{subfigure}[b]{0.32\textwidth}
         \centering
         \includegraphics[width=\textwidth,clip]{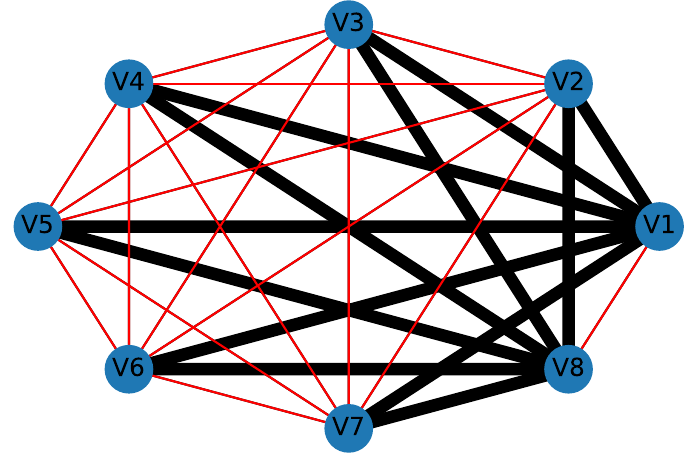}
         \caption{$\tilde{G}_8$}
     \end{subfigure}
     \hfill
     \begin{subfigure}[b]{0.32\textwidth}
         \centering
         \includegraphics[width=\textwidth,clip]{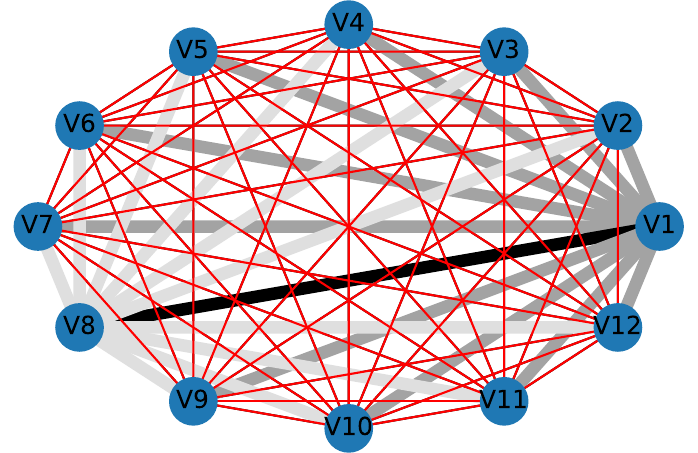}
         \caption{$\tilde{G}_{12}$}
     \end{subfigure}
     \hfill
     \centering
     \begin{subfigure}[b]{0.32\textwidth}
         \centering
         \includegraphics[width=\textwidth,clip]{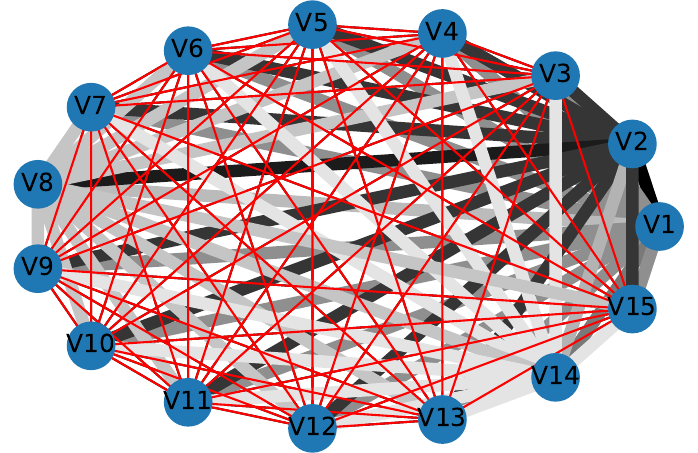}
         \caption{$\tilde{G}_{15}$}
     \end{subfigure}
\hfill\\
        \caption{Homotopy distance similarity graphs, $\tilde{G}_N$, for $N=8$ (left), $N=12$ (middle) and $N=15$ (right).  The thin red edges correspond to a weight of $M_{\infty}$ and indicate that the vertices (maximal chains in $\mathcal{H}(N)$) are identical.  Monochrome edges are accentuated with additional thickness and coloured by relative similarity weight, from light (minimum, non-zero) to dark (maximum, non-infinite).  For all graphs, the minimum non-zero edge similarity weight is $1.0$.  The maximum for $N=8$ (left) is $1.0$, for $N=12$ (middle) is $2.8$ and for $N=15$ (right) is $3.167$.  Vertices $V_1$, $V_2$, $V_3$ and $V_4$ correspond to the maximal chains $\mathcal{F}_1$, $\mathcal{F}_2$, $\mathcal{F}_3$ and $\mathcal{F}_4$ generated by the simulated AD subtype seeding at vertices $v_1$, $v_2$, $v_3$ and $v_4$, respectively (c.f.~Table~\ref{tab:vogel-subtypes} and Figs.~\ref{fig:chain-results-limbic-mtl}-\ref{fig:chain-results-post-temp}).} 
        \label{fig:similarity-graphs}
\end{figure}

In a previous study \cite{goriely2021}, we found that differences in parcellation and tractography, the two primary methods used to assemble weighted structural brain network graphs, can alter the results of simulations using models such as \eqref{eqn:fkpp-edge-degrad}-\eqref{eqn:fkpp-edge-degrad-ic}. 
Thus, a good degree of variation in simulated results can exist due to differences in structural brain networks, raising the question of whether or not the categorical differentiation between simulated AD subtypes, using $\mathcal{H}(N)$ as discussed in Sec.~\ref{subsec:HN-differentiated-subtypes}, could be reproduced by random chance and at a reasonable frequency.\\ 

\textbf{AD subtype distance similarity graphs} The space $\mathcal{H}(N)$ expresses three visually distinct maximal chains for simulated AD subtypes as $N$ increases from $N=4$ to $N=18$ (Sec.~\ref{subsec:HN-differentiated-subtypes}).  To more objectively approach the question of similarity between maximal chains we consider community detection on \textit{similarity graphs} derived from computed homotopy distances.

Given a graph $G_N$, for a fixed number of nodes $N$, the model \eqref{eqn:fkpp-edge-degrad}-\eqref{eqn:fkpp-edge-degrad-ic} generates an edge filtration and maximal chain, through $\mathcal{H}(N)$, for each vertex $v_1$, $v_2$, $\dots$, $v_N \in G_N$.  Let $\mathcal{F}_i$ denote the maximal chain corresponding to a particular vertex $v_i$ in $G_N$.  The discrete homotopy metric $d_H$ (Defn.~\ref{defn:sec:chain-compare:homotopy-metric}) yields a symmetric, weighted adjacency matrix, $A^{d_H}$, with $A^{d_H}_{ij} = d_H(\mathcal{F}_i, \mathcal{F}_j)$ from which we can define an auxiliary, undirected, weighted \textit{similarity graph} $\tilde{G}_N$.  

In the graph $\tilde{G}_N = (\tilde{V}, \tilde{E})$, a vertex $\tilde{V}_i \in \tilde{V}$ corresponds to the maximal chain $\mathcal{F}_i$ generated by the seeding site $v_i$.  Two vertices are connected by an edge $\tilde{E}_{ij} \in \tilde{E}$ if $A^{d_H}_{ij} > 0$ (different maximal chains) or if $A^{d_H}_{ij} = 0$ with $i \neq j$ (identical maximal chains generated by two different seeding locations).  The weight associated with edge $\tilde{E}_{ij}$ is
\[
    \tilde{w}_{ij} = \left\{ \begin{array}{cc} M/A^{d_H}_{ij} & A^{d_H}_{ij} > 0, \\ 
    M_{\infty} & A^{d_H} = 0 \text{ and } i \neq j, \end{array}\right.
\]
where $M = \max_{ij} A^{d_H}_{ij}$ and $M_{\infty}$ is a positive integer indicator of infinity, $M_{\infty} \gg M / \hat{M}$ and $\hat{M}$ is the minimal non-zero entry of $A^{d_H}$.   An edge weight $\tilde{w}_{ij} = M_{\infty}$ implies that $V_i$ and $V_j$ are identical ($d_H\left(\mathcal{F}_i,\mathcal{F}_j\right)=0$); the minimum non-zero edge weight in $\tilde{G}_N$ is always $\tilde{E}_{ij}^{\text{min}}=1.0$, the maximum edge weight, less than $M_{\infty}$, is $\tilde{E}_{ij}^{\text{max}}=M/\hat{M}$ and larger similarity weights indicate that $V_{ij}$ are closer together in $\mathcal{H}(N)$. 
The similarity graphs $\tilde{G}_N$ for $N\in\left\{8,12,15\right\}$ are shown in Fig.~\ref{fig:similarity-graphs}.  The graphs for $N=4$ and $N=6$ do not show any differences (Figs.~\ref{fig:chain-results-limbic-mtl:4:limbic}-\ref{fig:chain-results-limbic-mtl:6:mtl}) between the simulated AD subtype seeding vertices of interest (Table~\ref{tab:vogel-subtypes}) and the coincident paths for $N=18$ agrees with that of $N=15$ for these vertices.\\

\begin{table}[ht]
    \caption{Louvain communities in AD subtype similarity networks (Fig.~\ref{fig:similarity-graphs})}
    \begin{center}
    \begin{tabular}{|c|c|c|}
    \hline
    Graph & Communities & Modularity ($Q$)\\
    \hline
    $\tilde{G}_{8}$ & $\left\{V_1, V_8\right\}$, $\left\{\text{All other vertices}\right\}$ & $1.154 \times 10^{-1}$\\
    $\tilde{G}_{12}$& $\left\{V_1, V_8\right\}$, $\left\{\text{All other vertices}\right\}$ & $4.404 \times 10^{-4}$ \\
    $\tilde{G}_{15}$& $\left\{V_1, V_2, V_8, V_{14}\right\}$, $\left\{\text{All other vertices}\right\}$ & $1.435 \times 10^{-3}$\\
    \hline
    \end{tabular}
    \end{center}
    \label{tab:commnty-detection}
\end{table}

\textbf{Community detection in AD similarity graphs} We used the CDLIB implementation \cite{rossetti2019} of the well known, and widely used, Louvain method \cite{alotaibi2022,blondel2008,connor2017,huang2021} to investigate the community structure of the graphs $\tilde{G}_N$.  Community detection provides an unsupervised approach for  vertex similarity and the Louvain algorithm detects community structure in a weighted, undirected graph $G=(V,E)$ by maximising a graph modularity \cite{girvan2002,newman2006} score, $Q\in [-1,1]$.  Heuristically, given a choice of disjoint vertex partition $V = \coprod V_k$, a high modularity score indicates that there is more significant (weighted) within-partition connectivity, i.e.~within each vertex module $V_k = \left\{v_{k_1},v_{k_2},\dots,v_{k_l}\right\}$, compared to connectivity between the vertices of $V_i$ and $V_j$ where $i\neq j$.  The results of the Louvain community detection are shown in Table~\ref{tab:commnty-detection} for $M_{\infty} = 100 (M/\hat{M})$.  Comparing Figs.~\ref{fig:chain-results-limbic-mtl}-\ref{fig:chain-results-post-temp} to Table~\ref{tab:commnty-detection}, we once more see the patterns discussed in Sec.~\ref{subsec:HN-differentiated-subtypes}.  First, for $\tilde{G}_8$, $V_1$ (corresponding to the chain $\mathcal{F}_1$ in Fig.~\ref{fig:chain-results-limbic-mtl:8:limbic}) is differentiated from $V_2$, $V_3$ and $V_4$ ($\mathcal{F}_2$, $\mathcal{F}_3$ and $\mathcal{F}_4$ in Figs.~\ref{fig:chain-results-limbic-mtl:8:mtl}, \ref{fig:chain-results-post-temp:8:post} and \ref{fig:chain-results-post-temp:8:temp}).  The same observation holds for $\tilde{G}_{12}$ and is in line with the visual inspection of Figs.~\ref{fig:chain-results-limbic-mtl:12:limbic}, \ref{fig:chain-results-limbic-mtl:12:mtl}, \ref{fig:chain-results-post-temp:12:post} and \ref{fig:chain-results-post-temp:12:temp}.  Finally, in $\tilde{G}_{15}$, $V_2$ ($\mathcal{F}_2$,  Fig.~\ref{fig:chain-results-limbic-mtl:15:mtl}) is recognised as more similar to $V_1$ ($\mathcal{F}_1$,  Fig.~\ref{fig:chain-results-limbic-mtl:15:limbic}), than to $V_3$ and $V_4$ ($\mathcal{F}_3$ and $\mathcal{F}_4$ of Figs.~\ref{fig:chain-results-post-temp:15:post} and \ref{fig:chain-results-post-temp:15:temp}, respectively),  reinforcing the observation of Sec.~\ref{subsec:HN-differentiated-subtypes}.\\

\textbf{Comparing community structure to null models} To check that the results of the Louvain algorithm, applied to the homotopy distance similarity graphs (Fig.~\ref{fig:similarity-graphs}), were not likely to be due to noise in the weights or attributable to statistically random network effects we considered three null model approaches used in the recent study of biological networks \cite{connor2017}.  We first tested whether our results may be due to the aggregated contribution of smaller, possibly spurious, similarity weights in the homotopy distance similarity matrices, using the random permutation approach of \cite{connor2017}.  Starting with the original weighted graphs $G_8$, $G_{12}$ and $G_{15}$ (Fig.~\ref{fig:connectome-33-subgraphs}), ten null model graphs, $G_N^{(k)}$ for $k=1,2,\dots,10$, were created by randomly permuting the edge weights for each fixed $N\in\left\{8,12,15\right\}$.  For each $G_N^{(k)}$, we computed the maximal chains corresponding to each vertex, the pursuant distance and similarity matrices and, finally, the similarity graph (e.g.~of the type shown in Fig.~\ref{fig:similarity-graphs}); all null similarity graphs use $M^{(k)}_{\infty} = 100(M^{(k)}/\hat{M}^{(k)})$, as done previously.  Considering similarity edges weights $w_{ij} < M_{\infty}^{(k)}$, we followed the thresholding procedure described in \cite{connor2017}.  Due to our small vertex count, compared to \cite{connor2017}, we recorded the mean threshold, over all null models $G^{(k)}_N$ with $k\in\left\{1,2,\dots,10\right\}$ and $N$ fixed, stabilising the giant component such that there were no isolated vertices, which would be implausible unless all maximal chains, $\mathcal{F}_i$, coincided exactly.  The final thresholds $\tilde{T}_N$, determined for each $\tilde{G}_N$ (c.f.~Fig.~\ref{fig:similarity-graphs}), were: $\tilde{T}_8 = 0.739$, $\tilde{T}_{12} = 2.087$ and $\tilde{T}_{15} = 2.072$.  All edges with weights less than the corresponding thresholds were removed from the similarity graphs (c.f.~Fig.~\ref{fig:similarity-graphs}) and the Louvain community detection was re-run; the results were identical to those of Table~\ref{tab:commnty-detection}.

We tested the observed modularity of the homotopy distance similarity graphs (c.f.~Table~\ref{tab:commnty-detection} and Fig.~\ref{fig:similarity-graphs}) against two well known network null models often used to gauge the influence of random structure.

\section{Conclusion}\label{sec:conclusion}
%


In this manuscript, we have introduced the homotopy poset $\mathcal{H}\left(N\right)$, the quotient of the poset of spanning subgraphs of a (complete) graph, $G=(V,E)$, on $N=|V|$ vertices with respect to the equivalence relation of graph homotopy.  Equivalence classes in $\mathcal{H}(N)$ have a unique representation in terms of their graph homotopy polynomial and a graph's homotopy polynomial encodes the structure of its homology groups (cf.~Defn.~\ref{defn:homotopy-poly} and Exmp.~\ref{exmp:topo-comput}). In addition, we have detailed the first computational algorithms (Sec.~\ref{sec:pe-quotient-construction}) to construct $\mathcal{H}(N)$, for $N$ arbitrary, homotopy polynomials and their successors; the algorithms are guaranteed to terminate. 

We have illustrated the potential for the application of our results in AD research.  We demonstrated that simulated neurodegenerative edge filtrations show topological differences based on AD subtype, quantifiable by $d_H$, an observation that raises the question of potential applicability to subtype determination based on patient neuroimaging data.  In terms of applications, this approach is by no means exclusive to AD research and is applicable to more general contexts of network dynamical systems; examples include the study of graph percolation, information propagation in social networks and population dynamics in ecology.  Further, graphs are inherently combinatorial objects, and can be viewed as the 1-skeletons of simplicial or CW-complexes. Combinatorial results derived for simplicial and CW-complexes can similarly be applied to graphs as a special case, for instance homology and cohomology \cite{jonsson2007simplicial}.  These combinatorial and topological graph theories can similarly be applied to sequences of subgraphs, as we do with homotopy equivalence; the resulting sequence of groups may distinguish different conditions. As algorithms and their implementation continue to improve in computational topology, so does tailoring topological theory to solve concrete real-world problems.

A clear limitation of our work is that, presently, more advanced computational algorithms are needed in order to more efficiently construct $\mathcal{H}(N)$ and work with the paths corresponding to edge filtrations in order to make this process tractable for large values of $N$.  This drawback does limit the current feasibility of using our approach for large, real-world applications. To address this issue, herein, we limited our investigation to modest subgraphs  (e.g.~Fig.~\ref{fig:connectome-33-subgraphs}) of the left hemisphere of a coarse, but standard, parcellation ($N=83$) of a whole-brain brain structural connectome (Fig.~\ref{fig:connectome-33-subgraphs}).  This simplification allowed us to use standard software, such as Mathematica, for the required model simulation (c.f.~\eqref{eqn:fkpp-edge-degrad}-\eqref{eqn:fkpp-edge-degrad-ic}) and construction of $\mathcal{H}(N)$ (c.f.~Algorithm~\ref{alg:algorithm-1} and Algorithm \ref{alg:algorithm-2}, Sec.~\ref{sec:pe-quotient-construction}).  To compute similar results on the full $N=83$ connectome, or more refined whole brain structual connectome parcellations with vertex counts in the hundreds or thousands \cite{goriely2021}, will require a more thoughtful computational strategy.  As a point of future work, we hope to improve the computational algorithms presented and to compare the simulated results, on larger brain structural connectome networks, to patient neuroimaging data.

\section{Acknowledgements}\label{sec:ack}
The authors thank Daniele Celoria for fruitful discussions and helpful comments on this manuscript. AG is grateful for the support by the Engineering and Physical Sciences Research Council of Great Britain under research grants EP/R020205/1. HAH gratefully acknowledges EPSRC EP/R005125/1 and EP/T001968/1, the Royal Society RGF$\backslash$EA$\backslash$201074 and UF150238, and Emerson Collective. CG gratefully acknowledges the support by NIH fellowship grant 1F32HL162423-01. DB and HAH are members of the Centre for Topological Data Analysis, funded in part by EPSRC EP/R018472/1. Finally, the authors thank one of the reviewers for the suggestion of considering quotienting by homology. For the purpose of Open Access, the authors have applied a CC BY public copyright licence to any  Author Accepted Manuscript (AAM) version arising from this submission. 

\bibliographystyle{abbrv}

\end{document}